\documentclass[acmsmall,screen,nonacm]{acmart}

\AtBeginDocument{%
  }

\usepackage{microtype}
\usepackage{booktabs}
\usepackage{longtable}
\usepackage{subcaption}
\usepackage{empheq}
\usepackage{balance}
\usepackage{pifont}
\usepackage{wrapfig}
\usepackage{hyperref}
\usepackage{fancyvrb}
\usepackage{xspace}
\usepackage{algorithmicx}
\usepackage[noend]{algpseudocode}
\usepackage{algorithm}
\algrenewcommand{\algorithmiccomment}[1]{\hfill{\color{gray}\(\triangleright\) #1}}
\usepackage{amsfonts}
\usepackage{amsmath}
\usepackage{amsthm}
\usepackage{semantic}
\usepackage{listings}
\usepackage{array}
\usepackage{fancyvrb}
\usepackage{mathpartir}
\usepackage{stmaryrd}
\usepackage{graphicx}
\usepackage{caption}
\usepackage{syntax, etoolbox}
\usepackage{tikz}
\usepackage{proof}
\usepackage{color}
\usepackage{cleveref}
\usepackage{mathtools}
\usepackage{adjustbox}
\usepackage{bbding}
\usepackage{soul}
\usepackage{multirow}
\usepackage{placeins}
\usepackage[abbreviations]{glossaries-extra}
\makeglossaries
\usepackage{bcprules}
\usepackage{caption}
\usepackage{wasysym}

\definecolor{ForestGreen}{RGB}{34,139,34}

\usepackage[mathscr]{euscript}

\usepackage{enumitem}

\usepackage{comment}
\usepackage{wrapfig}

\theoremstyle{definition}
\newtheorem{definition}{Definition}[section]

\DeclareMathOperator*{\argmin}{arg\,min}

\usepackage[utf8x]{inputenc}
\lstdefinelanguage{eggccir}{
  keywords={fn, let, if, with, then, else, do, while, print,store,load,alloc,ctx,dowhile},
  keywordstyle=\color{blue}\bfseries,
  morekeywords={[2]},
  keywordstyle={[2]\color{teal}},
  morekeywords={[3]InLoop, InIf,
  },
  keywordstyle={[3]\color{purple}},
  morecomment=[l]{//},
  morecomment=[s]{/*}{*/},
  commentstyle=\color{gray},
  morestring=[b]{"},
  stringstyle=\color{red},
  sensitive=true,
}

\lstdefinelanguage{Rust}{
  keywords={print, do, fn, let, mut, match, if, else, while, for, loop, in, break, continue, return, struct, enum, impl, trait, pub, mod, use, crate, super, self, Self, ref, move, as, const, static, type, where, async, await, dyn, unsafe, extern},
  keywordstyle=\color{blue}\bfseries,
  morekeywords={[2]String, Vec, Option, Result, Some, None, Ok, Err, Box, Rc, Arc, Cell, RefCell, HashMap, HashSet},
  keywordstyle={[2]\color{teal}},
  morekeywords={[3]println, format, dbg},
  keywordstyle={[3]\color{orange}},
  morecomment=[l]{//},
  morecomment=[s]{/*}{*/},
  commentstyle=\color{gray},
  morestring=[b]{"},
  stringstyle=\color{red},
  sensitive=true
}

\lstdefinelanguage{egglog}{
  sensitive,
  morecomment=[l]{;},
  moredelim=[s][{\color[rgb]{0,0,0.75}}]{\#[}{]},
  morestring=[b]{"},
  alsodigit={},
  alsoother={},
  alsoletter={!},
  alsoletter={:},
  alsoletter={-},
  otherkeywords={=>},
  otherkeywords={:-},
  otherkeywords={|},
  otherkeywords={.},
  morekeywords={assert, function, relation, datatype, rewrite,
                rule, union, merge, default, when,
                check, set, run, sort, define, extract,
                old, new},
  otherkeywords={:merge, :default, :when}
}%

\lstdefinestyle{eggcc}{
  basicstyle=\small\ttfamily,
}
\lstdefinestyle{eggccterm}{
  style=eggcc,
  numbers=none,
  xleftmargin=0pt,
  mathescape=false,
  aboveskip=0pt,
  belowskip=0pt,
  columns=fullflexible,
  showstringspaces=false
}

\lstnewenvironment{eggccterm}[1][]{\lstset{style=eggccterm,#1}}{}

\newcommand{\rvsdg}{RVSDG\xspace}
\newcommand{\rvsdgs}{RVSDGs\xspace}

\newcommand{\ir}{RVSDG$_{\text{EQ}}$\xspace}

\newcommand{\gurobi}{Gurobi\xspace}
\newcommand{\cbc}{CBC\xspace}
\newcommand{\eggcc}{\textsc{eggcc}\xspace}
\newcommand{\eggccT}{EGGCC}
\newcommand{\egraph}{e-graph\xspace}
\newcommand{\eclass}{e-class\xspace}
\newcommand{\eclasses}{e-classes\xspace}

\newcommand{\Egraph}{E-graph\xspace}
\newcommand{\egraphs}{e-graphs\xspace}
\newcommand{\Egraphs}{E-graphs\xspace}
\newcommand{\classof}[2]{\ensuremath{\mathbb{C}_{#2}({#1})}}
\newcommand{\termsc}[2]{\ensuremath{\mathbb{T}_{#2}({#1})}}

\newcommand{\bril}{Bril\xspace}
\newcommand{\polybench}{PolyBench\xspace}
\newcommand{\peggy}{Peggy\xspace}
\newcommand{\llvm}{LLVM\xspace}
\newcommand{\eqsat}{equality saturation\xspace}
\newcommand{\Eqsat}{Equality saturation\xspace}

\newcommand{\tighten}{\looseness=-1}

\newcommand{\todo}[1]{{\color{red}[TODO: #1]}}
\newcommand{\rt}[1]{{\color{red}[RT: #1]}}
\newcommand{\statepath}{state path\xspace}
\newcommand{\effectunsafe}{effect-unsafe\xspace}

\makeatletter
\DeclareRobustCommand{\effectsafe}{%
  \@ifnextchar( {\effectsafe@paren}{\effectsafe@plain}%
}
\newcommand{\effectsafe@plain}{effect-safe\xspace}
\def\effectsafe@paren(#1){%
  \ensuremath{\mathsf{EffSafe}(#1)}%
}
\makeatother
\newcommand{\Effectsafe}{Effect-safe\xspace}
\newcommand{\effectsafety}{effect-safety\xspace}
\newcommand{\effectsafeext}{effect-safe extraction\xspace}
\newcommand{\EffectSafeExt}{Effect-Safe Extraction\xspace}
\newcommand{\opteffectsafeext}{optimal effect-safe extraction\xspace}
\newcommand{\statewalk}{statewalk\xspace}

\newcommand{\statewalkwidth}{statewalk width\xspace}
\newcommand{\statewalks}{statewalks\xspace}

\newcommand{\extractionset}{extractable set\xspace}
\newcommand{\extractionsets}{extractable sets\xspace}

\newcommand{\effectful}{effectful\xspace}

\newcommand{\ES}{\mathsf{es}}

\newcommand{\statewalkequiv}{\mathrel{\equiv_{\mathsf{SW}}}}

\newcommand{\sw}{\ensuremath{\boldsymbol{\omega}}}

\newcommand{\eclasssws}[1]{\ensuremath{\mathsf{SW}(#1)}}
\newcommand{\children}[1]{\ensuremath{\mathsf{Children}(#1)}}

\newcommand{\subterms}[1]{\ensuremath{\mathsf{Subterms}(#1)}}

\newcommand{\swwidth}[1]{\ensuremath{\operatorname{sww}(#1)}}

\newcommand{\tiger}{\textsc{statewalk DP}\xspace}

\newcommand{\egglog}{\textsc{Egglog}\xspace}

\newcommand{\M}[1]{\ensuremath{#1}}
\newcommand{\N}{\M{\mathbb{N}}}

\newcommand{\uop}[2]{\texttt{#1}\ensuremath{({#2})}}
\newcommand{\bop}[3]{\texttt{#1}\ensuremath{(#2,~ #3)}}
\newcommand{\trop}[4]{\texttt{#1}\ensuremath{(#2,~ #3,~ #4)}}

\newcommand{\Arg}{\ensuremath{\text{Arg}}}
\newcommand{\argt}{\texttt{Arg}\xspace}

\newcommand{\kwdload}{\text{\lstinline|load|}}
\newcommand{\kwdstore}{\text{\lstinline|store|}}

\newcommand{\val}{\ensuremath{v}}

\newcommand{\stateval}{\ensuremath{\sigma}}
\newcommand{\statevals}{\ensuremath{\Sigma}}

\newcommand{\argval}{a}
\newcommand{\baseval}{\ensuremath{\phi}}

\newabbreviation{cfg}{CFG}{control-flow graph}
\newabbreviation{peg}{PEG}{program expression graph}
\newabbreviation{rvsdg}{RVSDG}{Regionalized Value-State Dependence Graph}
\newabbreviation{ir}{IR}{intermediate representation}
\newabbreviation{ssa}{SSA}{static single-assignment form}
\newabbreviation{dag}{DAG}{directed acyclic graph}
\newabbreviation{stg}{STG}{Spineless Tagless G-Machine}
\newabbreviation{ghc}{GHC}{Glasgow Haskell Compiler}

\newcommand{\C}[1]{\texttt{#1}\xspace}
\newcommand{\F}[1]{\scriptstyle{\lstset{basicstyle=\scriptsize\ttfamily}#1}\xspace}

\DeclareMathSymbol{\mhyphen}{\mathord}{AMSa}{"39}

\newcommand{\bigstep}[3]{\ensuremath{#1 \vdash #2 \Downarrow #3}}

\newcommand{\myparagraph}[1]{\textbf{#1.}\ \;}

\newcommand{\overalleggcctigerILPNOMINOzeroOzeroEggccExtractionTime}{30.589\xspace}

\newcommand{\overalleggcctigerILPOzeroOzeroEggccExtractionTime}{26.211\xspace}

\newcommand{\overalleggcctigerOzeroOzeroEggccNonExtractionTime}{7.925\xspace}

\newcommand{\overalleggcctigerOzeroOzeroEggccExtractionTime}{1.536\xspace}

\newcommand{\overallllvmOthreeOzeroLLVMCompileTime}{0.834\xspace}

\newcommand{\polybenchOveralleggcctigerOzeroOzeroNormalizedMean}{0.571\xspace}

\newcommand{\polybenchOverallllvmOthreeOzeroNormalizedMean}{0.326\xspace}

\newcommand{\polybenchOverallllvmOthreeOthreeNormalizedMean}{0.312\xspace}

\newcommand{\brilOveralleggcctigerOzeroOzeroNormalizedMean}{0.876\xspace}

\newcommand{\brilOverallllvmOthreeOzeroNormalizedMean}{0.695\xspace}

\newcommand{\brilOverallllvmOthreeOthreeNormalizedMean}{0.542\xspace}

\newcommand{\raytraceOveralleggcctigerOzeroOzeroEggccNonExtractionTime}{171.605\xspace}

\newcommand{\raytraceOveralleggcctigerOzeroOzeroEggccExtractionTime}{22.603\xspace}

\newcommand{\AvgPolybenchRegionalizedEgraphsPerBenchmark}{209\xspace}
\newcommand{\NumpolybenchBenchmarks}{30\xspace}

\newcommand{\NumbrilBenchmarks}{64\xspace}

\newcommand{\NumBenchmarksAllSuites}{96\xspace}

\newcommand{\NumeggcctigerILPGurobiRegionTimeoutBenchmarksBril}{17\xspace}
\newcommand{\NumeggcctigerILPNOMINRegionTimeoutBenchmarks}{4\xspace}

\newcommand{\AvgEggcctigerOZeroOZeroExtractionTimeSecsOnILPGurobiSolvedBenchmarks}{0.360\xspace}

\newcommand{\NumRaytraceRegionalizedEgraphs}{724\xspace}

\newcommand{\MaxRaytraceTigerExtractionTimeSecs}{0.018141\xspace}
\newcommand{\NumRaytraceILPRegionalizedEgraphTimeouts}{35\xspace}
\newcommand{\NumRegionalizedEgraphs}{9610\xspace}
\newcommand{\NumILPGurobiRegionTimeouts}{234\xspace}

\newcommand{\AvgILPGurobiRegionExtractTimeSecs}{0.965\xspace}
\newcommand{\AvgTigerLiveOnSatelliteOnRegionExtractTimeSecs}{0.000426\xspace}
\newcommand{\MaxTigerLiveOnSatelliteOnRegionExtractTimeSecs}{0.0619\xspace}
\newcommand{\MeanStatewalkWidthAllRegions}{33.15\xspace}
\newcommand{\MaxStatewalkWidthAllBenchmarks}{106,259\xspace}

\newcommand{\PercentRegionsStatewalkWidthUnderTwo}{94.7\%\xspace}

\newcommand{\PercentRegionsStatewalkWidthUnderSix}{98.4\%\xspace}

\newcommand{\StatewalkWidthHeatThreeDLiveOnSatelliteOn}{24\xspace}
\newcommand{\GeometricMeanTigerSpeedupVsGurobiWithTimeouts}{520\xspace}
\newcommand{\GeometricMeanTigerSpeedupVsGurobiWithTimeoutsX}{520$\times$\xspace}
\newcommand{\GeometricMeanILPEncodingVarsPerEgraphSize}{4.55\xspace}
\newcommand{\MaxILPEncodingVarsPerEgraphSize}{17.63\xspace}

\newcommand{\FenwickMeanCycleSpeedupEggcctigerWITHCTXOZeroOZeroVsllvmOThreeOThree}{6.52\xspace}
\newcommand{\FenwickMeanCycleSpeedupEggcctigerWITHCTXOZeroOZeroVsllvmOThreeOZero}{23.01\xspace}

\newcommand{\ArticlePages}{28}
\ifdefined\ArticleOnly
  
\fi
\AtEndDocument{%
  \ifnum\getpagerefnumber{mainend}>0
    \ifnum\getpagerefnumber{mainend}=\ArticlePages\relax\else
      \ClassWarningNoLine{acmart}{Article body ends on page
        \getpagerefnumber{mainend}, but \string\ArticlePages\space says
        \ArticlePages. Update \string\ArticlePages\space in main.tex to match,
        and split the PDF after that page}%
    \fi
  \fi
}

\begin{document}

\title{Efficient Extraction for Effectful E-graphs}



\author{Oliver Flatt}
\orcid{0000-0002-0656-235X}
\affiliation{%
  \institution{University of Washington}
  \city{Seattle}
  \state{Washington}
  \country{USA}
}
\email{oflatt@cs.washington.edu}

\author{Anjali Pal}
\orcid{0009-0006-0692-0707}
\affiliation{%
  \institution{University of Washington}
  \city{Seattle}
  \state{Washington}
  \country{USA}
}
\email{anjalip@cs.washington.edu}

\author{Yihong Zhang}
\orcid{0009-0006-5928-4396}
\affiliation{%
  \institution{University of Washington}
  \city{Seattle}
  \state{Washington}
  \country{USA}
}
\email{yz489@cs.washington.edu}

\author{Ryan Tjoa}
\orcid{0009-0003-0731-5398}
\affiliation{%
  \institution{University of Washington}
  \city{Seattle}
  \state{Washington}
  \country{USA}
}
\email{rtjoa@cs.washington.edu}

\author{Kirsten Graham}
\orcid{0009-0002-7245-184X}
\affiliation{%
  \institution{University of Washington}
  \city{Seattle}
  \state{Washington}
  \country{USA}
}
\email{kmgraham@cs.washington.edu}

\author{Alex Fischman}
\orcid{0009-0000-0112-6711}
\affiliation{%
  \institution{University of Washington}
  \city{Seattle}
  \state{Washington}
  \country{USA}
}
\email{adfisch@cs.washington.edu}

\author{Chandrakana Nandi}
\orcid{0000-0001-8633-8413}
\affiliation{%
  \institution{Certora Inc.}
  \city{Seattle}
  \state{Washington}
  \country{USA}
}
\affiliation{%
  \institution{University of Washington}
  \city{Seattle}
  \state{Washington}
  \country{USA}
}
\email{cnandi@cs.washington.edu}

\author{Eli Rosenthal}
\orcid{0009-0008-9386-1614}
\affiliation{%
  \institution{Google}
  \city{Oakland}
  \state{California}
  \country{USA}
}
\email{ezrosenthal@gmail.com}

\author{Zachary Tatlock}
\orcid{0000-0002-4731-0124}
\affiliation{%
  \institution{University of Washington}
  \city{Seattle}
  \state{Washington}
  \country{USA}
}
\email{ztatlock@cs.washington.edu}

\author{Haobin Ni}
\orcid{0000-0002-7718-7905}
\affiliation{%
  \institution{University of Washington}
  \city{Seattle}
  \state{Washington}
  \country{USA}
}
\email{hn42@cs.washington.edu}

\renewcommand{\shortauthors}{Flatt et al.}

\begin{abstract}
  \Egraphs have enabled recent advances in
  program optimization, synthesis, and verification,
  yet remain difficult to apply to effectful programs
  whose memory and I/O operations must respect execution order.
Existing effect-aware extraction algorithms
  rely on integer linear programming (ILP)
  and dominate total runtime.\tighten

We introduce \textbf{\tiger},
  a new extraction algorithm that
  enforces effect ordering efficiently without external solvers.
We prove that finding \textit{any} effect-safe extraction is NP-complete,
  but show that \tiger is tractable in \textit{\statewalkwidth},
  a parameter that measures the complexity of
  dataflow interactions among effects.
In practice, \statewalkwidth generally remains small,
  enabling \tiger to achieve
  order-of-magnitude speedups over ILP extraction
  while producing programs comparable to LLVM
  across our benchmarks.
We implement the algorithm in \eggcc,
  a prototype \egraph-based compiler for imperative Bril programs,
  and demonstrate that effect-aware extraction
  is no longer a bottleneck.\tighten


\end{abstract}



\maketitle

\AtBeginEnvironment{grammar}{\small}
\section{Introduction}
\label{sec:intro}

\Eqsat and \egraphs~\cite{denali, nelson, egg, egglog}
  power new state-of-the-art program optimizers, synthesizers, and verifiers
  across a variety of domains~\cite{
    chassis, peggy, szalinski, spores, tensat, felix, mlir-egg, aegraph, caviar,
    reducing-overparameterization, optimizing-tensor,
    recto2023compilerarrayprogramsvectorized, megalibm}.
These tools follow a common workflow:
  initialize an \egraph $\mathcal{G}$ with input program $P$;
  expand $\mathcal{G}$ by
  applying semantics-preserving rules
  to build a large space of programs equivalent to $P$;
  then \textit{extract} an optimized
    version of $P$ from $\mathcal{G}$
    guided by a cost function (e.g., program size).
Decoupling the
  exploration of equivalent program fragments from
  extraction mitigates phase ordering~\cite{phaseordering},
  freeing engineers to focus on semantics-preserving optimization rules, not
  transformation order.

Most \egraph-based optimizers only operate over \textit{pure} programs
  (i.e., those without side effects),
  though as our case study shows,
  optimizing \textit{effectful} programs with \egraphs can
  yield significant speedups over LLVM (\autoref{sec:Fenwick}).
Rewriting effectful programs using \egraphs is difficult due to
  the tension between effectful optimizations and congruence.
Congruence is a fundamental property of \egraphs that requires all equivalent
  terms to be interchangeable,
  which is essential
  for efficient representation and exploration of equivalent programs.
Unfortunately, under congruence even basic rewrites
  over effectful programs
  lead to \textit{ambiguous} terms whose subterms
  require incompatible effect orders.
The only viable approach is to tolerate these ambiguous terms,
  thereby relaxing the equivalence relation the \egraph represents.
Optimization rules must still be semantics-preserving according to the relaxed
  relation,
  but are free to relate terms with equivalent effect orderings.
This relaxed equivalence relation creates the key challenge hindering
  effectful \egraphs: enforcing valid effect ordering at extraction time,
  which we call \emph{\effectsafe extraction}.



Optimal extraction from pure \egraphs is known to be
  NP-complete~\cite{stepp-thesis,yihong-np},
  but prior to this work
  there was no known complexity result characterizing \effectsafe extraction\footnotemark[2].
We show that optimal \effectsafe extraction is NP-complete,
  and surprisingly, that finding \textit{any} \effectsafe extraction
  is NP-complete as well (\autoref{sec:np}).
\autoref{tab:theory} shows the complexity of pure and \effectsafe extraction.

\begin{table}[h]
\centering
\begin{minipage}[t]{0.48\linewidth}
\centering
\caption{Complexity bounds of extraction.
  Even without considering effect ordering,
  finding the optimal extraction under a cost model
  is known to be NP-complete~\cite{stepp-thesis, yihong-np}.
  Bold indicates our contribution.}
\label{tab:theory}
\footnotesize
\begin{tabular}{l|l|l}
 & Pure & Effect-Safe \\ \hline
Any & P & \textbf{NP-complete} \\
Optimal   & NP-complete & \textbf{NP-complete} \\
\end{tabular}
\end{minipage}
\hfill
\begin{minipage}[t]{0.48\linewidth}
\centering
\caption{Practical approaches to extraction.
  Most applications of pure \egraphs use greedy best-effort extraction algorithms,
  which work well in practice~\cite{extraction-gym, herbie, smoothe, eboost}.
  Bold indicates our contribution.
}
\label{tab:practice}
\footnotesize
\begin{tabular}{l|l|l}
 & Pure & Effect-Safe \\ \hline
Best-Effort\footnotemark & Many, see \S9 & \cite{peggy}, \textbf{\tiger} \\
Optimal     & \cite{treewidth-extract} & Future work  \\
\end{tabular}
\end{minipage}
\end{table}
\footnotetext{
  Note that in the context of optimization,
  the input program is a trivial, yet unsatisfying, solution
  to the extraction problem.
  Instead of attempting to find \textit{any} extraction,
    most implementations take a greedy approach
    to produce best-effort results.
}

Of course, just because a problem is NP-complete does not mean it goes away.
In practice, researchers have developed many extraction algorithms
  for pure \egraphs, summarized in \autoref{tab:practice}.
A few systems sidestep the challenge of \effectsafe extraction:
they either give up on effectful terms and
  limit \egraphs to pure program fragments only~\cite{denali,aegraph,eqsat-julia},
  or leave soundness entirely to the user~\cite{mlir-egg}.
Both approaches fall in the pure (top-left) quadrant of \autoref{tab:practice}
  because they do not enforce effect ordering constraints during extraction.
The only sound prior approach, introduced by Peggy~\cite{peggy} (top-right),
  requires calling out
  to an integer linear programming (ILP) solver to
  enforce effect ordering constraints.
  However, Peggy's ILP encoding is incomplete, so it is a best-effort
  extraction algorithm\footnotemark[2].
Furthermore, relying on a general-purpose ILP solver
  leads to unpredictable performance and frequent failures,
  as seen in both prior work and our own reimplementation
  (\autoref{fig:intro-egraph-extraction-bottleneck}).
Even when ILP succeeds, extraction typically dominates total
  optimization time, a problem we call the
  \emph{effectful \egraph extraction bottleneck}.

\footnotetext[2]{
  Peggy~\cite{peggy} uses an ILP solver to tackle the \effectsafe
  extraction problem, but its encoding is incomplete, so it does
  not imply any complexity bound on the \effectsafe extraction problem
  despite ILP being NP-complete.
  From \citet{peggy}: ``We have developed a simple constraint solving technique for finding
  a linearization of heap operations in a PEG, but this technique is not complete (as in, even
  if there is a linearization, we are not guaranteed to find it).''
}

\begin{figure}[t]
  \centering
  \begin{subfigure}[b]{0.48\textwidth}
    \centering
    \includegraphics[width=0.9\textwidth]{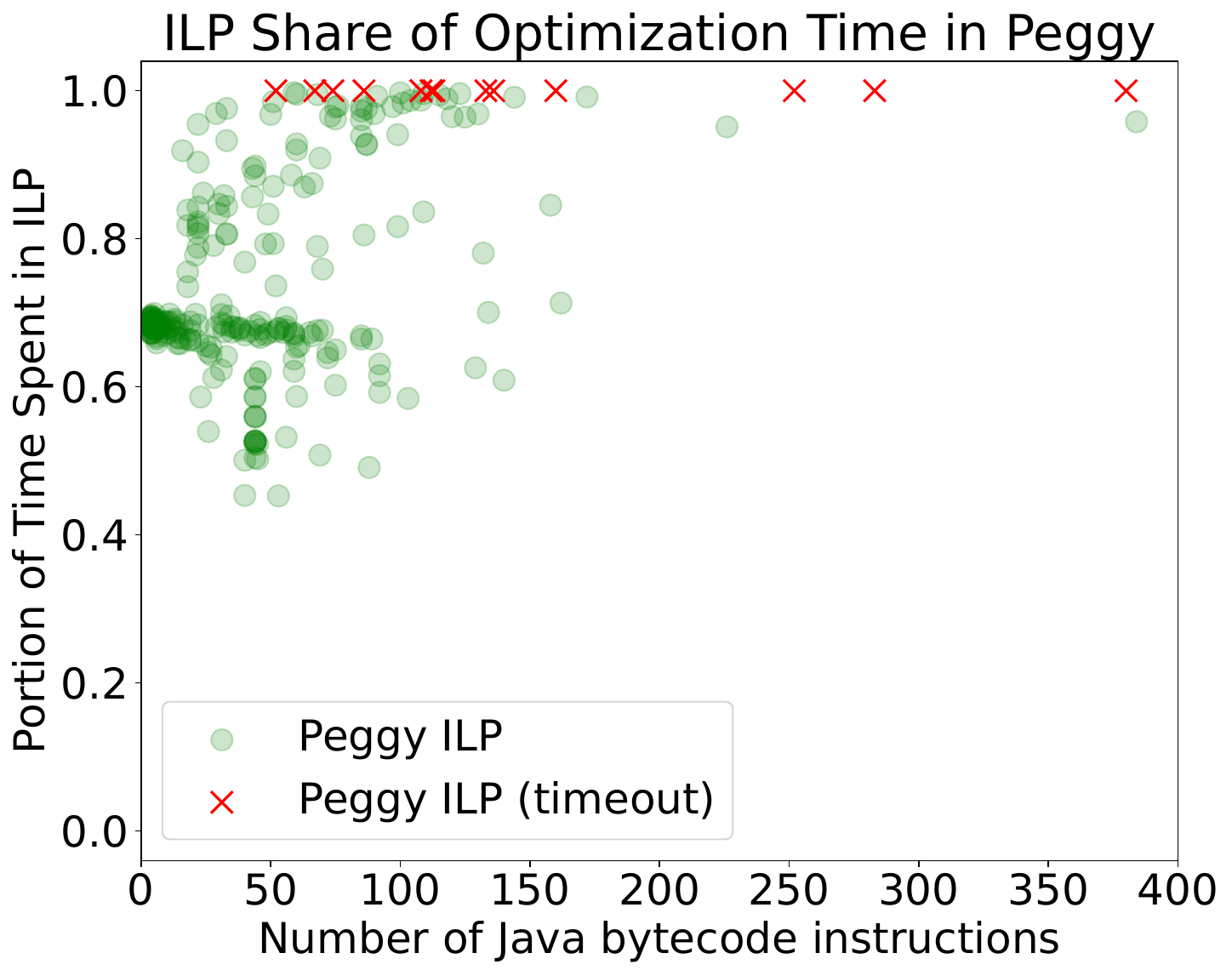}
  \end{subfigure}
  \hfill
  \begin{subfigure}[b]{0.48\textwidth}
    \centering
    \includegraphics[width=0.9\textwidth]{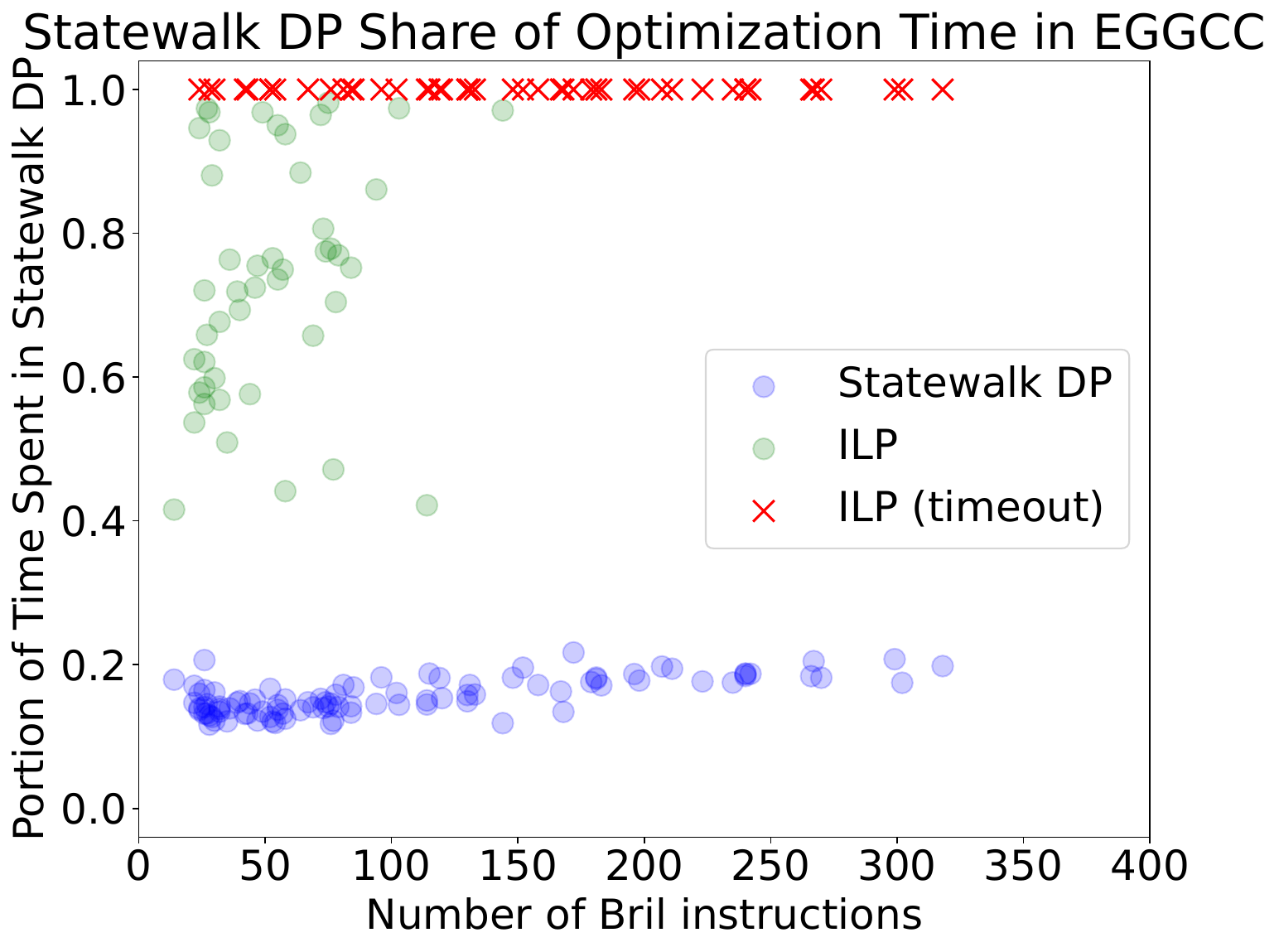}
  \end{subfigure}
  \caption{
    \textit{The effectful \egraph extraction bottleneck.}
    For ILP-based approaches (green),
      the majority of total optimization time
      is often spent on extraction,
        both in prior work~\cite{peggy} (left)
        and our \eggcc prototype (right).
    Although these graphs differ in
      source languages, benchmarks, and implementations,
      both show the bottleneck and
      the unpredictability of existing approaches.
    In contrast, \tiger (blue)
      behaves predictably and
      significantly reduces extraction time.
    \autoref{sec:eval} details the setup for both plots.
  }
  \Description{Two scatter plots. In both, the horizontal axis is program size
    and the vertical axis is the fraction of total optimization time spent in
    extraction, from 0 to 1. Left, for Peggy: green points sit in a band around
    0.6 to 0.7 with many more at 1.0, and red crosses along the top mark
    timeouts. Right, for eggcc: green ILP points again sit near 1.0 with red
    crosses for timeouts, while the blue statewalk DP points form a flat band
    near 0.1 to 0.2 across the whole range of program sizes.}
  \label{fig:intro-egraph-extraction-bottleneck}
\end{figure}

In this paper, we present the first practical effect-safe extraction algorithm,
  \tiger, which eliminates the bottleneck for best-effort approaches.
It is guaranteed to produce an effect-safe extraction while optimizing for
  cost greedily.
Our key idea is to \textit{refine}
  the relaxed equivalence used by rules
  into an effect-safe version during extraction.
\tiger partitions terms by equivalence classes over
  their \textit{\statewalks}:
  sequences of effectful operations
  along which their subterms may be reached.
\tiger then uses a cost function and
  dynamic programming to select
  a single effect-safe extraction for each partition.
We call the maximum number of
  partitions for any equivalence class
  the \egraph's \textit{\statewalkwidth}.
We show that \effectsafe extraction is
  fixed-parameter tractable (FPT)~\cite{fpt}
  in the \statewalkwidth parameter:
  \tiger takes $O(k n^2)$ time where
  $k$ is the \statewalkwidth and
  $n$ is the size of the \egraph,
  extending the popular $O(n^2)$ greedy algorithm
  to effectful \egraphs.\tighten

To evaluate \tiger,
  we developed \eggcc, a proof-of-concept optimizer.
\eggcc translates Bril programs~\cite{bril} to
  an RVSDG-based representation~\cite{rvsdg},
  performs equality saturation using \egglog~\cite{egglog},
  extracts with \tiger, and finally
  translates the optimized program back to Bril.
We found that \statewalkwidth generally remains
  small in practice (mean of \MeanStatewalkWidthAllRegions),
  enabling \tiger to achieve a \GeometricMeanTigerSpeedupVsGurobiWithTimeoutsX speedup
  over ILP-based extraction while producing
  code comparable to LLVM across our benchmarks.

\noindent
This paper makes the following contributions:
\begin{itemize}[leftmargin=2em, itemsep=0.3em]

  \item
  We identify the effectful \egraph extraction bottleneck
    and prove that both optimal extraction and \textit{any}
    extraction in the face of effects are NP-complete in general (\autoref{sec:np}).

  \item
  We present \tiger,
    a new effect-aware extraction algorithm
    that is tractable with respect to \statewalkwidth
    and produces best-effort results with respect to the cost model.
  We prove \tiger is sound and complete,
    and analyze its time complexity (\autoref{sec:tiger}).

  \item
  We implement \tiger in \eggcc,
    a proof-of-concept \egraph-based optimizer for
    imperative Bril~\cite{bril} programs (\autoref{sec:impl}).
      Through a case study, we show the value of \egraphs for optimizing
      effectful programs: adding a single domain-specific
      optimization enables significant speedups without interfering with
      existing ones (\autoref{sec:Fenwick}).

  \item
  We evaluate \tiger's speed and the programs it extracts,
    demonstrating significant speedups over prior ILP-based approaches
    while producing high-quality extractions (\autoref{sec:eval}).\tighten

\end{itemize}

\section{Background}
\label{sec:background}

We provide background on \egraphs and dataflow IRs that
  represent effectful programs.

\subsection{Formal Definition of an \Egraph}
\label{subsec:backgroundegraphs}

This section gives a formal definition for \egraphs
  which is built upon in \autoref{sec:formalizing-rvsdgs}.
We start by defining terms,
 then the congruence closure, and
 finally the \egraph as a data structure representing
 the congruence closure over a set of identities.
Note that while some recent work defines \egraphs as graphs with nodes
  and edges~\citep{egg},
our definition leverages the more classic definition
  of a congruence closure~\citep{common-subexpr}.
The two definitions are equivalent, with terms taking the
  place of e-nodes.
We use the classic
  definition because it eases formalization and proofs.

Given a set of ranked function symbols $\Sigma$, we inductively define terms of $\Sigma$ as
 $t=f(t_1, \ldots, t_n)$, where $f$ is an $n$-ary function symbol and $t_1,\ldots t_n$ are terms.
We call the finite sequence $\langle t_1, \ldots, t_n\rangle$ the \emph{children} of $t$
 and the function symbol $f$ the \emph{head} of $t$.
We denote the children of a term $t$ with $\children{t}$.
Constants, such as integers, are represented as nullary function symbols
  and are written as $f$ instead of $f()$.

We use partial equivalences and partial congruences~\citep{eqsat-semantics,kozen1993partial}
  to distinguish terms represented by an
  \egraph from those that are not.
A partial equivalence relation (PER) $\approxeq$ is a symmetric and
  transitive binary relation.
A partial congruence relation is a PER augmented with
  the congruence and subterm properties:
\begin{align*}
  f(t_1,\ldots, t_k) \approxeq f(t_1,\ldots, t_k) \land
\bigwedge_{i=1\ldots k}t_i \approxeq t'_i &\quad \implies \quad f(t_1, \ldots, t_k)\approxeq f(t'_1,\dots, t'_k) \tag*{\textsc{(Congruence)}}\\
 f(t_1,\ldots, t_k) \approxeq f(t_1,\ldots, t_k) &\quad \implies\quad  \bigwedge_{i=1\ldots k}t_i \approxeq t_i \tag*{\textsc{(Subterm)}}
\end{align*}

Given a set of grounded identities $E = \{t_1 = t_2, \ldots, t_{m-1} = t_m\}$,
  the congruence closure of $E$ is the smallest
  partial congruence relation containing $E$.
Note that the congruence closure may be infinite
  even if $E$ is finite.

\begin{wrapfigure}{r}{0.17\textwidth}
  \centering
  \includegraphics[width=\linewidth]{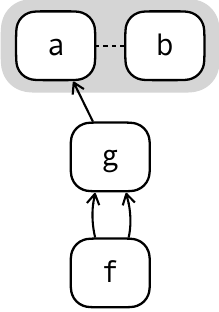}
  \caption{An example \egraph.}
  \Description{A small diagram of three e-classes. At the top, boxes labelled a
    and b sit together inside one shaded region, joined by a dotted line marking
    them equivalent. Below, a box labelled g has an arrow up into that shaded
    region. Below that, a box labelled f has arrows up to g.}
  \label{fig:example-egraph}
\end{wrapfigure}

\myparagraph{\Egraphs}
An \egraph $\mathcal{G}$ is a data structure computed
  from grounded equalities $E$
  that compactly represents the congruence closure of $E$ \citep{common-subexpr,eqsat-semantics}.
We write the PER
  represented by an \egraph $\mathcal{G}$ as $\approxeq_{\mathcal{G}}$.
Note that PERs are not necessarily reflexive:
  $x \approxeq_{\mathcal{G}} x$ holds for all $x$ represented by the \egraph
  but it does not hold for terms that are not represented by the \egraph.
An \egraph is a set of \textit{\eclasses},
  and each \eclass of an \egraph is a finite set of terms.
Each \eclass stores \emph{representative terms}
  drawn from the subterms of the original set of grounded identities $E$.
The \eclasses of an \egraph have a one-to-one correspondence
  with the equivalence classes of the congruence closure.
We say a term $t$ is \emph{represented} by
  \eclass $C \in \mathcal{G}$ when
  $\exists t' \in C$ such that $t' \approxeq_\mathcal{G} t$.
We denote the set of all terms represented by an \eclass with
  $\termsc{C}{\mathcal{G}}$.
Given a term $t \in \termsc{C}{\mathcal{G}}$,
  we define $\classof{t}{\mathcal{G}} = C$\footnote{
  The representative terms in our definitions
  serve as e-nodes in frameworks like egg.
  Using concrete representative terms makes defining our algorithm more natural.}.
Note that $\termsc{C}{\mathcal{G}}$ may be infinite, even though
  $C$ is finite\footnote{This is traditionally explained by cycles in the \egraph between e-nodes. We use grounded terms in this formalization, so the cycles are formed by the combination of equivalences between terms and the parent-child relationship on terms.}.

\autoref{fig:example-egraph} shows an example \egraph.
Dotted edges show equivalences and gray shading is used to highlight
  equivalence classes.
It has three \eclasses, with representative terms
$\{a,b\}$, $\{g(a)\}$, and $\{f(g(a), g(a))\}$ respectively.
The top \eclass represents equivalences $a\approxeq_{\mathcal{G}} b$,
 the middle \eclass represents equivalences $g(a)\approxeq_{\mathcal{G}} g(b)$,
 and the bottom \eclass represents equivalences
 $f(g(a),g(a))\approxeq_{\mathcal{G}} f(g(a),g(b))\approxeq_{\mathcal{G}} f(g(b),g(a))\approxeq_{\mathcal{G}} f(g(b),g(b))$.

While the mathematical definition of an \egraph treats
  terms as trees,
  practical implementations store and visualize them as directed acyclic graphs (DAGs).
  Each node is annotated with a function symbol and has ordered outgoing edges.
A term can be represented as a DAG by sharing common subexpressions,
  and a DAG can be unfolded into a tree by duplicating shared subterms.
Implementations typically rely on hash-consing to ensure
  any identical term corresponds to exactly one node in the DAG~\cite{egg}.
For example, the term $f(g(a), g(a))$ in \autoref{fig:example-egraph}
  is represented as a DAG in which the subterm $g(a)$ is shared.

In the context of program optimization,
 declarative optimization rules are applied repeatedly over the entire \egraph to find new identities, which are then
 added to $E$.
This technique is known as \emph{equality saturation}~\cite{peggy}
 and is implemented in multiple modern frameworks \citep{egg,egglog}.
After rule application,
  an optimal program is \textit{extracted} from the grown \egraph.
More formally, an extraction of an \eclass $C$ of an \egraph $\mathcal{G}$
 is a term $t \in \termsc{C}{\mathcal{G}}$.
The \egraph extraction problem is to find an extraction with low
 cost according to a cost function.
This paper focuses on this extraction procedure,
  but for \egraphs containing effectful programs, i.e., \textbf{effectful \egraphs}.

\subsection{Dataflow-based IRs and RVSDGs}
\Egraphs support a class of intermediate representations (IRs) based on dataflow graphs.
Each node in a dataflow graph represents a value.
Unlike control-flow graphs (CFGs),
  dataflow graphs directly encode data dependencies.
Dataflow graphs enable substitution of equivalent nodes and
  therefore satisfy congruence, a key property of \egraphs.

In a dataflow-based IR,
  evaluating nodes in any topological order must yield the same value.
Effectful programs, however, require ordering of effects.
For example, when printing to stdout, the order of operations matters.
Dataflow graph IRs support effectful programs by defining
  a special \emph{state value} that is threaded through
  all effectful operators to enforce a total order on effects.
All non-state values are \emph{pure values}.
Operators that produce state values are
  \textit{effectful} and those that produce pure values are \textit{pure}.
To ensure the dataflow graph is \emph{effect-safe},
  state values must be used linearly:
  each effectful operation must consume a state value and produce a new one.
In other words, the state value must not be duplicated or dropped.
We call the unique sequence of effectful operations on
  the state value the \emph{statewalk}.

Many dataflow IRs support control flow.
Our implementation (\autoref{sec:impl}) uses Regionalized Value State
  Dependence Graphs (RVSDGs)~\cite{cfgToRvsdg, rvsdg}.

\section{Overview}
\label{sec:overview}

This section
  illustrates the challenge of
  \effectsafe extraction
  and outlines \tiger.

\subsection{Effect-Safe Extraction}

\begin{figure}
  \centering
  	\includegraphics[width=1.0\textwidth]{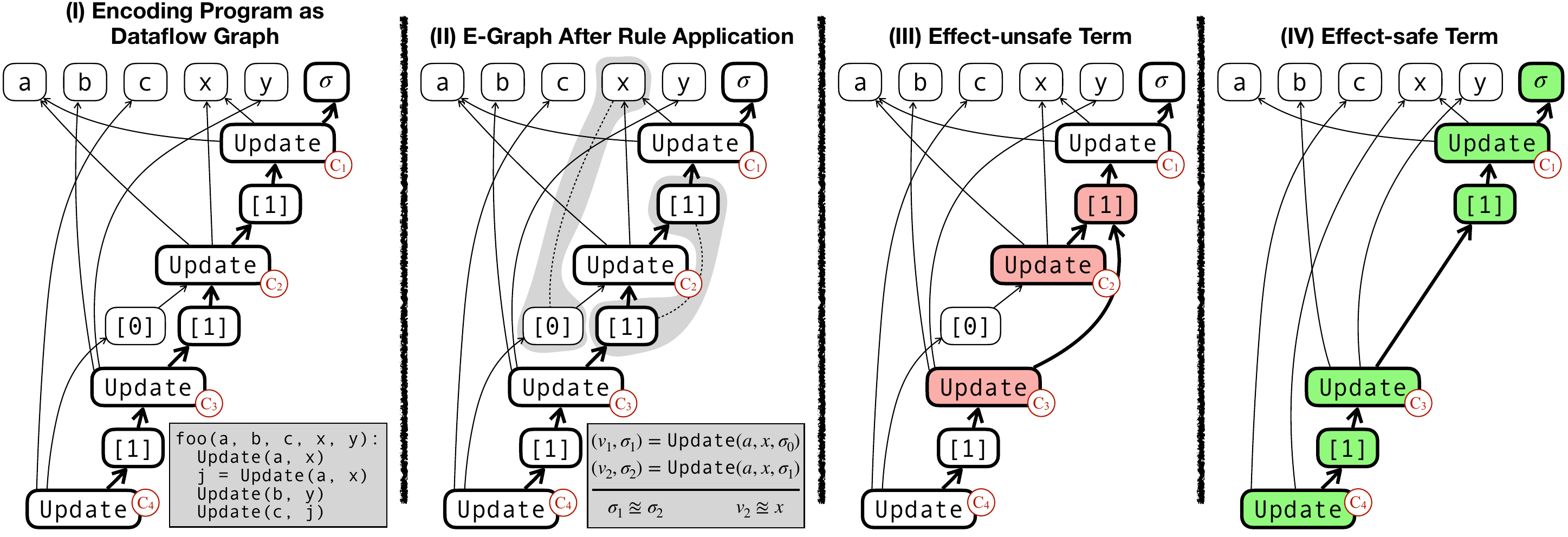}
  \caption{
    (I) An imperative program encoded as a dataflow graph;
    (II) The \egraph after applying an optimization rule that
      eliminates a redundant \C{Update};
    (III) An \effectunsafe term represented in the \egraph
      with two \C{Update} nodes relying on the same state value;
    (IV) An \effectsafe term represented in the \egraph
      that uses the state value linearly.
    Data flows top to bottom:
      operators point to their operands.
    Here and in later figures,
      effectful operators are drawn with a bold outline;
      all other operators are pure.
  }
  \Description{
    The \egraph contains two equalities.
    The first is between the state output of the first and second Update calls.
    The second is between the value output of the second Update call and the value of x.
  }
  \label{fig:overview-example}
\end{figure}

The core challenge is that
  not all terms represented by an effectful \egraph
  correspond to valid effect orderings.
We call such terms \textit{\effectunsafe}.
Crucially, despite having a well-defined semantics
  (\autoref{sec:formalizing-rvsdgs}),
  \effectunsafe terms cannot be translated back to sequential programs (for code generation).
At the same time, having \effectunsafe terms in an \egraph is unavoidable.
This follows from congruence, a key invariant of \egraphs stating that
 equivalent subterms can be substituted for each other in any context.
When the two equivalent subterms are effectful,
  this immediately leads to \effectunsafe terms with duplicated or dropped effects.

\autoref{fig:overview-example} illustrates the issue.
Panel~I encodes
  a simple imperative program as
  a dataflow graph.
The program contains four \C{Update} operations: $C_1$, $C_2$, $C_3$, $C_4$
  (in order of execution).
Each \C{Update}
  takes an address, a pure value, and a state value as arguments;
  writes the value to that address; and
  returns both the \textit{old} value previously stored at that address
  (indexed as \texttt{[0]} in the figure),
  and an updated state value
  (indexed as \texttt{[1]}).
Each effectful operation takes and returns a state value,
  which is threaded through the dataflow graph to encode the order of effects.

Panel~II shows the \egraph after
  applying an optimization rule that
  merges redundant \C{Update} operations:
when the same \C{Update} is repeated,
  the second performs no observable effect;
  it returns the same value and state as the first.
The two ground equalities in the \egraph reflect this equivalence:
the state after $C_1$ is equal to the state after $C_2$, and the
  pure value returned by $C_2$ is equal to \C{x}.
However, the resulting \egraph now contains
  both \effectsafe and \effectunsafe terms:
  congruence allows substituting a subterm with an equivalent one,
  even when both are effectful.

Panel~III shows
  an \effectunsafe term
  represented in the \egraph.
It is \effectunsafe because
$C_2$ and $C_3$ share
  the same input state.
Although this term has
  well-defined semantics (\autoref{sec:formalizing-rvsdgs}),
  it cannot be translated back to a valid program because
  the dataflow graph no longer defines a
  unique total order on effects.
As a result,
  the effectful operations could be scheduled in different ways,
  producing different behaviors when memory addresses alias:

\vspace{0.5em}
\begin{center}
\adjustbox{max width=\linewidth}{%
  \begin{minipage}[c]{0.2\linewidth}
\begin{eggccterm}
Update(a, x)
j = Update(a, x)
Update(b, y)
Update(c, j)
\end{eggccterm}
  \end{minipage}
  \hfill
  \begin{minipage}[c]{0.2\linewidth}
  \centering\texttt{-- or --}
  \end{minipage}
  \hfill
  \begin{minipage}[c]{0.2\linewidth}
\begin{eggccterm}
Update(a, x)
Update(b, y)
j = Update(a, x)
Update(c, j)
\end{eggccterm}
  \end{minipage}
}
\end{center}
\vspace{0.5em}

\noindent
Panel~IV shows an \effectsafe alternative.
Each effectful operation consumes the
  state value produced by the previous one,
  establishing a single well-defined order.

The key difficulty is that
  \effectsafety \textit{does not compose under congruence}.
A term may become \effectunsafe if
  one of its (\effectsafe) subterms is substituted with another (\effectsafe) term.
In the example,
  substituting the pure result of $C_2$ for
  $C_4$'s argument transforms
  the \effectsafe term from Panel~IV into
  the \effectunsafe one from Panel~III.

To summarize, our goal is to efficiently
 find an effect-safe extraction---a term from the \egraph whose
  effectful operations
  form a single linear order---with a small cost.
This is challenging because \effectsafety does not compose.
\autoref{sec:np} shows that
  finding \textit{any} \effectsafe extraction
  is NP-complete in general.
In practice, however,
  \egraphs produced by real program optimizers
  have additional structure,
  which \tiger exploits to efficiently extract \effectsafe terms.

\subsection{Statewalks and Statewalk Width}

The key idea behind \tiger
  is that different \statewalks
  (sequences of effectful operations)
  can expose the same opportunities
  for extracting pure subterms.
Rather than considering each \statewalk independently,
  \tiger identifies classes of \statewalks that
  enable the same extractions.
This lets \tiger explore each class of equivalent (i.e., substitutable under congruence) \statewalks
  once for each \eclass,
  without enumerating the full, potentially infinite,
  space of possible \statewalks.

Concretely, \tiger groups \statewalks by
  the set of pure \eclasses they make available.
We call two \statewalks
  \textit{\extractionset equivalent}
  when they are in the same \eclass and enable the same set of
  extractable pure \eclasses.

In \autoref{fig:overview-example}'s running example,
  there are infinitely many \statewalks
  to the root (bottom) \eclass.
Each \statewalk may
  Update \texttt{a}
  any number of times before
  updating \texttt{b} and \texttt{c}.
However,
  all these \statewalks enable
  the same pure subterms.
They are therefore \extractionset equivalent,
  and an extractor is free to choose any of them;
  \autoref{app:tiger-example} shows the \extractionsets
  \tiger computes for a small \egraph.
\tiger keeps only the cheapest representative,
  which in this case
  yields the \effectsafe extraction
  that avoids the redundant \texttt{Update}
  (Panel~IV).

To efficiently exploit \extractionset equivalence,
  \tiger maintains a dynamic programming (DP) table
  indexed by (\eclass, \extractionset) pairs.
For each pair,
  the table stores the cheapest
  representative \statewalk and the
  corresponding extracted term.
This is sound because
  \textit{\extractionset equivalent \statewalks are substitutable}:
  each leads to the same pure subterms and
  can be used interchangeably when
  building \effectsafe terms.
Starting from the leaves,
  \tiger propagates representative \statewalks
  from child operands to their parent operators.
This DP formulation
  reuses the shared decisions between
  \extractionset equivalent \statewalks.

We characterize the complexity of \tiger
  using a new parameter, \textit{\statewalkwidth}.
The \statewalkwidth of an \eclass
  is the number of distinct \extractionsets
  enabled by its \statewalks,
  and the \statewalkwidth $k$ of an entire \egraph
  is the maximum \statewalkwidth over its \eclasses.
In \autoref{fig:overview-example},
  every \statewalk exposes the same set of pure \eclasses,
  so the \statewalkwidth is~1.
If an \egraph representing $n$ terms has \statewalkwidth $k$,
  \tiger runs in $O(kn^2)$ time.
Our evaluation (\autoref{sec:eval}) shows
  that \statewalkwidth generally remains small in practice,
  enabling \tiger to extract effectful programs efficiently
  without sacrificing extraction quality.

The remainder of this paper
  formalizes the equivalences used in \tiger,
  defines statewalk width precisely,
  gives the full DP algorithm, and
  proves its correctness.
We then evaluate \tiger's performance and
  extraction quality across a suite of benchmarks.

\section{The Effect-Safe Extraction Problem}
\label{sec:formalizing-rvsdgs}

In this section, we formally define the \effectsafeext problem.
We first define the equivalence relation $\approxeq$
  over a semantics of effectful dataflow IR
  including \effectunsafe terms.
We formally define \effectsafety
 and PER
 $\equiv_\mathcal{G}$,
 which is $\approxeq$ restricted to
 \effectsafe programs in \egraph $\mathcal{G}$.
Finally, we define the \effectsafeext problem as
  finding an \effectsafe term in
	an equivalence class of $\equiv_\mathcal{G}$
	that corresponds to a root \eclass of $\mathcal{G}$.
Auxiliary definitions are collected in \autoref{app:definitions}.

%
%





\subsection{Abstract Dataflow Language and $\approxeq$}

\label{sec:semantics}


\begin{figure}[t]
\begin{subfigure}[t]{\textwidth}
\sf
\lstset{language=eggccir,basicstyle=\footnotesize\ttfamily}
\setlength{\grammarparsep}{5pt plus 1pt minus 1pt}
\setlength{\grammarindent}{6em}

\begin{grammar}
<Base values \baseval> $\in$ $\Phi=\{\bot,\top,\ldots \}$
\quad\quad
<States \stateval> $\in$ \statevals
\quad\quad
<Pure values> $\val_p$ $\in$ $V^p$ ::= \baseval \; $\mid$ $(\baseval, \ldots, \baseval)$

<Stateful values $\val_e$> $\in$ $V^e$ ::= \stateval \; | \; $(\baseval, \ldots, \baseval, \stateval)$
\quad\quad
<Values \val> ::= $\val_p$ | $\val_e$

<Pure Operators> \texttt{f}, \texttt{g}, $\dots$ $\in$ {$(V^p)^n \to V^p$}
\quad\quad\quad
<Effectful Operators> \texttt{F}, \texttt{G}, $\dots$ $\in$ {$(V^p)^{n-1} \times V^e \to V^e$}

<Indices n> $\in$ {\ensuremath{\N}}
\quad\quad\quad\quad
<Expression e> ::= {\ensuremath{v}}
               | {\argt}
               | {\texttt{f}($e_1$, \dots, $e_n$)}
               | {\texttt{F}($e_1$, \dots, $e_n$)}
               | ($e$, \dots)
               | e[$n$]
\end{grammar}
\caption{Syntax.}
\label{fig:syntax}
\end{subfigure}

\bigskip

\begin{subfigure}[t]{\textwidth}
  \footnotesize
  \newcommand\semrulesep{1.75em}
  \newcommand{\append}[2]{#1\ \texttt{+\hspace{-1pt}+}\ #2}
  \lstset{language=eggccir}
  \begin{equation*}
  \begin{array}{llll}
    \inference[]{}
      {
        \bigstep{\argval}{v}{v}
      }
      &
      \inference[]{}
      {
        \bigstep{\argval}{\argt}{\argval}
      }
      &
      \inference[]{
        \bigstep{\argval}{e_i}{v_i}\quad\text{for } i=0,\ldots,k-1
      }
      {
        \bigstep{\argval}{(e_0,\ldots, e_{k-1})}{(v_0,\ldots, v_{k-1})}
      }
      &
      \inference[]{
        \bigstep{\argval}{e}{(v_0,\ldots, v_{n-1})} \andalso
        0 \le k < n
      }
      {
        \bigstep{\argval}{e[k]}{v_k}
      }
    \end{array}
    \end{equation*}
    \vspace{1em}
  \begin{equation*}
  \begin{array}{ll}
    \inference[]{
        \bigstep{\argval}{e_i}{v_i}\ \text{for } i=1,\ldots,n \andalso
        \llbracket \texttt{f} \rrbracket(v_1,\ldots,v_n) = v'
      }
      {
        \bigstep{\argval}{\texttt{f}(e_1,\ldots,e_n)}{v'}
      }
      &
			\inference[]{
        \bigstep{\argval}{e_i}{v_i}\ \text{for } i=1,\ldots,n \andalso
        \llbracket \texttt{F} \rrbracket(v_1,\ldots,v_n) = v'
      }
      {
        \bigstep{\argval}{\texttt{F}(e_1,\ldots,e_n)}{v'}
      }
    \end{array}
  \end{equation*}
\caption{Big-step operational semantics.}
\label{fig:semantics}
\end{subfigure}
\caption{Syntax and semantics of a minimal effectful dataflow IR.}
\Description{Typeset text rather than a picture. The upper part is a grammar
  giving the base values, states, pure and stateful values, pure and effectful
  operator signatures, indices, and the expression forms. The lower part is six
  big-step inference rules, one each for values, Arg, tuple construction, tuple
  indexing, pure operator application, and effectful operator application.}
\end{figure}

\renewcommand{\arraystretch}{1.2}
\begin{table}[t]
\centering
\renewcommand{\arraystretch}{1.0}
\caption{Notions of equivalence used in this paper,
  each introduced in the section named beside it,
  and each a refinement of the one above it.
  $\bigstep{\argval}{t}{v}$ is the big-step judgement of \autoref{fig:semantics}.
  $\termsc{C}{\mathcal{G}}$ is the set of terms represented by \eclass $C$, and
    $\classof{t}{\mathcal{G}}$ is the \eclass of a term $t$
    (\autoref{subsec:backgroundegraphs}).
  $\effectsafe(t)$ holds when $t$ uses the state linearly
    (\autoref{def:effect-safety}), and $\text{\effectful}(t)$ when $t$ evaluates
    to a stateful value (\autoref{sec:effect-safety}).
  $\ES(t)$ is the \extractionset of $t$: the pure \eclasses that become
    extractable once the effects along $t$ have been performed
    (\autoref{sec:statewalkequiv}).}
\label{tab:equivalences}
\begin{tabular}{|c|c|c|c|}
\hline
Equivalence & Description & Definition \\
\hline
$t_1\approxeq t_2$ & Equiv.~under the semantics (\autoref{sec:semantics}) & $\forall \argval, v.~ \bigstep{\argval}{t_1}{v} \iff \bigstep{\argval}{t_2}{v}$ \\
\hline
$t_1\approxeq_{\mathcal{G}} t_2$ & Equiv.~in an \effectful \egraph (\autoref{sec:semantics}) & $\exists C.~ t_1,t_2\in \termsc{C}{\mathcal{G}}$ \\
\hline
$t_1\equiv_{\mathcal{G}} t_2$ & Effect-safe refinement of $\approxeq_{\mathcal{G}}$ (\autoref{sec:effect-safety}) & $t_1 \approxeq_{\mathcal{G}} t_2 \land \effectsafe(t_1) \land \effectsafe(t_2)$ \\
\hline
$t_1\statewalkequiv t_2$ & Statewalk refinement of $\equiv_\mathcal{G}$ (\autoref{sec:statewalkequiv}) & $t_1 \equiv_\mathcal{G} t_2 \wedge\text{\effectful}(t_1)\wedge $ \\
& & $ \text{\effectful}(t_2)\wedge \ES(t_1) = \ES(t_2)$ \\
\hline
\end{tabular}
\end{table}


We define (acyclic) dataflow graphs as terms as described in
  \autoref{sec:background}, making them amenable for formal
  treatment in an \egraph.
\autoref{fig:syntax} gives the syntax of a minimal dataflow IR.
We categorize values into pure ones and stateful ones.
A pure value is a base value or a tuple of base values, and a stateful value is either a state
 or a tuple whose last value is a state.
Tuples do not nest.
A pure operator maps pure values to a pure value.
An effectful operator takes
 any number of pure values plus a single state value, written last,
 and produces a stateful value.
Every term therefore has at most one \effectful child, so a
 \statewalk (\autoref{sec:effect-safety}) follows a chain of effects unambiguously.
An expression is a value, an \ensuremath{\argt} variable that refers
 to program inputs\footnote{We assume a dataflow program takes a unique tuple-valued variable \ensuremath{\argt}.},
 an application of a pure or an effectful operator,
 a tuple of expressions,
 or an indexing into an expression.


\autoref{fig:semantics} gives a big step operational semantics for
  the language defined in \autoref{fig:syntax}.
The judgement $\bigstep{\argval}{e}{v}$
  indicates that
  expression $e$ evaluates to value $v$ when
  its argument {\ensuremath{\argt}} is bound to the stateful value $\argval$.
We define the relaxed equivalence relation $\approxeq$ with respect to this semantics, i.e.,
$t_1 \approxeq t_2$ if and only if $\forall \argval, v.\; \bigstep{\argval}{t_1}{v} \iff \bigstep{\argval}{t_2}{v}$.
The list of equivalences defined in this paper is summarized in \autoref{tab:equivalences}.

The semantics is parametrized over the choice of base values and states.
For example, the evaluation section of this paper defines
  values to be primitives like pointers and booleans,
  while a state is defined as a pair of heap (mapping from location to values) and a log of effects
  (e.g., from printing).\tighten
Notice that the evaluation rule for effectful operators
  (\texttt{F}) is identical to that of pure operators
  (\texttt{f}).
It does not forbid effectful operations from duplicating or dropping state values.
However, only effect-safe programs can be
 used for code generation.



Consider the original program in \autoref{fig:overview-example}
  which performs four \texttt{Update} operations.
It is equivalent with respect to our semantics to
  the \effectunsafe program shown in the third pane.
We write both below in the syntax of \autoref{fig:syntax}, with
  $a$, $b$, $c$, $x$, $y$, and $\sigma$ for the components
  $\argt[0]$ through $\argt[5]$.

\begin{figure}[h]
  \centering
  \begin{minipage}[t]{0.45\linewidth}
    \textbf{Effect-safe:}
    {\small
\begin{lstlisting}[language=eggccir,basicstyle=\footnotesize\ttfamily,numbers=none,xleftmargin=0pt]
Update(c,
  Update(a, x, Update(a, x, $\sigma$)[1])[0],
  Update(b, y,
    Update(a, x, Update(a, x, $\sigma$)[1])[1])[1])
\end{lstlisting}
    }
  \end{minipage}
  \hspace{1em}
  \begin{minipage}[t]{0.45\linewidth}
    \textbf{Effect-unsafe:}
    {\small
\begin{lstlisting}[language=eggccir,basicstyle=\footnotesize\ttfamily,numbers=none,xleftmargin=0pt]
Update(c,
  Update(a, x, Update(a, x, $\sigma$)[1])[0],
  Update(b, y, Update(a, x, $\sigma$)[1])[1])
\end{lstlisting}
    }
  \end{minipage}
  \Description{Two code listings side by side rather than a picture, both
    written as a single nested term. The left, labelled effect-safe, threads the
    state through four Update operations in sequence. The right, labelled
    effect-unsafe, is identical except that two different Update terms take the
    same state value as their argument.}
\end{figure}


The \effectsafe original program performs four updates,
  each one consuming the state value produced by the previous update.
The \effectunsafe program instead re-uses the state value
  \texttt{Update(a, x, $\sigma$)[1]} in two distinct \verb|Update| terms:
  the outer \texttt{Update(a, x, \dots)} that encloses it,
  and the \texttt{Update(b, y, \dots)} beside it.
The first of those writes \texttt{x} to \texttt{a} while the second
  writes \texttt{y} to \texttt{b}.
These programs are equivalent with respect to the semantics
  in \autoref{fig:semantics}, but the \effectunsafe program cannot be translated back
  to an executable program.

$\approxeq$ relates all pairs of program terms equal under this semantics.
In practice, any \egraph $\mathcal{G}$
  represents a subset of $\approxeq$,
  as long as each equality used to build $\mathcal{G}$ is sound under $\approxeq$.
We denote this under-approximated equivalence relation represented by $\mathcal{G}$ as $\approxeq_{\mathcal{G}}$,
  which satisfies $(\approxeq_{\mathcal{G}}) \subseteq (\approxeq)$.

\myparagraph{Authoring rules}
An optimization rule over the \egraph is sound
  exactly when the terms it equates
  agree under the semantics of \autoref{fig:semantics}.
This is the only requirement we place on rules.
A sound rule may equate an \effectsafe term with an \effectunsafe one.
Rules and extraction therefore work over different relations
  (\autoref{tab:equivalences}):
  rules build $\approxeq_{\mathcal{G}}$,
  while extraction searches its \effectsafe restriction $\equiv_{\mathcal{G}}$.
That separation makes \eqsat over effectful programs possible,
  since an optimization such as
  the redundant-\texttt{Update} rule of \autoref{fig:overview-example}
  adds \effectunsafe terms to the \egraph.
Note that \effectsafety does not rely on
  design decisions specific to \eggcc's IR,
  such as its use of regions to model control flow (\autoref{sec:controlflows}).

%
%

\subsection{Effect Safety and $\equiv_\mathcal{G}$}

\label{sec:effect-safety}

A term is \effectful if it evaluates to a stateful value
  (and pure otherwise).
This happens if the term is $\argt$, an effectful operator,
 or an indexing expression into the last element of a stateful tuple value.
$\argt$ is the only effectful term that has no child terms.
Purity is a property of the value a term denotes,
  not of how that value is computed.
In the example above, $\texttt{Update}(a, x, \sigma)[0]$ is pure ---
  it denotes the integer previously stored at $a$ ---
  even though obtaining it performs an effect.
An \eclass is also either \effectful or not,
  containing only \effectful terms or only pure terms.
An \eclass that contains both is unsound, since
  $\approxeq$ does not relate any effectful term to a pure term.
Nothing prevents a rule from merging an effectful term with a pure one,
  but such a rule is unsound by the criterion of \autoref{sec:semantics},
  and \tiger's guarantees, like those of any extractor,
  hold only for \egraphs built from sound rules.\tighten


We define \effectsafety and \statewalks, which are \effectsafe by construction, as follows:
\begin{definition}[Statewalk]
  \label{def:statewalk}
  For $k \geq 0$, $\sw = \langle e_0, e_1, \ldots, e_k \rangle$ is a \statewalk if:
  \begin{itemize}[leftmargin=1em,nosep]
    \item $e_0 = \argt$.
    \item For all $e_i \in \sw$, $e_i$ is \effectful.
    \item For each pair of terms $e_i, e_{i+1} \in \sw$, $e_i \in \children{e_{i + 1}}$.
    \item For all $e_i \in \sw$ and for all $t \in \children{e_i}$, all effectful subterms of $t$
      are in $\langle e_0, e_1, \ldots, e_{i - 1} \rangle$.
  \end{itemize}
\end{definition}

\begin{definition}[Effect-safety]
  \label{def:effect-safety}
  A term $t$ is \effectsafe under \statewalk $\sw$, denoted $\effectsafe(t, \sw)$,
    if all effectful subterms of $t$ are in $\sw$.
  We also say a term $t$ is effect-safe, denoted $\effectsafe(t)$,
    if it is effect-safe under some statewalk (i.e., $\exists \sw. \effectsafe(t, \sw)$).
\end{definition}

Note that a \statewalk is uniquely determined by
  its last term, since the entire \statewalk
  can be constructed by following
  the child relation on effectful terms.
For example, consider again the initial program in \autoref{fig:overview-example}
  which performs four \texttt{Update} operations.
It is \effectsafe and so it has a \statewalk,
  obtained by following the child relation
  through every term that carries the state value.
Below we name the three \texttt{Update} subterms
  $u1$, $u2$, and $u3$.
Some subterms occur more than once;
  a practical \egraph implementation stores each term only once using a hash-cons data structure.

\begin{figure}[h]
  \centering
  \begin{minipage}[t]{0.45\linewidth}
    \textbf{Term:}
    {\small
\begin{lstlisting}[language=eggccir,basicstyle=\footnotesize\ttfamily,numbers=none,xleftmargin=0pt]
(a, b, c, x, y, $\sigma$) = Arg
u1 = Update(a, x, $\sigma$)
u2 = Update(a, x, u1[1])
u3 = Update(b, y, u2[1])
Update(c, u2[0], u3[1])
\end{lstlisting}
    }
  \end{minipage}
  \hspace{1em}
  \begin{minipage}[t]{0.45\linewidth}
    \textbf{Statewalk:}
    {\small
\begin{lstlisting}[language=eggccir,basicstyle=\footnotesize\ttfamily,numbers=none,xleftmargin=0pt]
$\langle$Arg, $\sigma$,
 u1, u1[1],
 u2, u2[1],
 u3, u3[1],
 Update(c, u2[0], u3[1])$\rangle$
\end{lstlisting}
    }
  \end{minipage}
  \Description{Two code listings side by side rather than a picture. The left,
    labelled Term, writes the four-Update program using the intermediate names
    u1, u2 and u3 for its three inner Update subterms. The right, labelled
    Statewalk, gives the corresponding statewalk as a nine-element sequence
    beginning with Arg and alternating each Update with the indexing expression
    that extracts its state.}
\end{figure}

This forms a valid statewalk because $e_0 = \argt$, all
  $e_i$ are effectful,
  each $e_i$ is a child of $e_{i+1}$,
  and all effectful terms occur in the sequence.
The \statewalk contains more than just the four \texttt{Update} terms
  because every term that carries the state value is \effectful and
  therefore belongs to it, including
  the indexing expressions that extract the state:
  $\sigma$ (that is, $\argt[5]$) and each $ui[1]$.
Pure terms such as $a$ and $u2[0]$ are not in the \statewalk;
  they are instead the terms the \statewalk makes available
  for extraction (\autoref{sec:statewalkequiv}).

\begin{definition}

	The refined equivalence $\equiv_{\mathcal{G}}$ for \effectsafe programs for \egraph $\mathcal{G}$ is defined as:
  \[t_1 \equiv_{\mathcal{G}} t_2 \iff t_1 \approxeq_{\mathcal{G}} t_2 \land \effectsafe(t_1) \land \effectsafe(t_2)\]
\end{definition}

By transitivity of relation refinement,
 $(\equiv_{\mathcal{G}}) \subseteq (\approxeq_{\mathcal{G}}) \subseteq (\approxeq$).
Given an \egraph $\mathcal{G}$ and an \eclass $C$ of $\mathcal{G}$, the problem of \textit{\effectsafeext} is
 to efficiently find an \effectsafe
  term $t$ represented by $C$.
Formally, we want to find a term $t \in \termsc{C}{\mathcal{G}}$ that satisfies $\effectsafe(t)$.
In \autoref{sec:np}, we show that \effectsafe
  extraction is NP-complete.
\tiger solves this problem efficiently
  by leveraging \statewalkwidth (\autoref{sec:tiger}).
\tiger also greedily optimizes for the cost of the extraction,
 so it is a practical heuristic for
  the \textit{\opteffectsafeext} problem, which requires
  $t$ to have minimal cost according to a given cost model
  ($\argmin_{t: t\in \termsc{C}{\mathcal{G}}\land \effectsafe(t)} \textit{cost}(t)$).

\section{NP Completeness of \EffectSafeExt}
\label{sec:np}

\begin{figure}
  \centering
  	\includegraphics[width=1.0\textwidth]{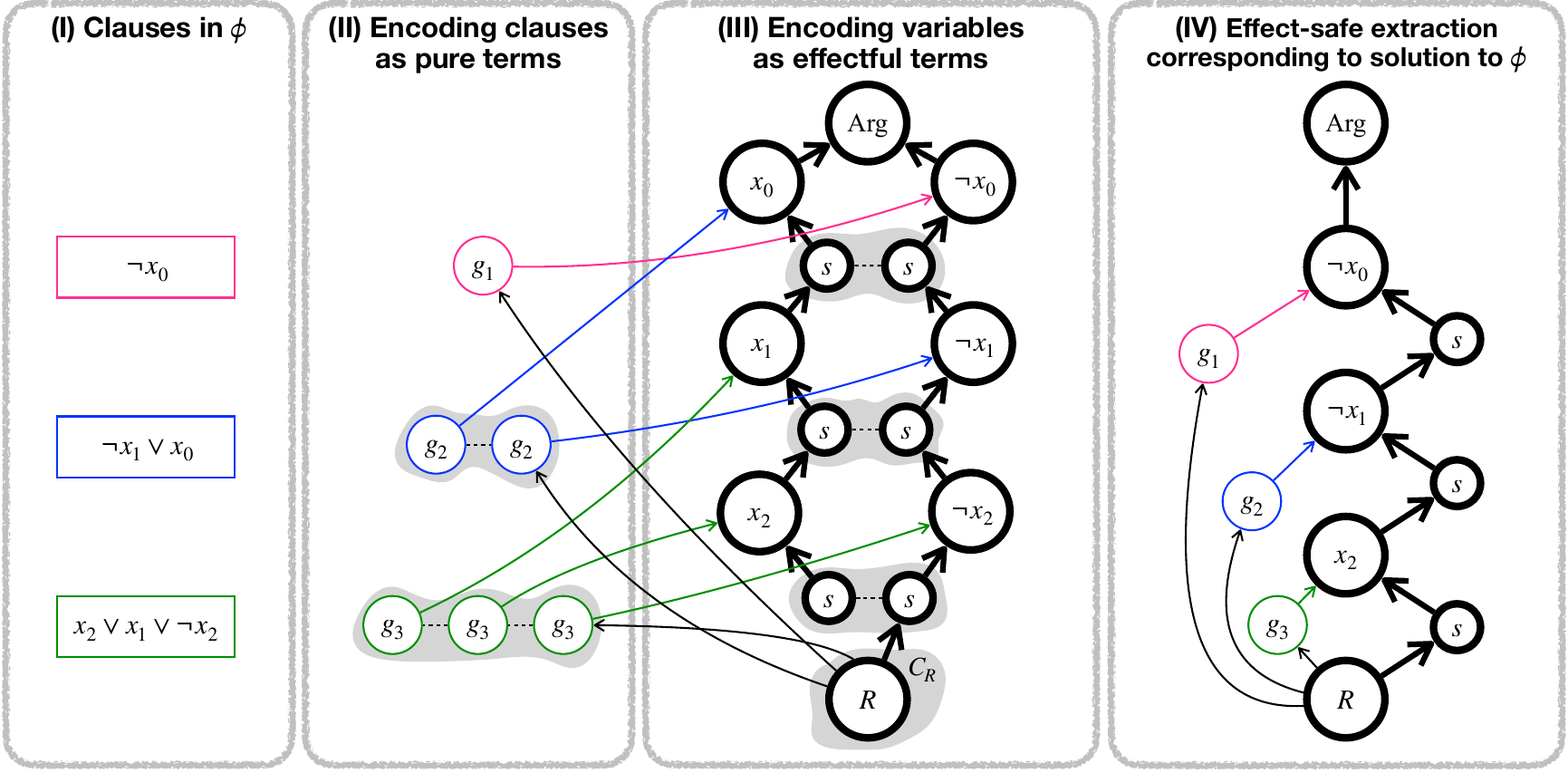}
  \caption{
    \Egraph encoding of the formula $\phi = \neg x_0 \land (\neg x_1 \lor x_0) \land (x_2 \lor x_1 \lor \neg x_2)$.
    Equivalences are shown with dotted lines, and equivalence classes are shown with a gray background.
    $C_R$ denotes the \eclass containing the root term $R(\ldots)$.
    Panel I shows the clauses in $\phi$.
    Panel II shows the encoding of each clause as an \eclass consisting of
      one pure term per literal in the clause.
    Panel III shows how the variables in $\phi$ are encoded using effectful terms, with auxiliary effectful terms $s(...)$ ensuring independence between variable choices.
    Panel IV shows an \effectsafe extraction of $R$, which corresponds to a solution to $\phi$.
  }
  \Description{A four-panel diagram of an e-graph encoding a boolean formula.
    Terms are drawn as circles, e-classes as gray backgrounds, and equivalences
    as dotted lines. Panel I lists the three clauses of the formula, each boxed
    in a different color. Panel II shows one e-class per clause, containing a
    pure term per literal, colored to match the clause. Panel III shows the
    variables encoded as a chain of effectful terms descending from Arg, with
    each variable and its negation paired and joined by auxiliary terms. Panel
    IV shows a single effect-safe extraction: one chain through the variable
    terms, with one clause term attached to each.}
  \label{fig:extraction-sat}
\end{figure}

In this section, we show that \effectsafe extraction is NP-complete.
First, we show how to encode CNF-SAT as an instance of \effectsafe extraction,
  showing that \effectsafe extraction is NP-hard.
Then, we show how to encode \effectsafe extraction in SAT, showing that
  \effectsafe extraction is NP-easy.
The full proofs are in \autoref{app:np-hardness} and \autoref{app:np-easy}.


\paragraph{NP-hardness} The high-level approach is to encode
  variables in the formula as effectful terms
  and clauses in the formula as pure \eclasses.
The \eclass for a clause contains one term
  for each literal in the clause.
The \effectful terms ensure that
  the \statewalks of the root \eclass correspond
  exactly to all possible assignments to
  the variables in $\phi$:
  each \statewalk includes exactly one of $x$ or $\neg x$
  for each variable $x$ in the formula.
Intuitively,
  variables correspond to choices of
    which statewalk to follow;
  clauses correspond to constraints on
    which pure terms have compatible effect orders;
  satisfying assignments correspond to
    effect-safe extractions.

\autoref{fig:extraction-sat} shows an example encoding
  of the formula $\phi = \neg x_0 \land (\neg x_1 \lor x_0) \land (x_2 \lor x_1 \lor \neg x_2)$.
An \effectsafe extraction of $C_R$ is shown in the last panel of \autoref{fig:extraction-sat}.
Its \effectsafety is witnessed by the \statewalk with final term
  $R(g_1(\ldots), g_2(\ldots), g_3(\ldots), s(x_2(s(\neg x_1(s(\neg x_0(\argt)))))))$.
This \statewalk corresponds to the satisfying assignment of $\phi$ given by
  $\{x_0 = \bot, x_1 = \bot, x_2 = \top\}$.
The extracted terms $g_i(\ldots)$ correspond to partial assignments of
  variables that ensure satisfaction of clause $i$ in $\phi$.


More generally, let $\phi = \varphi_1 \land \varphi_2 \land \ldots \land \varphi_k$ be an arbitrary boolean formula
  in conjunctive normal form
  where $\varphi_1, \varphi_2, \ldots, \varphi_k$ are clauses over variables $x_0, x_1, \ldots, x_n$.
We encode the formula using the following operators:
\begin{enumerate}
  \item Arg: the effectful leaf term.
  \item $s$: an effectful operator that takes a state value and produces a state value.
  \item $x_i$ and $\neg x_i$: effectful operators that take a state value and
    produce a tuple containing a pure value and a state value.
  \item $g_i$: a pure operator that takes a pure value and produces a pure value.
  \item $R$: an effectful operator that takes $k$ pure values and a state value.
\end{enumerate}
We construct an effectful \egraph $\mathcal{G}$ as follows\footnote{
          The dataflow IR syntax in \autoref{fig:syntax} would require
          indexing ($t[0]$ and $t[1]$)
          for unpacking the pure and stateful values out of tuples.
          These do not affect the encoding, and they are omitted for simplicity.
        }:
\begin{enumerate}
  \item Add effectful terms $s(x_0(\Arg))$ and $s(\neg x_0(\Arg))$, denoted as $t_{x_0}$ and $t_{\neg x_0}$.
  \item For each variable $i=1\ldots n$,
        add effectful terms $s(x_i(t_{x_{i-1}}))$ and $s(\neg x_i(s(t_{\neg x_{i-1}})))$, denoted as $t_{x_{i}}$ and $t_{\neg x_{i}}$.
  \item Add equivalences $t_{x_{i}}\approxeq_{\mathcal{G}} t_{\neg x_{i}}$ for each $x_i$.
  \item For each clause $\varphi_i = \ell_1 \lor \ldots \lor \ell_m$,
        and for each literal $\ell_j \in \varphi_i$, where $\ell_j$ is either
        $x_k$ or $\neg x_k$ for some $k$,
        add the pure term $g_i(x_k(t_{x_{k-1}}))$ or $g_i(\neg x_k(t_{\neg x_{k-1}}))$, respectively,
        denoted as $t_{\varphi_{i,j}}$.
  \item For each clause $\varphi_i= \ell_1 \lor \ldots \lor \ell_m$, add equivalences between all $t_{\varphi_{i,j}}$ for
  $j=1\ldots m$.
  \item Add an effectful term $R(t_{\varphi_{1,1}},\ldots t_{\varphi_{k,1}}, t_{x_n})$; denote the \eclass of this term as $C_R$.
\end{enumerate}

\paragraph{NP-easiness}
  Finding any \egraph extraction, without considering \effectsafety,
    can be encoded as an instance of SAT using $O(n^3)$ variables and clauses,
    where $n$ is the size of the \egraph.
  To find an \effectsafe extraction,
		instead of finding just a single term,
		we solve for a \statewalk,
		which is a sequence of terms.
	This can be done by cloning the single-term extraction encoding
	  for each potential term in the \statewalk and
	  adding constraints to ensure the sequence of terms forms a \statewalk.
  The key insight is that
	  while there is no upper bound on the length of a \statewalk,
    only \statewalks up to length $n^2$ need to be
    considered in the search for an \effectsafe extraction.
  As a result, it is possible to encode the \effectsafe extraction problem
    as an $O(n^5)$ SAT instance.
  The formal proof, including the proof extension for optimal \effectsafe extraction
    is in \autoref{app:np-easy}.

  Hardness does not leave an optimizer stuck, since the input program is itself
    an \effectsafe solution, and prior work falls back to it when extraction
    fails~\cite{peggy}.
  Our problem definition does not assume such a solution is in hand,
    because the goal is to find a \emph{better} program;
    whether the input can instead \emph{guide} extraction remains open.

\section{Statewalk DP}

\label{sec:tiger}
For pure \egraphs, extractors can pick any term
  from the root \eclass to give a valid solution.
Finding the optimal solution with respect to a cost model
  is known to be NP-complete~\cite{stepp-thesis, yihong-np},
  and existing solutions do not scale especially
  in the presence of cycles~\cite{goharshady}.
Therefore, the most ubiquitous and successful extraction
  algorithms greedily search for good terms without guaranteeing optimality.
In this section we present \tiger,
  an algorithm with guarantees of \effectsafe extraction
  while optimizing for cost.
For a worked example, see \autoref{app:tiger-example}.

Typical extraction algorithms for pure \egraphs build extractions bottom-up,
  greedily constructing extractions for each \eclass.
This approach is complete (always finding an extraction
  when one exists) and terminating.
Crucially, extraction algorithms for pure \egraphs leverage
  the \emph{congruence} property of \egraphs to
  choose any term from each \eclass when constructing new terms.
The extractor can freely compose
  a function symbol and a valid term from each direct child \eclass,
  allowing it to keep only one extracted term per \eclass.

With effects, the \egraph can contain \effectunsafe terms,
  and the congruence property no longer holds
  over \effectsafe terms.
Therefore, by default, terms in the same \eclass
  cannot necessarily be used interchangeably when constructing new terms.
In other words, composing \effectsafe terms does not
  necessarily lead to another \effectsafe term.
As a result, applying a standard extraction algorithm to \effectful \egraphs
  does not guarantee completeness: it can fail to find an \effectsafe extraction when one exists.
In fact, finding \textit{any} \effectsafe extraction is NP-complete (\autoref{sec:np}).

We solve this problem by grouping \effectsafe terms
  into equivalence classes, recovering the ability
  to choose any term within the same equivalence class.\footnote{
  We write \emph{equivalence class} for a class of one of the relations in
    \autoref{tab:equivalences}, and \emph{\eclass} for the terms an \egraph stores.}
We call the ability to substitute one \effectsafe term
  for another \emph{substitutability}.
The \tiger algorithm solves the \effectsafe extraction problem
	using a dynamic programming approach on the equivalence classes of $\statewalkequiv$,
	a refinement of $\equiv_{\mathcal{G}}$.
Our description focuses on the completeness of \tiger.
We present proof sketches for soundness and completeness in \autoref{sec:tiger-proof}
  and a complexity analysis in \autoref{sec:complexity}.

Our \tiger implementation uses a linear cost model
	and produces a best-effort low cost term
	as in a typical effect-unaware greedy extractor.
Our cost model and the optimality of \tiger's solution
	are discussed in \autoref{sec:cost-model}.

\subsection{Achieving Substitutability Using $\statewalkequiv$}
\label{sec:statewalkequiv}


Extractions can be composed when they are \effectsafe
  under the \emph{same} \statewalk.
Intuitively, this is because the subterms all rely on the same
  ordering of \effectful operations.
Every \effectsafe term corresponds to a \statewalk,
  so a na\"{i}ve extractor could exhaustively enumerate through all \effectsafe
  terms bottom-up, storing one term per \statewalk for each \eclass.
However, this approach is impractical
  because the number of \statewalks is unbounded.

We would like to reduce the number of \statewalks
  considered for each \eclass.
For pure \egraphs, congruence
  allows us to store only one term per \eclass.
We would like to recover a similar property
  for \effectsafe terms, which we call substitutability.
Our insight is an equivalence relation
  $\statewalkequiv$ which refines $\equiv_\mathcal{G}$,
  and allows us to recover
  substitutability, the ability to substitute one \effectsafe term for another.
With substitutability we will store
  dramatically fewer \statewalks.


First we define the \extractionset of a \statewalk,
  which we'll use to group \statewalks into equivalence classes under $\statewalkequiv$.
Informally, the \extractionset of \statewalk $\sw$ is the set of pure \eclasses in $\mathcal{G}$
which contain an \effectsafe term under $\sw$.
More precisely:
\begin{definition}[Extractable Set]
	$$\ES(\sw) = \{\classof{t}{\mathcal{G}} \mid \text{pure}(t) \wedge \effectsafe(t, \sw) \}$$
\end{definition}
Note that any \effectsafe \effectful term
  exactly corresponds to the \statewalk constructed from its \effectful subterms,
  so for a term $t$, $\ES(t)$ denotes the \extractionset of the \statewalk defined by $t$.


$\statewalkequiv$ equates any two \effectsafe \effectful terms,
	$t_1$ and $t_2$,
	that are equivalent under $\equiv_\mathcal{G}$
	and induce \statewalks with the same \extractionset.
\begin{definition}[Statewalk Equivalence, $\statewalkequiv$]
	$$t_1 \statewalkequiv t_2  \iff \text{\effectful}(t_1) \wedge \text{\effectful}(t_2) \wedge t_1 \equiv_\mathcal{G} t_2 \wedge \ES(t_1) = \ES(t_2)$$
\end{definition}

By this definition, $\statewalkequiv\; \subseteq\; \equiv_\mathcal{G}$,
	so solving the extraction problem for $\statewalkequiv$ solves the \effectsafe extraction problem of $\equiv_\mathcal{G}$.
Because equivalence classes of $\equiv_\mathcal{G}$
	have a one-to-one correspondence
	with \eclasses of $\mathcal{G}$ that are inhabited by \effectsafe terms,
	each pair  $(\classof{t}{\mathcal{G}}, \ES(t))$,
	corresponds to one partition produced by $\statewalkequiv$'s refinement of $\equiv_\mathcal{G}$.
We use these pairs as DP states in \tiger.\tighten

Under $\statewalkequiv$, \effectful terms are equivalent
  exactly when the statewalks they define are substitutable
  for constructing new \effectsafe terms.


\begin{lemma}[Substitutability Within Effectful Term]
  Given
  \begin{itemize}[leftmargin=2em]
    \item \effectsafe \effectful terms $t_1 \statewalkequiv t_2$,
    \item their corresponding \statewalks $\sw_1 = \langle \argt, \ldots, t_1 \rangle$ and $\sw_2 = \langle \argt, \ldots, t_2 \rangle$,
    \item an extended statewalk $\langle \argt, \ldots, t_1, s \rangle$,
  \end{itemize}

  Then exists $s' \equiv_\mathcal{G} s$ with statewalk $\langle \argt, \ldots, t_2, s' \rangle$,
\end{lemma}

We construct $s'$ by substitution.
Since $s$ extends the \statewalk, $t_1$ is a child of $s$, and by
  \Cref{lem:one-effectful-child} it is the only \effectful one,
  so we substitute it with $t_2 \equiv_\mathcal{G} t_1$.
For each pure child $c_i$ of $s$, we substitute it with $c_i' \equiv_\mathcal{G} c_i$
  such that $\effectsafe(c_i', \sw_2)$, which is guaranteed to
  exist because $\sw_1$ and $\sw_2$ have the same \extractionset.
All children of $s'$ are \effectsafe under $\sw_2$ and
  $s'$ refers only to $t_2$ so $s'$ is \effectsafe.
Because $\mathcal{G}$ is closed under congruence,
  $s$ and $s'$ are in the same \eclass of $\mathcal{G}$,
  so $s' \equiv_\mathcal{G} s$.

\begin{lemma}[Substitutability Within Pure Term]
  Given
  \begin{itemize}[leftmargin=2em]
    \item \effectsafe \effectful terms $t_1 \statewalkequiv t_2$,
    \item their corresponding \statewalks $\sw_1 = \langle \argt, \ldots, t_1 \rangle$ and $\sw_2 = \langle \argt, \ldots, t_2 \rangle$,
    \item an $s$ such that $s$ is pure and $\effectsafe(s, \sw_1)$,
  \end{itemize}

  Then there exists $s' \equiv_\mathcal{G} s$ such that $\effectsafe(s', \sw_2)$.
\end{lemma}

This follows directly from the definition of \extractionset.
Because $s$ is pure and $\effectsafe(s, \sw_1)$,
  its \eclass $\classof{s}{\mathcal{G}}$ belongs to $\ES(\sw_1)$.
Since $t_1 \statewalkequiv t_2$ gives $\ES(\sw_1) = \ES(\sw_2)$,
  that \eclass also belongs to $\ES(\sw_2)$,
  which by definition means it contains a term $s'$ with $\effectsafe(s', \sw_2)$.
Both terms are \effectsafe and share an \eclass, so $s' \equiv_\mathcal{G} s$.
Note that the \effectful child of a pure $s$ need not be $t_1$:
  it may be any earlier element of $\sw_1$,
  as with $u2[0]$ in the example of \autoref{sec:effect-safety}.

With substitutability recovered,
	it suffices for \tiger to store
	a single \statewalk per equivalence class of $\statewalkequiv$.
For a specific \eclass,
  the number of \statewalk equivalence classes
	is at most $2^n$
	because each equivalence class must correspond to a different \extractionset
	and there are only $2^n$ possible \extractionsets.
Thus, the number of equivalence classes of $\statewalkequiv$ is at most $n \cdot 2^n$,
  and \tiger is guaranteed to terminate.




\subsection{\tiger Algorithm}

\begin{algorithm}[t]
\caption{\tiger \textsc{ExtractMap}: find the pure terms extractable under a given \statewalk}
\label{alg:tiger-extractmap}
\begin{algorithmic}[1]

\Function{ExtractMap}{$\sw$}
  \State \textbf{Input:} a valid statewalk, $\sw$
  \State \textbf{Output:} a mapping from each pure \eclass to a term of the \eclass that is \effectsafe with $\sw$.
  \State $E \gets \{\}$
  \For{$\ell \in \{t \ | \ t \in \mathcal{G} \land \text{pure}(t) \land \text{leaf}(t)\}$}
    \State $\textbf{if } E[\classof{\ell}{\mathcal{G}}] = \bot \textbf{ then } E[\classof{\ell}{\mathcal{G}}] \gets \ell$
  \EndFor

  \State $\text{fixed\_point} \gets \textbf{false}$
  \While{\textbf{not} fixed\_point}
    \State $\text{fixed\_point} \gets \textbf{true}$
    \For{$t = f(t_1, \ldots, t_n) \in \{ t \ | \ t \in \mathcal{G} \land \text{pure}(t) \} $}
      \State $ \textbf{if } E[\classof{t}{\mathcal{G}}] \neq \bot \textbf{ then continue}$ \Comment{Already have an \effectsafe representative.}
      \If{$ \forall\ i.\ E[\classof{t_i}{\mathcal{G}}] \neq \bot \lor (\exists \ e_i'.\ \classof{e_i'}{\mathcal{G}} = \classof{t_i}{\mathcal{G}} \land e_i' \in \sw) $}
        \State \Comment{Each pure $t_i$ can be substituted with a term that is \effectsafe under \sw;}
        \State \Comment{the effectful $t_i$ (if present) can be substituted by $e_i' \in \sw$.}
        \For{$t_i \in \{ t_1, \ldots, t_n \} $}
          \State $ t_i' \gets \textbf{if } \text{pure}(t_i) \textbf{ then } E[\classof{t_i}{\mathcal{G}}] \textbf{ else } e_i' $ \Comment{At most one $t_i$ is effectful.}
        \EndFor
        \State $E[\classof{t}{\mathcal{G}}] \gets f(t_1', \ldots, t_n')$
        \State $\text{fixed\_point} \gets \textbf{false}$
      \EndIf
    \EndFor
  \EndWhile
  \State \Return $E$
\EndFunction
\end{algorithmic}
\end{algorithm}

\begin{algorithm}[t]
\caption{\tiger \textsc{Extract} (entry): return an \effectsafe term for an \effectful \eclass}
\label{alg:tiger-extract}
\begin{algorithmic}[1]

\Function{Extract}{$c_r$}
  \State \textbf{Input:} an effectful \eclass $c_r$.
  \State \textbf{Output:} an \effectsafe term represented by $c_r$.

  \State $DP \gets \{\}$ \Comment{Store a \statewalk for each \effectful \eclass and extractable set}
  \State $e \gets \text{Keys}(\Call{ExtractMap}{\langle \Arg \rangle})$ \Comment{pure \eclasses extractable with $\langle \Arg \rangle $}
  \State $DP[\classof{\Arg}{\mathcal{G}}][e] \gets \langle \Arg \rangle$
  \State $\text{worklist} \gets \text{new queue}$ \Comment{Worklist stores (\effectful \eclass, extractable set) pairs}
  \State $\text{worklist}.\text{insert}((\classof{\Arg}{\mathcal{G}}, e))$
  \While{$\textbf{not } \text{worklist.empty}()$}
    \State $(c, es) \gets \text{worklist.pop}()$
    \State $\langle e_0, \ldots, e_k \rangle \gets DP[c][es]$ \Comment{$\classof{e_k}{\mathcal{G}} = c$}
    \State $\text{eff\_deps} \gets \{ s \ | \ s = g(s_1, \ldots, s_m) \in \mathcal{G} \land \neg\text{pure}(s) \land \exists\ j.\ \classof{s_j}{\mathcal{G}} = c \}$
    \For{$t\ = f(t_1, \ldots, t_n) \in \text{eff\_deps}$}
      \If{$ \forall\ i.\ \text{pure}(t_i) \implies \classof{t_i}{\mathcal{G}} \in es $} \Comment{Can extend $ \langle e_0, \ldots, e_k \rangle $ with $t$.}
        \State $M \gets \Call{ExtractMap}{\langle e_0, \ldots, e_k \rangle}$
        \For{$t_i \in \{ t_1, \ldots, t_n \} $}
          \State $t_i' \gets \textbf{if } \text{pure}(t_i) \textbf{ then } M[\classof{t_i}{\mathcal{G}}] \textbf{ else } e_k $ \Comment{$ \neg \text{pure}(t_i) \implies \classof{t_i}{\mathcal{G}} = c $}
        \EndFor
        \State $t' \gets f(t_1', \ldots, t_n')$ \Comment{$t'$ equivalent to $t$ and \effectsafe.}
        \State $\sw' \gets \langle e_0, \ldots, e_k, t' \rangle$
        \State $es' \gets \text{Keys}(\Call{ExtractMap}{\sw'})$
        \If{$DP[\classof{t}{\mathcal{G}}][es'] = \bot$}
          \State $DP[\classof{t}{\mathcal{G}}][es'] \gets \sw'$
          \State $\text{worklist.insert}((\classof{t}{\mathcal{G}}, es'))$
        \EndIf
      \EndIf

    \EndFor
  \EndWhile
  \State $ \textbf{if } DP[c_r] = \bot \textbf{ then return } \bot $
  \State $\text{Choose } es \text{ such that } DP[c_r][es] \neq \bot \text{ and cost is minimized.}$
  \State $\langle e_0,\ldots,e_k \rangle \gets DP[c_r][es]$
  \State \Return $e_k$ \Comment{the extracted term for \eclass $c_r$}
\EndFunction

\end{algorithmic}
\end{algorithm}

\autoref{alg:tiger-extractmap} and \autoref{alg:tiger-extract}
	show the two components of \tiger.
\autoref{alg:tiger-extract} is the entry point to the algorithm
	which takes an \egraph and an \effectful root \eclass
	and returns an \effectsafe extraction for the given \eclass.

$\textsc{ExtractMap}$ maps a \statewalk to its \extractionset.
It is very similar to a standard extraction algorithm that
  iteratively constructs extracted terms for pure \eclasses via bottom-up enumeration.
The key difference is that it only uses the members of the input \statewalk
	as \effectful subterms to construct new pure terms.
When there are multiple subterms in the \statewalk for a particular \effectful \eclass,
  any of them can be used as the representative.
By construction, the representative term for each pure \eclass is \effectsafe
  with respect to the provided \statewalk.

$\textsc{Extract}$ is the main extraction loop. It keeps a worklist of
  equivalence classes of $\statewalkequiv$, whose representative \statewalks
	can potentially be extended to construct new \statewalks.
The DP table keeps track of one \statewalk per equivalence class,
  indexed by an \effectful \eclass and an \extractionset for the \statewalk.
For each \statewalk processed by the worklist,
  we enumerate all possible extensions of the \statewalk (\effectful \eclasses that
  take the last term in the current \statewalk as input).
When a new \statewalk is constructed, its \extractionset is computed
  using $\textsc{ExtractMap}$.
If the \statewalk and its \extractionset form a new equivalence class under $\statewalkequiv$,
  the \statewalk is added to the worklist.
The algorithm terminates when all items in the worklist have been processed,
  at which point, all equivalence classes of $\statewalkequiv$ have been identified.
If there is an entry in the DP table for the root \eclass, it corresponds
  to an \effectsafe extraction.
Otherwise, no \effectsafe extraction exists.
Though not shown in the algorithm pseudocode,
	the implemented algorithm keeps the lowest cost representative \statewalk for each
  equivalence class and uses a priority
  queue for the worklist, as in Dijkstra's shortest path algorithm.

\subsection{\tiger Algorithm Correctness}
\label{sec:tiger-proof}
In this section, we give a proof sketch for
  the soundness and completeness of \tiger.
The full proof is in \autoref{sec:tiger-appendix}.

\myparagraph{$\textsc{ExtractMap}$ is sound}
  Every entry in the mapping is \effectsafe with $\sw$.
  All terms added to the mapping are \effectsafe by construction
    because \effectsafe terms with the same statewalk are composable.
  The algorithm leverages this composability to construct the mapping
    via bottom-up enumeration starting from the leaves.

\myparagraph{$\textsc{ExtractMap}$ is complete}
  If an \effectsafe term exists for a given \eclass with $\sw$,
    then the mapping will contain an entry for that \eclass.
  All \effectsafe terms must be composed from \effectsafe subterms with
    the same \statewalk,
    so bottom-up enumeration until fixed-point ensures
    that an \effectsafe term for every \eclass in the \extractionset of $\sw$
    will be found.
  Soundness and completeness are both proven by induction on the depth
    of the smallest term in each \eclass.
  

\myparagraph{$\textsc{Extract}$ is sound}
  Soundness is implied by the following inductive invariant:
  each entry in \tiger's DP table is a valid statewalk for the
  corresponding (\eclass, extractable set).
  Since DP entries are never changed after they are set,
    and since the DP table is initialized with a valid
      \statewalk $\langle \Arg \rangle$,
    it suffices to show that for each loop iteration that
    assigns to $DP$, $\sw'$ is a valid \statewalk by construction
    that extends a \statewalk that was already present in $DP$.
  We prove the invariant by induction on the depth of the smallest term in each \eclass.

\myparagraph{$\textsc{Extract}$ is complete}
  Completeness is implied by the following inductive invariant:
  the (\eclass, extractable set) pair corresponding to every valid \statewalk
  is eventually added to the worklist.
  The (\eclass, extractable set) pairs in the worklist
    correspond to the equivalence classes of $\statewalkequiv$.
  Storing a single \statewalk for each pair suffices because 
    of substitutability under equivalent \statewalks.
  Every valid \statewalk can be built via some sequence of
    statewalks processed by the worklist.
  $\textsc{Extract}$ builds statewalks by bottom-up enumeration
    until fixed-point, ensuring that the equivalence class of every
    valid \statewalk is represented in the DP table.
  We prove the invariant by induction on the \statewalk.\tighten

\subsection{Complexity}
\label{sec:complexity}

As mentioned above,
  \tiger terminates because the number of equivalence classes of $\statewalkequiv$ is finite.
Despite being exponential in the worst case
  (as in the construction in \autoref{sec:np}),
  we observe in practice that the number of equivalence classes of $\statewalkequiv$
  is generally small.
We capture this insight with the parameter \statewalkwidth.

Informally, \statewalkwidth measures the maximum number of
	equivalence classes in $\statewalkequiv$
  corresponding to a single equivalence class of $\equiv_\mathcal{G}$.
This is equivalent to the maximum number of different \extractionsets for terms
 in one equivalence class of $\equiv_{\mathcal{G}}$.
More formally, we have:
\begin{definition}[Statewalk Width]
  $\swwidth{\equiv_\mathcal{G}} = \max_{\text{\effectful}(t)}{\bigl| \{{\ES(t')} \mid \; t' \equiv_\mathcal{G} t\} \bigr|}$
\end{definition}

By this definition, we obtain the exponential upper bound:
\begin{lemma}[Upperbound on statewalk width]
  $\swwidth{\equiv_\mathcal{G}} \le 2^n$, where $n$ is the number of equivalence classes of $\equiv_\mathcal{G}$
\end{lemma}

Let $k = \swwidth{\equiv_\mathcal{G}}$
  and $n$ be the number of representative terms in $\mathcal{G}$,
  the complexity for \tiger finding any \effectsafe extraction is $O(kn^2)$.
If we augment the algorithm to minimize cost,
  the complexity grows by a logarithmic factor.

The complexity of processing each representative term in a
  topological order
  exactly once is $O(n)$, resembling
  a standard effect-unaware extraction algorithm.
Similarly, the complexity of the \autoref{alg:tiger-extractmap} is $O(n)$,
	traversing $\mathcal{G}$ once.

In \autoref{alg:tiger-extract},
	each \effectful representative term of $\mathcal{G}$
	is processed at most $\swwidth{\equiv_\mathcal{G}}$ times.
This is because the representative term has only one
  \effectful direct child, whose \eclass
  can correspond to at most $\swwidth{\equiv_\mathcal{G}}$
  many equivalence classes of $\statewalkequiv$.
So the number of calls to \autoref{alg:tiger-extractmap} is bounded by
  $O(kn)$ over all representative terms.
Thus, the complexity of \tiger is $O(kn^2)$.

As we will see in \autoref{sec:eval},
  statewalk width is low in practice.
We hypothesize that optimization rules
  tend to preserve extractable sets
  across different statewalks.
For example, a rule rewriting only pure terms cannot increase \statewalkwidth, covering many
  peephole optimizations.
For rules that do rewrite effectful terms,
  only rarely does an application drop information
  (i.e., extractable \eclasses) between statewalks.\tighten

\section{Implementation}
\label{sec:impl}

This section discusses the implementation of \tiger, including
	performance optimizations and an extension to support control flow.
Finally, we describe \eggcc, an \egraph-based optimizer
  for effectful programs with control flow, in which we instantiate
  \tiger.

\subsection{Control Flow and Regions}
\label{sec:controlflows}


The algorithm in \autoref{sec:tiger} operates over a general dataflow IR, 
 which does not support control flow.
To handle control flow, 
 we extend \tiger to work with the Regionalized Value State
  Dependence Graph (\rvsdg) representation~\cite{rvsdg}.
\rvsdgs represent control flow using \emph{regions},
  which are subgraphs with well-defined inputs and outputs.
Each region defines a scope for its contained nodes.
Nodes only refer to other nodes in the same region and never refer
	to nodes in other regions directly.
A node in an \rvsdg can be simple (e.g., those in \autoref{sec:formalizing-rvsdgs})
 or structural.
Structural nodes contain sub-regions.
Loops are structural nodes with the loop body as a sub-region, and
conditionals are structural nodes with sub-regions for each branch.
This gives a hierarchical representation of control flow as nested regions.

Rewriting runs over the whole \egraph,
	in which a region is an ordinary node:
	an argument of an \texttt{If} or \texttt{DoWhile}.
To extract a program, however,
  we decompose the \egraph into \emph{regionalized} \egraphs,
	each representing terms in a single region.
This step is necessary because \tiger assumes
	all \effectful terms have a single \effectful child,
	but branching constructs have multiple \effectful children (e.g., the then and else branches of an if term).
Each regionalized \egraph is constructed by replacing
  subterms of a representative term referring to regions (such as a loop body)
  with symbolic placeholders.
To extract, we first run \tiger over the regionalized \egraph that contains the 
  target \eclass, which returns an extraction that may contain placeholders.
For each placeholder, we recursively extract from a regionalized \egraph 
	rooted at the placeholder's \eclass.
The full extraction for the root \eclass is reconstructed from the extractions
	for each sub-region by substituting the extracted term for each placeholder.
\autoref{app:regionalized} illustrates the decomposition on a small example.\tighten

Regionalized \egraphs break extraction into smaller subproblems, but a single
	regionalized \egraph often still contains tens of thousands of terms, and
	\autoref{sec:eval} shows this is not enough to make ILP-based extraction
	tractable for \eggcc's \egraphs.

\subsection{\tiger Optimizations}
\label{sec:tiger-optimizations}

To improve performance,
	our implementation of \tiger includes several optimizations.
We briefly explain the most impactful ones.

\myparagraph{Pruning regionalized \egraphs}
We remove terms that are not reachable from the root \eclass
  and uninhabited \eclasses from the regionalized \egraph
  to reduce its size.

\myparagraph{Implicit terms}
We do not explicitly construct and save the terms in the main loop of \tiger.
Instead, we mark the existence of a term and maintain hints for later reconstruction.
We construct the terms explicitly only
	after confirming that an \effectsafe term from the root \eclass exists
	and that no lower-cost variant was discovered.

\myparagraph{Memoization}
We cache calls to $\textsc{ExtractMap}$
	to avoid recomputing the same query.
We improve cache hit rate by leveraging
	the fact that query results
	represented by implicit terms
        are determined by the set
	of \eclasses in the \statewalk.

\myparagraph{Reusing partial results with persistent B-tree}
$\textsc{ExtractMap}$ is called on \statewalks that grow by a
	single term in each subsequent call, so results for the prefix
	of each \statewalk can be reused from previous calls.
Instead of copying the entire state, we use a persistent B-tree to maintain these
	intermediate results.
As a result, each call to $\textsc{ExtractMap}$ only needs to perform
	updates that are introduced by the last term in the input \statewalk.


\myparagraph{Incremental hashing}
We use hashes to compare two \extractionsets quickly.
To avoid doing a full scan of the \extractionset,
	which would defeat the purpose of using a persistent B-tree,
  we use an incremental hash function inspired by bloom filters and hyper-dimensional computation.

\myparagraph{Liveness}
For each \statewalk, there is a fixed set of \eclasses
	containing \effectful terms that could potentially be used to extend the \statewalk.
If there are no terms in this set of \eclasses
	that depend on a pure \eclass
	in the \statewalk's \extractionset,
	then the pure \eclass can be dropped from the \extractionset,
	reducing
	the number of equivalence classes explored by \tiger.
  
\myparagraph{Associativity-commutativity (AC) mitigation}
We observe occasional exponential blow up of \statewalkwidth
	in our evaluation,
	which arises from an interaction between an optimization that
	collapses the ordering of read operations and Bril's implementation
	of multi-dimensional arrays,
  which leads to many densely nested load operations.
Because load operations do not change the state value,
	every partial permutation of loads can lead to a different equivalence class in $\statewalkequiv$.
This resembles the classic Associativity-Commutativity (AC) problem of \egraphs~\cite{ac1, ac2}.
In our implementation, we use sound heuristics to identify
	\eclasses that might cause this blow up and use a fixed order of
	exploring extensions to \statewalks from these \eclasses.
This performance optimization drastically reduces the number of equivalence classes \tiger explores.


\subsection{Cost Model}
\label{sec:cost-model}
Returning the original program is a trivial but undesirable solution to the extraction problem.
To find a term that corresponds to 
	an optimized version of the input program,
	we use a linear cost model that assigns a fixed constant cost to each function symbol.
The cost of a term is the sum of the costs of all function symbols it contains.
This cost model is used to pick the extraction used for each \eclass.
The cost for each sub-region in a regionalized \egraph is estimated
	using a traditional greedy extractor that does not guarantee \effectsafety.

Like traditional greedy extractors, \tiger is not guaranteed to
	find the globally optimal term under the cost model because
	it greedily makes locally optimal decisions.
However, it does not lose more optimality than a standard greedy extractor.
That is, if a greedy extractor that does not consider effects
	happens to find an \effectsafe term, \tiger's solution will be at least as good:
\tiger considers strictly more potential extractions per \eclass than
	a traditional greedy extractor (since it tracks one extraction per
	equivalence class under $\statewalkequiv$, which corresponds to a set of
	extractions per \eclass).
In practice, \tiger produces good extractions (\autoref{sec:eval}).


\subsection{The \eggcc Program Optimizer}

We implement \tiger in \eggcc,
 a prototype \egraph-based optimizer for effectful programs.
\eggcc takes a program written in a simple LLVM-like IR called \bril~\cite{bril},
 converts it to an \rvsdg~\cite{cfgToRvsdg}, and
 optimizes the \rvsdg in \egglog~\cite{egglog}.
We have implemented over a dozen advanced optimizations using declarative rewrite rules in \egglog
	that introduce complex interactions between effectful operations,
  making \eggcc a challenging setting for \effectsafe extraction.
For example, \eggcc performs redundant load elimination, loop invariant code motion,
	redundant branch elimination, dead loop elimination, and conditional invariant code motion.

After applying optimizations in \egglog,
 \eggcc extracts an \effectsafe \rvsdg from the \egraph,
 and converts the optimized \rvsdg
  back to \bril.
During the conversion back to \bril,
  \eggcc applies
  simple cleanup optimizations like eliminating unreachable branches~\cite{cfgToRvsdg}.
In \autoref{sec:eval}, we show \eggcc produces programs with similar quality to LLVM.




\section{Evaluation}
\label{sec:eval}

Our evaluation answers three research questions.
\begin{enumerate}
  \item[\textbf{RQ1}] How does the performance of \tiger compare to ILP-based extraction? (\autoref{subsec:tigerilpextraction})
  \item[\textbf{RQ2}] How does the runtime of extracted programs produced by \tiger compare to those produced by ILP? (\autoref{subsec:outputquality})
  \item[\textbf{RQ3}] Can \tiger enable domain-specific effectful optimizations? (\autoref{sec:Fenwick})
\end{enumerate}

We show that \tiger is \GeometricMeanTigerSpeedupVsGurobiWithTimeouts
  times faster than \gurobi.\footnote{Speedup is computed as the
  geometric mean of the ratios of runtimes on each regionalized \egraph,
  with timeouts counted as 5 minutes.
For the few infeasible points,
  time taken to report infeasibility is used.}
We then show \tiger produces programs of similar quality
  to LLVM and also to ILP-based extraction.
Finally, we present a case study showing the benefit of
  domain-specific optimizations enabled by \tiger.
\textit{The key takeaway is that with \tiger,
  extraction is no longer the bottleneck
  in \eggcc's end-to-end optimization pipeline}.\\

\noindent{\textbf{Setup.}}
We directly compare the performance of \tiger and solver-based
  \effectsafe extraction algorithms by
  implementing multiple extraction techniques in \eggcc.
For ILP extraction, we compare to the state-of-the-art commercial
  solver \gurobi and the open-source solver \cbc.
The ILP encoding of the extraction problem uses the same
  cost model as \tiger.

We run \eggcc on
  the \bril benchmark suite
  and the \polybench benchmark suite~\cite{polybench} translated to \bril.
The \bril benchmark suite contains \NumbrilBenchmarks programs spanning
  geometry, matrix operations, numerical methods, graph algorithms,
  and bitwise operations.
In total, we have \NumBenchmarksAllSuites benchmarks.

For each benchmark, we compare \tiger
  against an ILP-based extraction algorithm
  that guarantees \effectsafe extractions.
All experiments were run on a machine with an
  AMD EPYC 7702P CPU and 512 GB of DDR4 memory.
The many cores and large memory let us run
  multiple benchmarks in parallel,
  making the comparison to ILP-based extraction feasible.
Individual runs of \eggcc and runs of \tiger never exceed
  the limit of 16 GiB of memory.
We test the performance of \eggcc's binaries
  sequentially on a single core.


\paragraph{\textbf{ILP encoding}} For comparison to ILP,
  \eggcc encodes the extraction problem as an ILP instance
  and calls out to a (configurable) external ILP solver
  to find a low-cost solution,
  similar to prior work in Peggy~\cite{peggy}.
Each term $t$ in the \egraph is associated with a binary variable $p_t$,
  indicating whether $t$ is selected in the extraction.
For every candidate term $u$ for the subchild at position $i$ of term $t$,
  a binary variable $s_{t, i, u}$ indicates whether $u$ is selected
  as the $i$-th child of $t$.
Constraints are added to ensure the variables form a single
  term represented by the root \eclass
  and that the extraction is \effectsafe.
For full details of the encoding, see \autoref{app:ilp-details}.

\subsection{\tiger Compared to ILP-based Extraction (\textbf{RQ1})}
\label{subsec:tigerilpextraction}

We show that \tiger solves
  the bottleneck of \effectsafe extraction.
We investigate the effect of
  factors including individual regionalized \egraphs,
  a large program, ILP encoding size, and \statewalk width.
In every case, \tiger outperforms ILP-based extraction,
  which is infeasible in practice.

Across all benchmarks, \eggcc spends an average of
  \overalleggcctigerOzeroOzeroEggccNonExtractionTime seconds outside
  of extraction, with an average of \overalleggcctigerOzeroOzeroEggccExtractionTime seconds spent in its extraction phase\footnote{\eggcc invokes \tiger as a separate process, so the measured extraction phase also includes transferring the serialized \egraph across the process boundary, parsing it, and constructing the regionalized \egraphs. \tiger's own search accounts for under 3\% of this phase on average.}.
\tiger never times out.
For comparison, \llvm at \texttt{-O3} takes an average of
  \overallllvmOthreeOzeroLLVMCompileTime seconds
  to compile the same programs.
That \eggcc is slower is expected: it explores a space of equivalent programs
  rather than applying passes in a fixed order, which costs compile time but is
  what makes optimizations like those of \autoref{sec:Fenwick} expressible.
\begin{wrapfigure}{r}{0.5\textwidth}
    \centering
    \includegraphics[width=\linewidth]{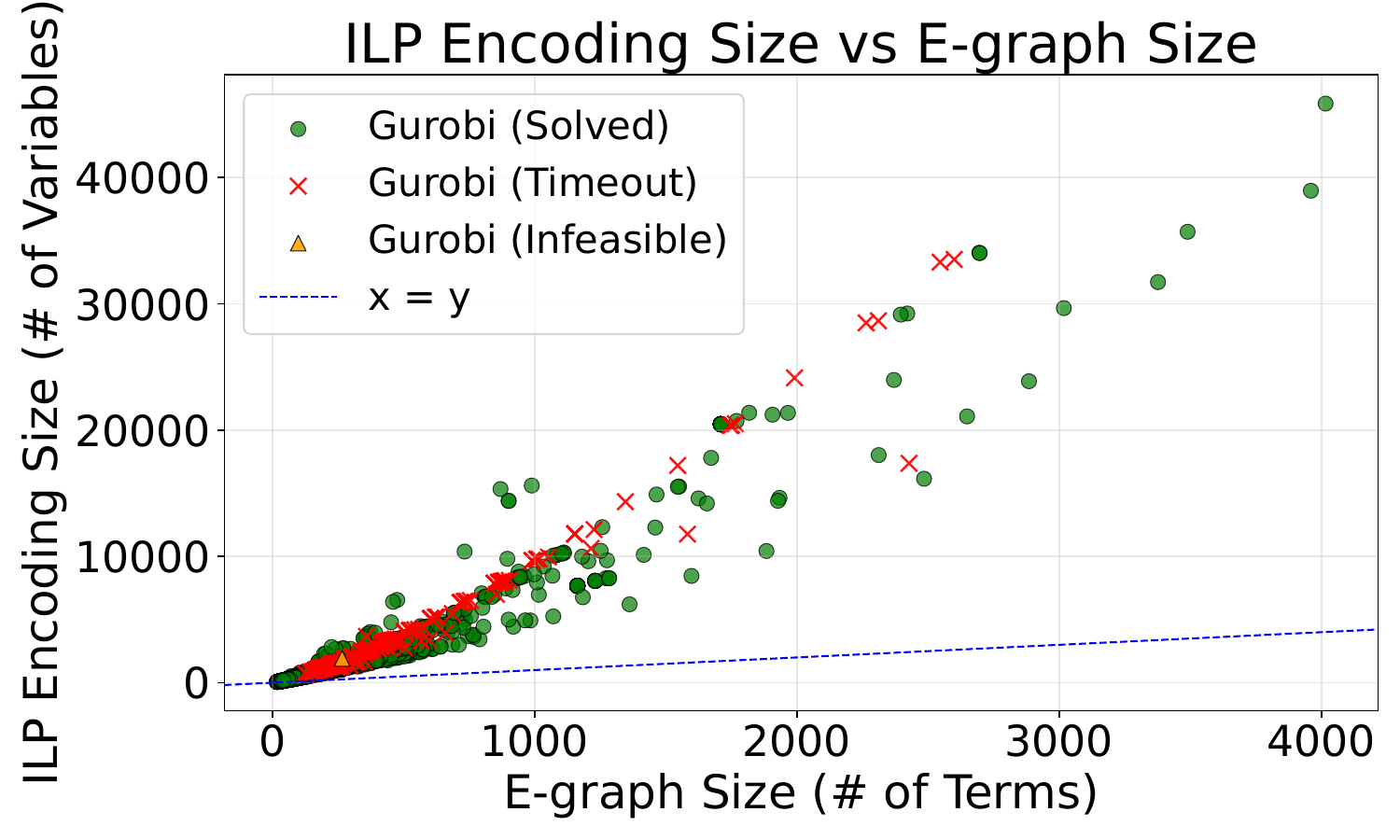}
    \caption{Scatter plot showing ILP encoding size compared to \egraph size.
    Encoding size is measured in number of variables, which also bounds the number of constraints.
    \Egraph size is measured in number of representative terms.}
    \Description{A scatter plot with e-graph size in number of terms on the
      horizontal axis and ILP encoding size in number of variables on the
      vertical axis. The points rise roughly in proportion to e-graph size, with
      substantial spread. Separate markers distinguish the instances Gurobi
      solved, those it timed out on, and those it reported infeasible, and the
      timeouts are scattered across the whole range rather than concentrated at
      the large end. A dashed reference line for x equals y runs far below the
      data.}
    \label{fig:ilp-encoding-size-scatter}
\end{wrapfigure}

On the other hand, ILP-based extraction is intractable
  for many benchmarks.
On the \polybench benchmark suite, both \gurobi and \cbc
  \textit{time out on all \NumpolybenchBenchmarks benchmarks after
  5 minutes} for extracting a single regionalized \egraph.
A longer timeout is not feasible because
  benchmarks in the \polybench suite produce an average of \AvgPolybenchRegionalizedEgraphsPerBenchmark regionalized \egraphs.
In the \bril benchmark suite,
  \gurobi times out on \NumeggcctigerILPGurobiRegionTimeoutBenchmarksBril benchmarks.
On benchmarks where \gurobi does not time out,
  \eggcc spends an average of \AvgEggcctigerOZeroOZeroExtractionTimeSecsOnILPGurobiSolvedBenchmarks seconds running \tiger
  compared to \overalleggcctigerILPOzeroOzeroEggccExtractionTime seconds running ILP-based extraction.
Timeouts in \gurobi's solving do not seem to correlate with
  the size of the \egraph or the size of the resulting ILP encoding.
\autoref{fig:ilp-encoding-size-scatter} shows the relationship
  between \egraph size and ILP encoding size.

\begin{figure}
  \includegraphics[width=0.62\linewidth]{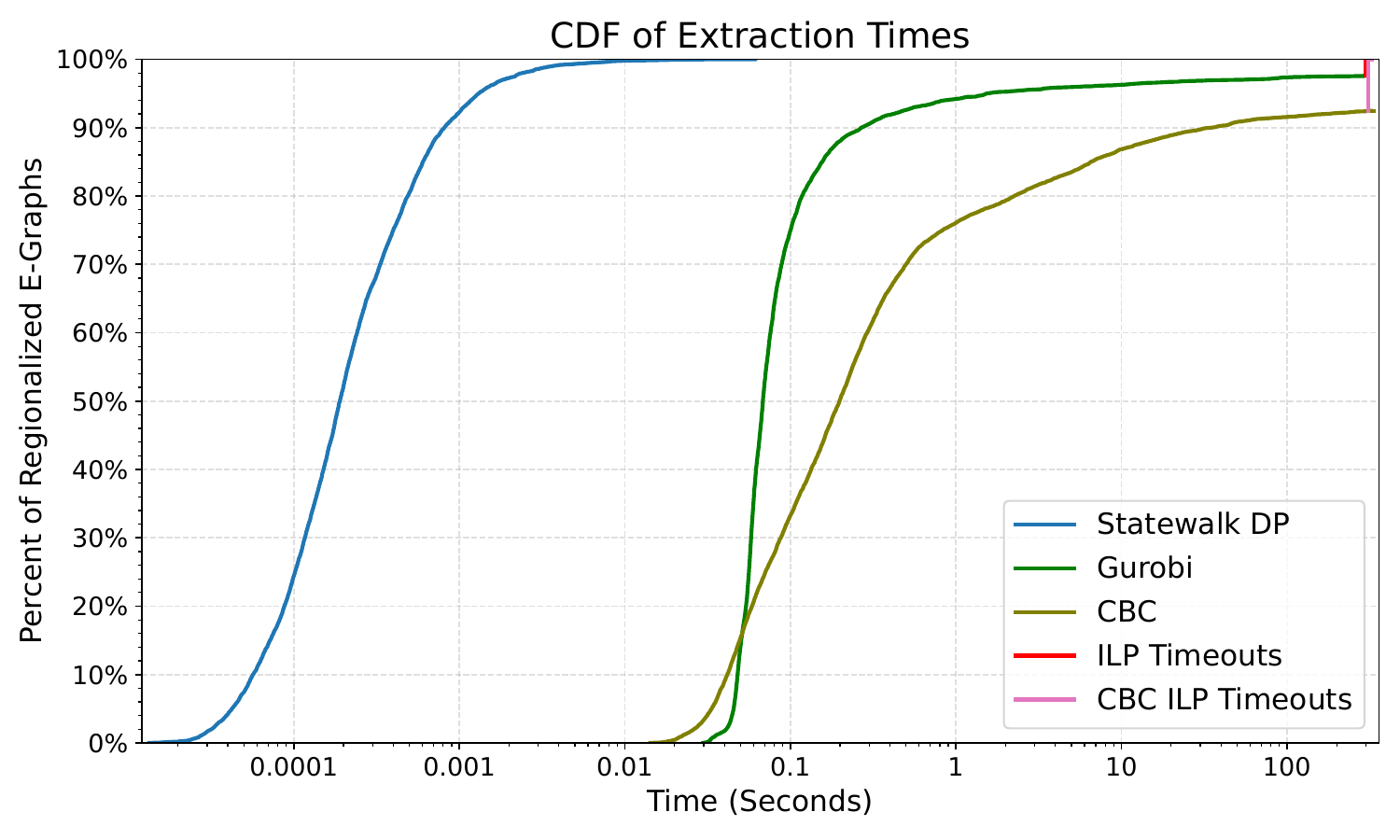}
  \caption{
    Cumulative distribution function of extraction times for \tiger and
    ILP-based extraction, over all regionalized \egraphs across all
    benchmarks (including the raytracer).
    The x-axis is on a \textbf{log scale}.
    The vertical lines at the right are timeouts and infeasible extractions,
    which occur only for the ILP baselines:
    \tiger is complete (\autoref{sec:tiger-proof}) and
    returns an \effectsafe extraction for every regionalized \egraph here.
    \autoref{app:ilp-details} explains why the ILP encoding
    can report infeasibility.
  }
  \Description{A cumulative distribution plot. The horizontal axis is extraction
    time in seconds on a logarithmic scale; the vertical axis is the percent of
    regionalized e-graphs extracted within that time. The statewalk DP curve
    rises steeply between roughly a hundredth of a millisecond and one
    millisecond and reaches 100 percent well before the others begin. The Gurobi
    and CBC curves rise only after about a hundredth of a second and flatten
    below 100 percent, with vertical lines at the far right marking their
    timeouts and infeasible results.}
  \label{fig:extraction-time-cdf}
\end{figure}

The number of variables in the encoding
  is \GeometricMeanILPEncodingVarsPerEgraphSize times the size of the \egraph on average (geometric mean),
  with a maximum of \MaxILPEncodingVarsPerEgraphSize times.
This shows that despite a reasonable ILP
  encoding scaling with problem size,
  ILP-based extraction remains unreliable.

\paragraph{\textbf{Region-level analysis}}
We perform a finer-grained comparison of \tiger and ILP-based extraction
  by measuring their runtimes on all \NumRegionalizedEgraphs
  regionalized \egraphs produced by \eggcc
  across all benchmarks.
ILP-based extraction takes an average of \AvgILPGurobiRegionExtractTimeSecs
  seconds per regionalized \egraph,
  and times out on \NumILPGurobiRegionTimeouts after
  5 minutes.
\tiger extraction takes an average of \AvgTigerLiveOnSatelliteOnRegionExtractTimeSecs
  seconds per regionalized \egraph
  and spends a maximum of \MaxTigerLiveOnSatelliteOnRegionExtractTimeSecs seconds on any single regionalized \egraph.
\autoref{fig:extraction-time-cdf} shows the distribution of extraction times for
  \tiger and ILP-based extraction
  on all the regionalized \egraphs.
A point $(t, p)$ means a fraction $p$ of extractions
  finished within time $t$.
\tiger completes $92\%$ in under a millisecond and all of them within
  \MaxTigerLiveOnSatelliteOnRegionExtractTimeSecs seconds;
  at $0.1$ seconds \gurobi has completed $75\%$ and \cbc $33\%$.
The ILP baselines exhaust the 5-minute timeout
  on a tail of difficult extractions.

The experiments in this section run with the same
  cost model for \tiger and ILP-based extraction.
The ILP-based extraction finds an optimal solution (modulo the incompleteness of the encoding).
However, even when disregarding cost completely,
  we find that \gurobi still times out on
  \NumeggcctigerILPNOMINRegionTimeoutBenchmarks benchmarks and remains
  a bottleneck, spending an average of \overalleggcctigerILPNOMINOzeroOzeroEggccExtractionTime seconds
  running ILP-based extraction across all benchmarks.
This suggests that the bottleneck of ILP-based extraction is not optimizing
  cost, but rather finding any \effectsafe extraction at all.

\begin{wrapfigure}{r}{0.4\textwidth}
  \centering
  \includegraphics[width=\linewidth]{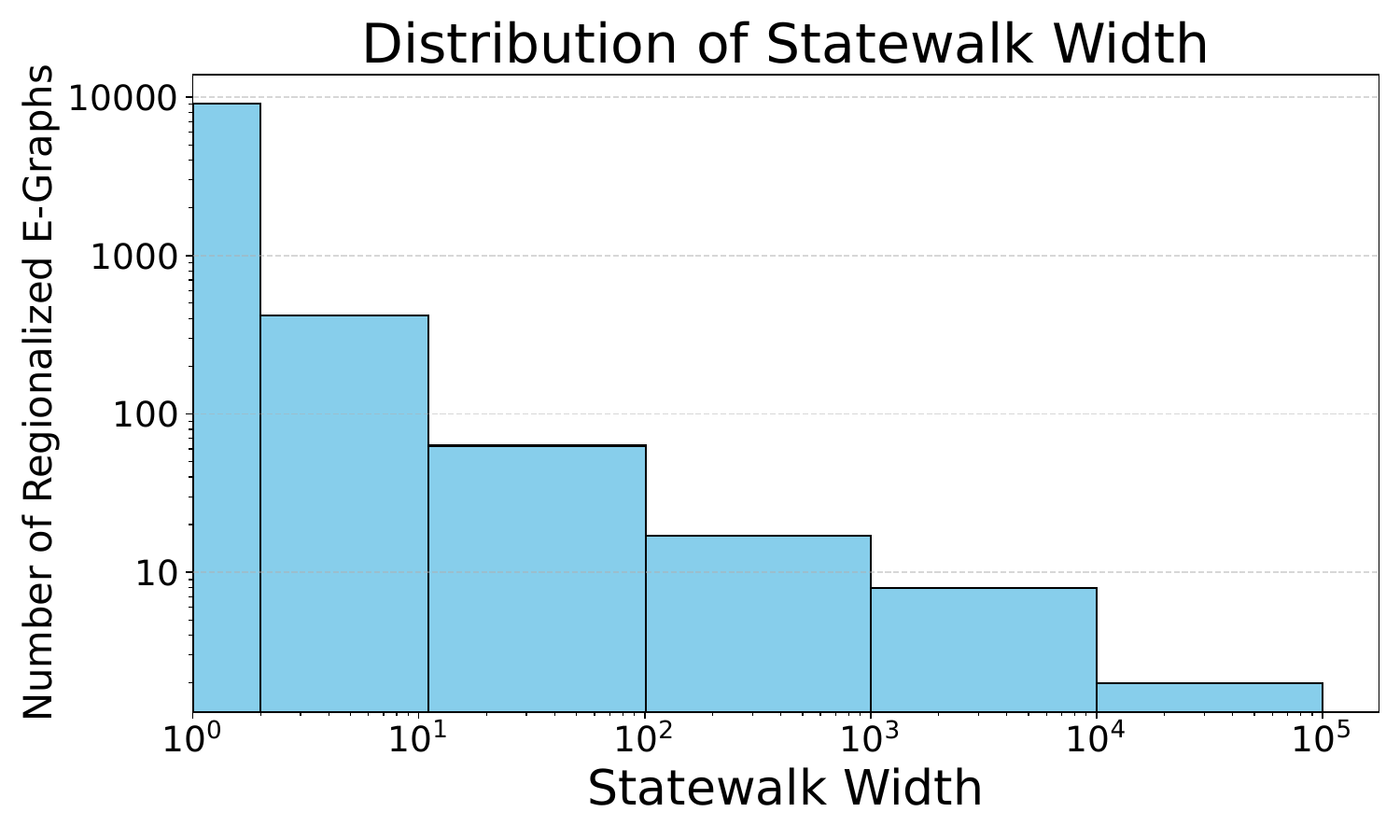}
  \caption{
 The distribution of statewalk width in the \bril benchmark suite and the \polybench benchmark suite.
 \PercentRegionsStatewalkWidthUnderTwo of regions have a \statewalk width of 1,
  and \PercentRegionsStatewalkWidthUnderSix have a \statewalk width of 5 or less.
  }
  \Description{A histogram of statewalk width, with statewalk width on the
    horizontal axis and the number of regionalized e-graphs on the vertical axis,
    both logarithmic. The leftmost bar, at width one, is orders of magnitude
    taller than any other, and the bars fall away steeply as width increases,
    with a short bar at the far right for the largest widths observed.}
  \label{fig:statewalk-width}
\end{wrapfigure}

\paragraph{\textbf{Case study on a large program}}
We also run \eggcc on a raytracer compiled from
  821 lines of code written in a subset of Rust to 2503 lines of \bril code.
\eggcc spends \raytraceOveralleggcctigerOzeroOzeroEggccNonExtractionTime
  seconds performing equality saturation on the raytracer
  and \raytraceOveralleggcctigerOzeroOzeroEggccExtractionTime seconds
  in its extraction phase.
Of that phase, \tiger's own search accounts for
  only $0.598$ seconds,
  summed over all \NumRaytraceRegionalizedEgraphs regionalized \egraphs
  that \eggcc produces while optimizing the raytracer;
  the remainder is the serialization and process overhead noted above.
The raytracer produces large regionalized \egraphs,
  totaling 158,593 terms after pruning (\autoref{sec:tiger-optimizations}),
  with a maximum of 4,016 terms in a single regionalized \egraph.
Despite the large size, \tiger still takes a maximum of
  \MaxRaytraceTigerExtractionTimeSecs seconds
  to extract from any single regionalized \egraph.
ILP-based extraction, in contrast,
  times out on \NumRaytraceILPRegionalizedEgraphTimeouts
  of the \NumRaytraceRegionalizedEgraphs regionalized \egraphs after 5 minutes.

\paragraph{\textbf{\tiger and ILP Extraction Time Scaling}}


\begin{figure}
  \centering
  \includegraphics[height=0.232\linewidth]{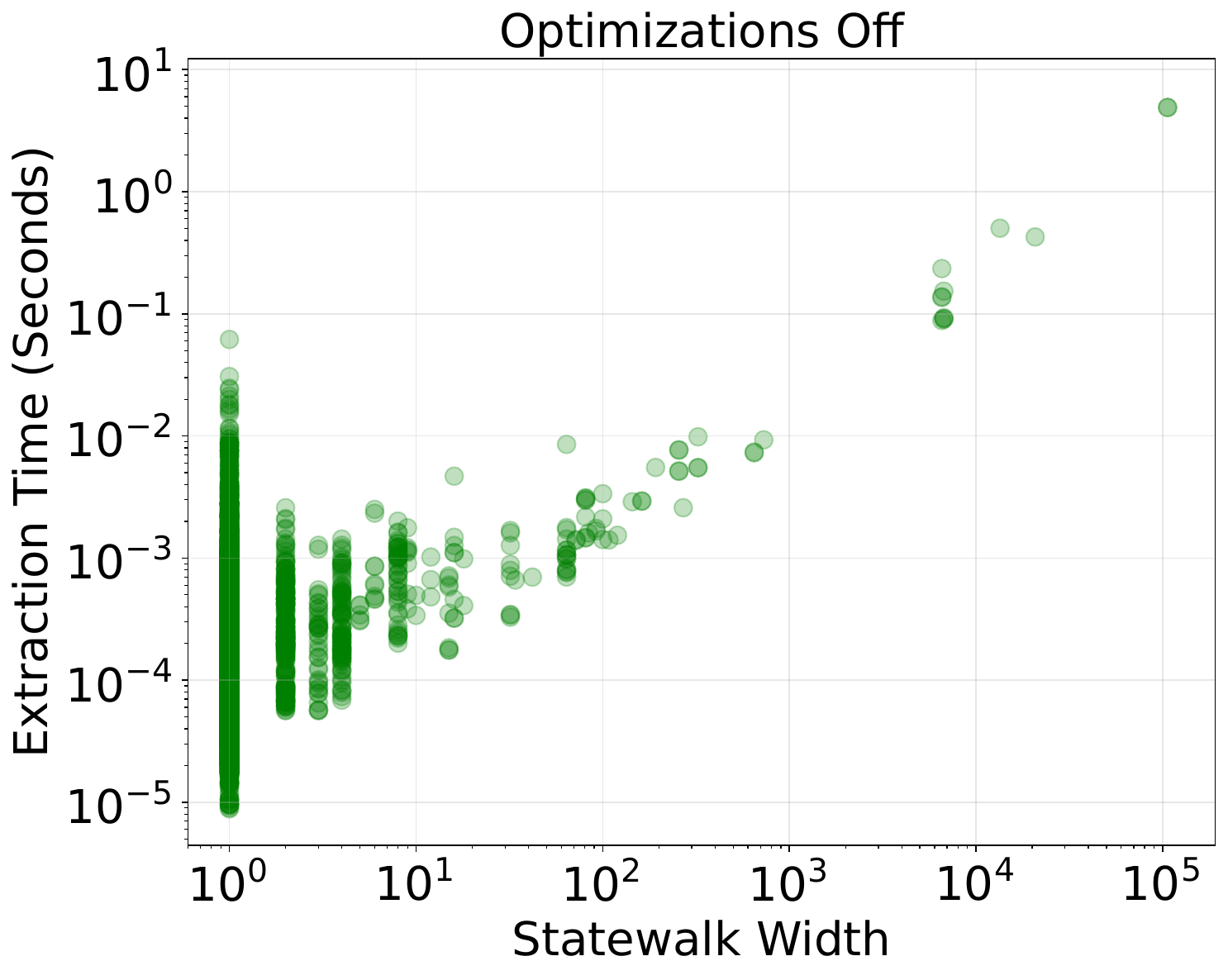}\hfill
  \includegraphics[height=0.232\linewidth]{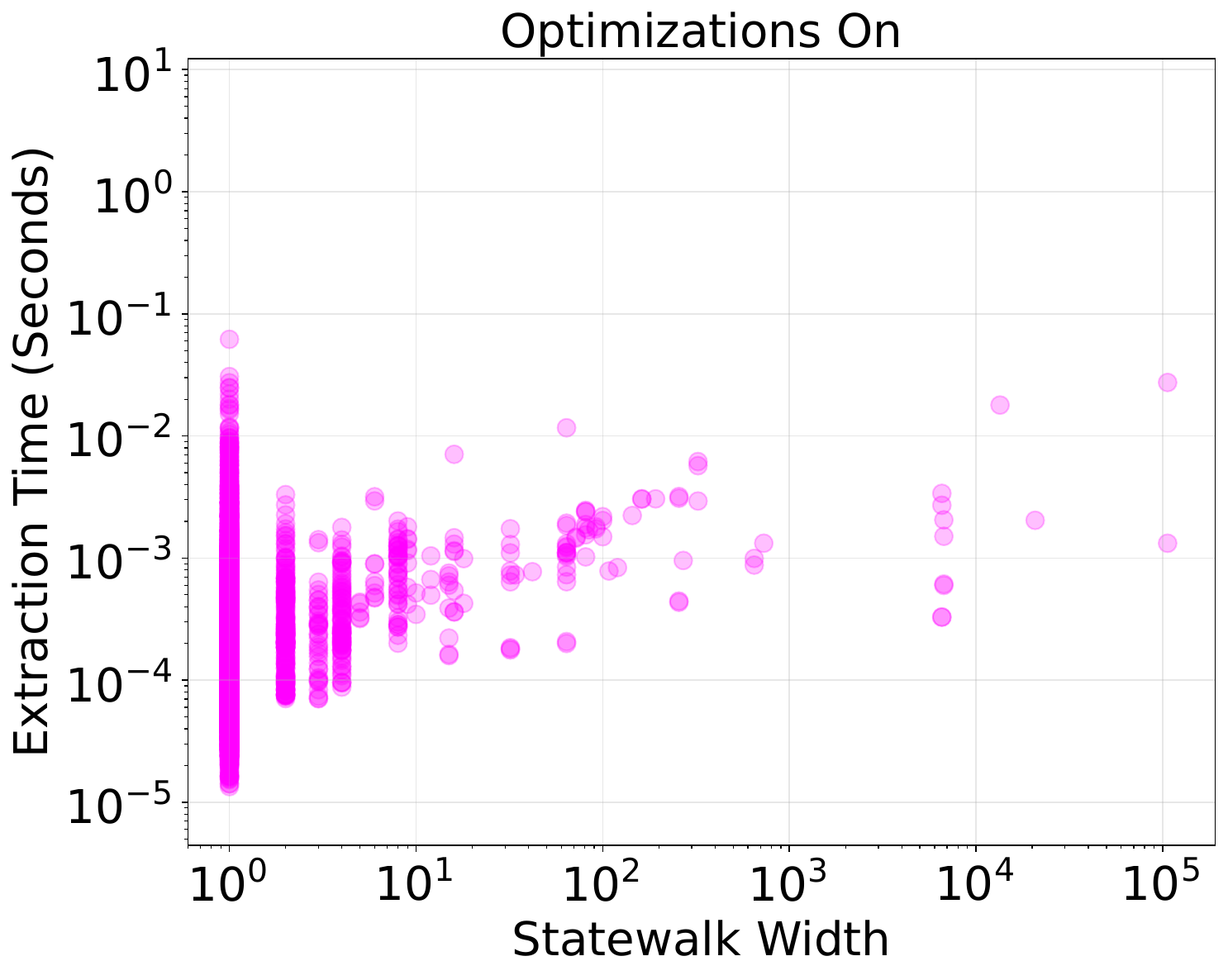}\hfill
  \includegraphics[height=0.232\linewidth]{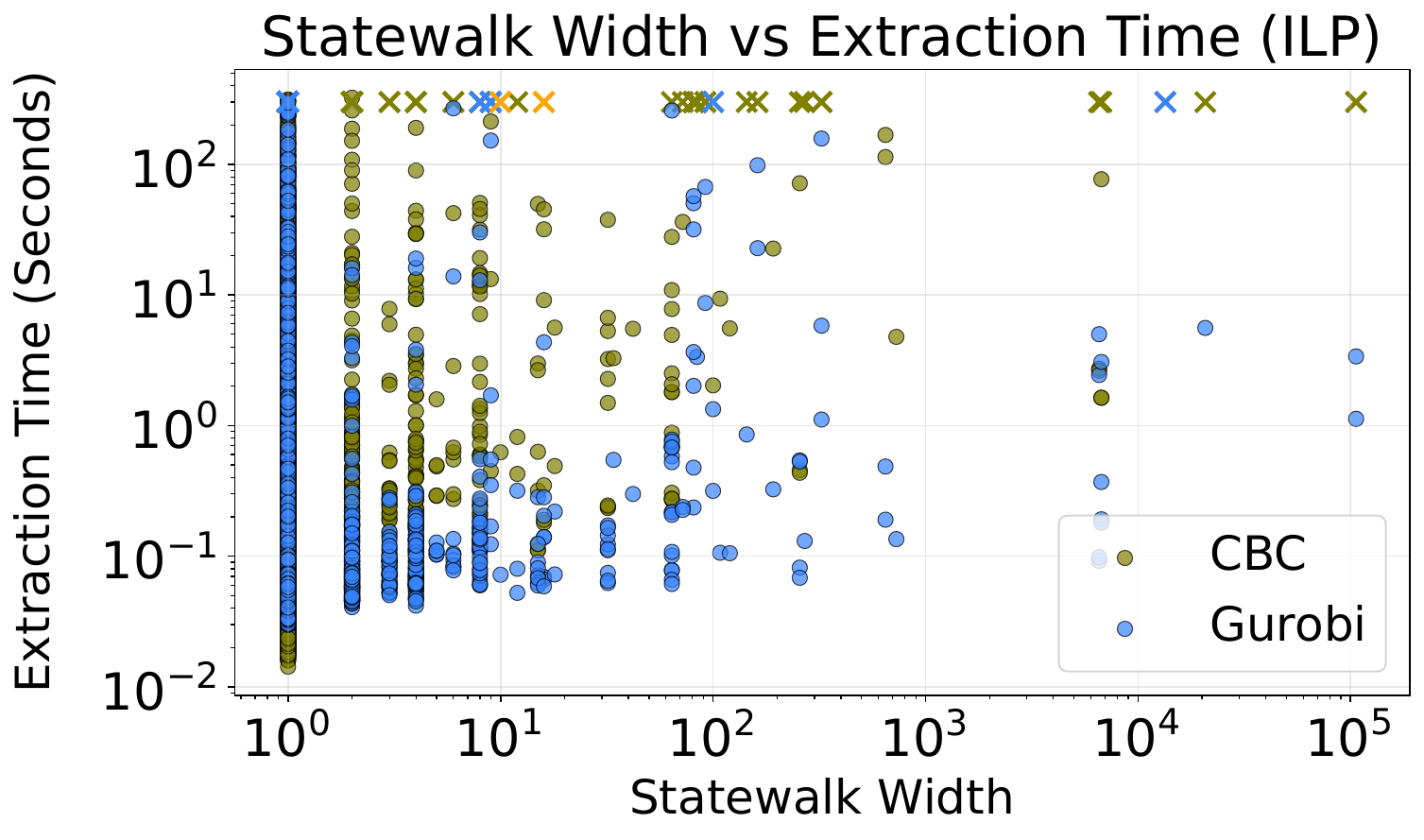}

  \caption{
    Extraction time against the \statewalkwidth of the regionalized \egraph,
    for \tiger without its optimizations (left),
    \tiger with them (middle), and
    ILP-based extraction (right).
    \tiger's optimizations are liveness analysis and AC mitigation.
    \textbf{All axes are logarithmic}; on linear axes almost every point
    would land on top of another.
    The two \tiger panels share axis limits, so they are directly comparable.
    Without optimizations, extraction time grows with \statewalkwidth,
    reaching $4.9$ seconds at the largest widths; with them it stays below
    \MaxTigerLiveOnSatelliteOnRegionExtractTimeSecs seconds throughout,
    leaving the upper right empty.
    ILP-based extraction shows no such trend but times out
    (crosses) at every \statewalkwidth.
    The dense column at the left of each panel is the
    \PercentRegionsStatewalkWidthUnderTwo of regions with \statewalkwidth 1.
  }
  \Description{Three scatter plots in a row, all with statewalk width on the
    horizontal axis and extraction time on the vertical axis, both logarithmic.
    The left panel, statewalk DP with its optimizations off, shows points
    trending upward from left to right, reaching several seconds at the largest
    widths. The middle panel, with the optimizations on, shares the same axes
    but its points stay in a flat band along the bottom, leaving the upper right
    of the panel empty. The right panel, ILP-based extraction, shows no upward
    trend but has crosses marking timeouts spread across every statewalk width.
    All three panels have a dense column of points at statewalk width one.}
  \label{fig:statewalk-runtime}
\end{figure}

We study how the extraction time of \tiger and ILP-based extraction
  changes with respect to \statewalkwidth.
\autoref{fig:statewalk-width} shows the empirical measurements of statewalk width in the \egraph
 produced by \eggcc on the \bril benchmark suite and \polybench benchmark suite.
\PercentRegionsStatewalkWidthUnderTwo of regions have a \statewalk width of 1,
  while \PercentRegionsStatewalkWidthUnderSix have a \statewalk width of 5 or less.
Two outliers have a \statewalkwidth of \MaxStatewalkWidthAllBenchmarks.
These outliers come from the \polybench ``heat-3d'' benchmark,
  which computes heat diffusion in a large grid.
The statewalk width is large due to a series of memory accesses
  with no data dependencies between them, enabling associativity
  and commutativity between them.
Our Associativity-Commutativity (AC) mitigation optimization
  reduces this factor significantly, making the maximum
  size of the dynamic programming table for a single \eclass
  a manageable \StatewalkWidthHeatThreeDLiveOnSatelliteOn, despite
  the colossal statewalk width.\tighten

\autoref{fig:statewalk-runtime} shows how the extraction time of \tiger and ILP-based extraction
  scales with respect to \statewalk width.
The extraction time of ILP-based extraction does not seem to correlate with
  \statewalk width, while \tiger's extraction time without optimizations
  does.
With \tiger's optimizations enabled,
  it scales well
  even in the presence of large \statewalk widths.

\subsection{Output Quality: Runtime of Extracted Program (\textbf{RQ2})}
\label{subsec:outputquality}

A head-to-head comparison
  between \tiger and ILP-based extracted programs shows that
  \tiger's outputs are on par with ILP (when ILP is feasible).
We also show \eggcc's \egraphs are non-trivial,
  rich enough to produce significant speedups over the baseline,
  even compared to LLVM's optimizations.\tighten

\begin{figure}
  \centering
  \begin{minipage}[t]{0.72\linewidth}
    \begin{adjustbox}{width=\linewidth}
      \includegraphics{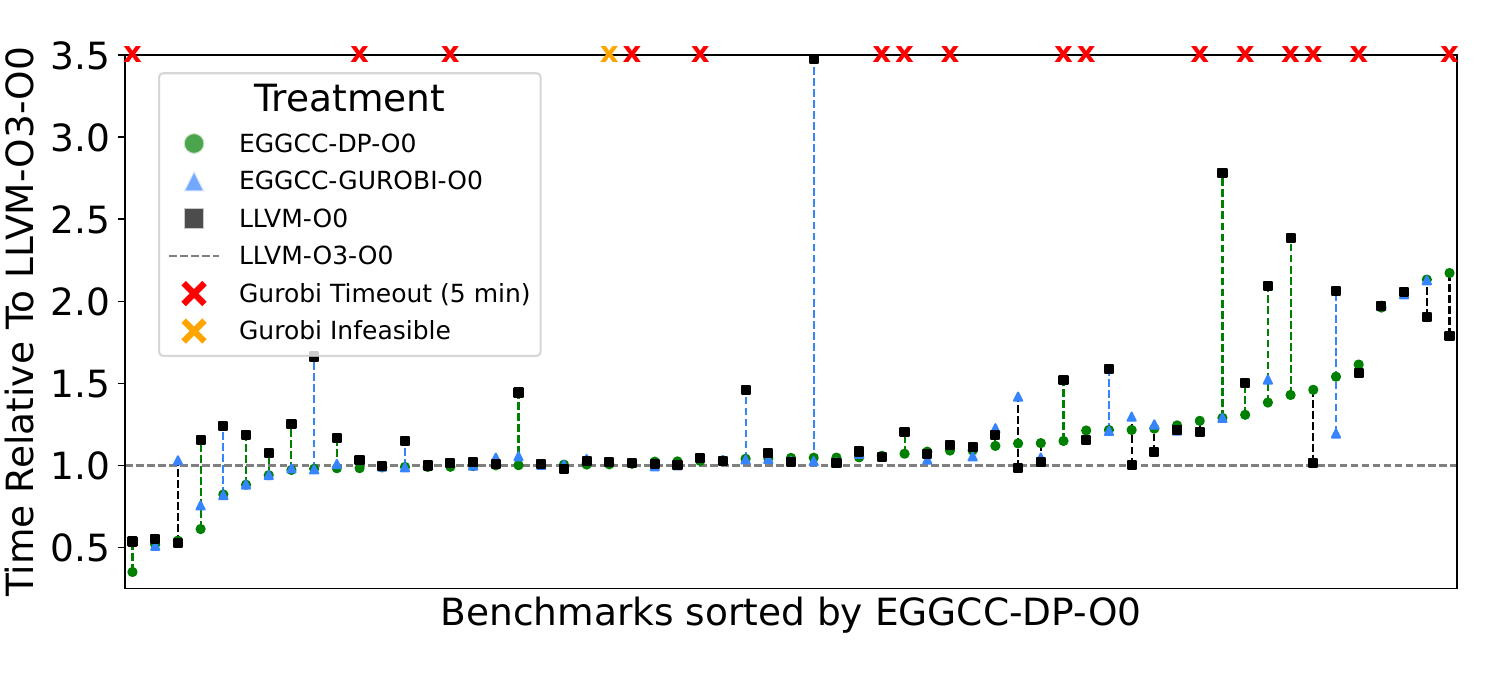}
    \end{adjustbox}
  \end{minipage}
  \begin{minipage}[t]{0.165\linewidth}
    \begin{adjustbox}{width=\linewidth}
      \includegraphics{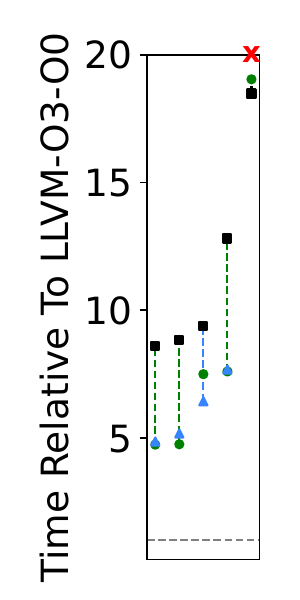}
  \end{adjustbox}
  \end{minipage}

  \caption{Comparison of \eggcc to \llvm on the \bril benchmark suite.
  Lower numbers are better.
  The y-axis shows the ratio between the average runtime and
  the average runtime of the baseline (LLVM-O3-O0). The grey line at 1.0 indicates the baseline.
  The vertical lines indicate the difference between the worst and best treatment for each benchmark.
  The graph is split into two parts due to 5 outlier benchmarks for which LLVM-O3-O0 gave large speedups.}
  \Description{A chart with one vertical interval per benchmark. The horizontal
    axis lists the Bril benchmarks sorted by the eggcc with statewalk DP result;
    the vertical axis is runtime relative to the LLVM-O3-O0 baseline, drawn as a
    grey line at 1.0. Each interval spans the worst to the best treatment for
    that benchmark, with separate markers for eggcc with statewalk DP, eggcc
    with Gurobi, and LLVM-O0, plus crosses where Gurobi timed out or was
    infeasible. Most benchmarks cluster near the baseline. The chart is split
    into two panels, the narrow right-hand one covering five outliers whose
    vertical axis extends much further.}
  \label{fig:eggcc-vs-llvm-bril}
\end{figure}


To measure the quality of \tiger's extracted programs,
  we compare four treatments across the \bril and \polybench benchmarks:
\begin{itemize}[leftmargin=1em]
  \item \texttt{LLVM-O0}: run \texttt{clang-o0} on the input program.
  \item \texttt{LLVM-O3-O0}: runs \texttt{llvm-o3} and \texttt{clang-o0} on the input program.
  \item \texttt{\eggccT-DP-O0}: runs \texttt{llvm-o0} and \texttt{clang-o0} on the output of \eggcc using \tiger.
  \item \texttt{\eggccT-GUROBI-O0}: runs \texttt{llvm-o0} and \texttt{clang-o0} on \eggcc output using ILP extraction with \gurobi.
\end{itemize}

To run \llvm, we convert \bril programs to
  \llvm by using the ``Brillvm'' tool from the developers of Bril.
In the case of \texttt{LLVM-O3-O0}, we run LLVM with
  optimization level O3 and generate
  code with optimization level O0, disabling
  target-specific optimizations.
Disabling target-specific optimizations
  ensures fair comparison to \eggcc,
  which is target-agnostic.
We measure the runtime of the resulting binaries by
  using x86's \texttt{rdtsc} instruction
  to count CPU cycles taken to execute the program.
Each resulting binary is run 200 times
  and we report the average runtime.
\autoref{app:absolute-runtimes} gives the absolute runtimes
  behind the normalized numbers reported here,
  along with their standard deviations.\tighten

\begin{wrapfigure}{r}{0.5\textwidth}
    \centering
    \includegraphics[width=0.5\textwidth]{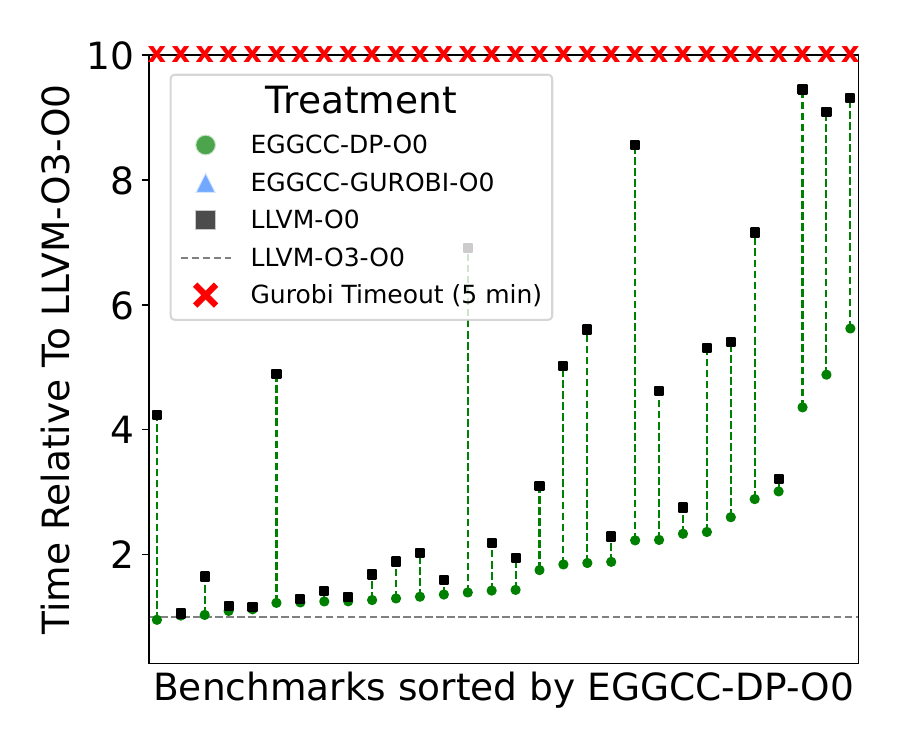}
    \caption{Comparison of \eggcc to \llvm on the \polybench benchmark suite.
      Lower numbers are better.}
    \Description{A chart with one vertical interval per benchmark. The horizontal
      axis lists the PolyBench benchmarks sorted by the eggcc with statewalk DP
      result; the vertical axis is runtime relative to the LLVM-O3-O0 baseline,
      drawn as a dashed line at 1.0. Markers distinguish eggcc with statewalk DP,
      eggcc with Gurobi, and LLVM-O0. Most benchmarks sit between one and four
      times the baseline, rising toward the right. A solid row of crosses runs
      along the very top of the chart, marking that Gurobi timed out on every
      benchmark in the suite.}
    \label{fig:eggcc-vs-llvm-polybench}
\end{wrapfigure}

Brillvm uses the \texttt{alloca} instruction to
  allocate registers,
  which is a common way to compile imperative programs to
  \llvm~\cite{llvm-tutorial-section-7}.
For all treatments, we therefore enable \llvm's \texttt{sroa} pass
  which scalarizes memory allocations~\cite{sroa}.
We also enable register allocation related optimizations for all treatments
   using the flags \texttt{-regalloc=greedy} and \texttt{-optimize-regalloc}
   since \eggcc does not do machine-level optimizations.

\autoref{fig:eggcc-vs-llvm-bril} and \autoref{fig:eggcc-vs-llvm-polybench}
  compare the performance of the resulting binaries
  from \eggcc with different extraction algorithms
  to those from \llvm.
We show that \eggcc with \tiger
  performs similarly to \eggcc with ILP-based extraction.
On the \polybench benchmark suite,
  normalized to \texttt{LLVM-O0},
  \eggcc achieves a geometric mean time across all benchmarks
   of \polybenchOveralleggcctigerOzeroOzeroNormalizedMean,
  compared to \polybenchOverallllvmOthreeOzeroNormalizedMean for \texttt{LLVM-O3-O0}.
On the \bril benchmark suite, \eggcc achieves a mean normalized time of
  \brilOveralleggcctigerOzeroOzeroNormalizedMean
  compared to \brilOverallllvmOthreeOzeroNormalizedMean for
  \texttt{LLVM-O3-O0}.
We also tried LLVM with target-specific optimization on,
 which achieves
 \polybenchOverallllvmOthreeOthreeNormalizedMean
 on \polybench and \brilOverallllvmOthreeOthreeNormalizedMean
 on \bril.
This matches our expectation since
 \eggcc is a prototype compiler and misses many
  optimizations present in \llvm.
Spot-checking the generated assembly, the largest gaps are
  scalar promotion of memory, which leaves redundant loads and stores,
  along with generalized loop unrolling and stack coloring;
  none of these is a fundamental obstacle.\tighten

\subsection{Case Study: Hacker's Delight (RQ3)}

\label{sec:Fenwick}

To showcase the impact of
  performing domain-specific effectful optimizations
  in \egraphs, made possible by \tiger,
  we perform a case study
	using an optimization taken from 
	Hacker's Delight~\cite{hackersdelight},
    a collection of efficient implementation tricks
    for low-level, bit-manipulating arithmetic.

\begin{wrapfigure}{R}{0.32\textwidth}
  \centering
\begin{lstlisting}[language=Rust,numbers=none,basicstyle=\footnotesize\ttfamily,aboveskip=0pt,belowskip=0pt]
lb = 1
while n % (lb * 2) == 0 {
  lb = lb * 2
}
return lb
\end{lstlisting}
	\rule[0pt]{\linewidth}{0.4pt}
\begin{lstlisting}[language=Rust,numbers=none,basicstyle=\footnotesize\ttfamily,aboveskip=0pt,belowskip=0pt]
return n & (-n)
\end{lstlisting}
  \Description{Two short code listings rather than a picture, separated by a
    horizontal rule. Above the rule, the naive lowbit implementation: set lb to
    1, then while n is divisible by lb times 2, double lb, and return lb. Below
    the rule, the one-line equivalent: return n bitwise-and negative n.}
\end{wrapfigure}

The Hacker's Delight trick we implement
  is a one-liner for the \texttt{lowbit} function.
For a positive integer $n$,
	\texttt{lowbit(n)} returns the least significant $1$-bit of $n$.
This function is used in bit-manipulation
  and data structures like Fenwick trees~\cite{fenwicktree},
	which are used to support operations on
	statistical frequency tables.
We show the na\"ive implementation and the one-liner\footnote{
Why this works:
Modern machines use $2$'s complement for binary representation of signed integers.
$-n$ in $2$'s complement
  effectively flips all the bits of $n$
	except for the least significant $1$ bit \texttt{and} the trailing zeros.
The bitwise \texttt{and} instruction then picks out the only bit
  that is $1$ in both numbers,
  which is the lowest $1$ bit of $n$.
} on the right.

Neither \llvm nor GCC performs this optimization.
Adding optimizations like this
	to a traditional compiler
	is extremely delicate,
  requiring deep knowledge
  of existing passes
	to mitigate phase-ordering problems
	and careful pattern matching to perform the optimization reliably.
Our rules sidestep both difficulties by running alongside \eggcc's other rules:
	the \egraph already holds the syntactic variants that a hand-written
	pattern would have to anticipate.

\begin{wrapfigure}[17]{r}{0.34\textwidth}
  \centering
  \includegraphics[width=\linewidth]{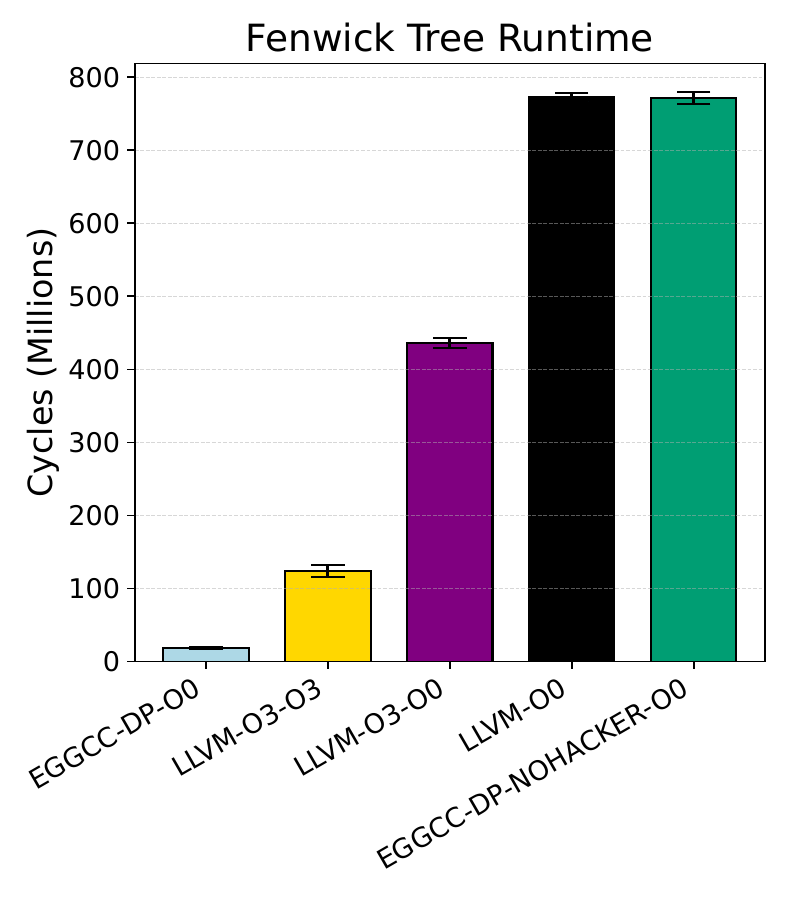}
  \caption{Fenwick tree runtime in CPU cycles. Lower is better.}
  \Description{A bar chart of five treatments, with CPU cycles in millions on the
    vertical axis. The leftmost bar, eggcc with the Hacker's Delight rules, is a
    small fraction of the height of every other bar. The remaining four rise from
    a low yellow bar for clang at O3 to three tall bars of similar height for
    LLVM-O3-O0, LLVM-O0, and eggcc without the Hacker's Delight rules.}
  \label{fig:fenwick}
\end{wrapfigure}
\paragraph{Implementation}
We implemented the \texttt{lowbit} optimization as three
  additional \egglog rules, $36$ lines excluding comments,
	without modifying any existing \eggcc rules.
The first rule marks predicates that test whether a value is even.
The second uses it to identify loops that iterate
  over the trailing zeros of a variable $n$.
The third detects
  an output of such a loop equivalent to \texttt{lowbit(n)},
  marking it equivalent to the optimized implementation.
\autoref{app:lowbit} gives all three.
These rules cooperate seamlessly with
  the existing analyses and optimizations of \eggcc,
  including effectful optimizations.
As such, the rules were easy to develop.

To showcase the power of this optimization,
  we implemented the Fenwick tree data structure in Rust,
	whose core operations involve
	calling the na\"ive implementation
	in a loop.
The program performs $2n$ operations
  on a tree of size $n$, with $n=2 \times 10^5$.

\paragraph{Results}
The runtime performance of the Fenwick tree implementation
  under different optimization treatments is shown in
  \autoref{fig:fenwick}.
  \begin{itemize}[leftmargin=2em]
    \item \texttt{\eggccT-DP-O0} (leftmost, light blue) includes the 
      Hacker's Delight rules alongside the rest of \eggcc's optimizations.
      Applying the Hacker's Delight optimization to the Fenwick tree
      eliminates the entire hot loop, yielding
      a \FenwickMeanCycleSpeedupEggcctigerWITHCTXOZeroOZeroVsllvmOThreeOThree times
      speedup compared to \texttt{LLVM-O3-O3}
      and a \FenwickMeanCycleSpeedupEggcctigerWITHCTXOZeroOZeroVsllvmOThreeOZero times
      speedup compared to \texttt{LLVM-O3-O0}.
    \item \texttt{LLVM-O3-O3} (second, yellow) runs \texttt{clang-o3}.
      This treatment replaces divisions and multiplications with bitwise shifts
      but could not further optimize the while loop.
    \item \texttt{LLVM-O3-O0} and \texttt{LLVM-O0} (middle, purple and black) are as described in \autoref{subsec:outputquality}.
    \item \texttt{\eggccT-DP-NOHACKER-O0} (rightmost, green) runs \eggcc without the Hacker's Delight rules.
  \end{itemize}
This comparison illustrates 
  the potential of applying domain-specific optimizations to effectful programs
	enabled by \eqsat and \tiger.


\subsection{Demonstrating the Effectful \Egraph Extraction Bottleneck in \peggy}
\autoref{fig:intro-egraph-extraction-bottleneck} in \autoref{sec:intro}
  shows the results of running \peggy and \eggcc with ILP-based extraction
  to
  demonstrate that \effectsafe extraction is a bottleneck for \egraph-based
  program optimizers.
Here, we describe the experimental setup for \peggy.
We ran \peggy on (1) a subset of 11 SpecJVM programs,
  which Peggy was originally evaluated on, and
 (2) 30 programs in PolyBench,
 which we manually translated to Java.
By default, \peggy optimizes methods independently and does not perform inlining,
 while \eggcc does.
To derive benchmarks with larger method bodies, we inlined methods in \peggy and ran
 \peggy on both the original programs and the inlined programs.
In total, we optimized 281 methods using \peggy.
 We use Java bytecode instruction count as the
  measure of program length for \peggy.
We measure the percentage of time spent in \egraph extraction, with a timeout of 10 minutes.
\peggy timed out on 19 of the 281 methods,
 and extraction became a bottleneck quickly when the programs got larger.
The results are shown in \autoref{fig:intro-egraph-extraction-bottleneck} (Left).

\section{Related Work}
\label{sec:related}

\paragraph{Effect-safe Extraction in Peggy}

In this paper, we compare our approach to Peggy~\cite{peggy}, which 
  uses a succinct dataflow-based IR called Program Expression
  Graphs (PEG) for program representation.
In Peggy, translating the extracted PEG 
 back to a CFG (for code generation)
 can fail and is difficult to implement, as documented in a separate 37-page technical
 report~\citep{peg-cfg}.
Peggy did not study the effect-safe extraction problem
  theoretically.
Due to the incompleteness of its ILP encoding,
  it cannot guarantee the optimality of its extraction,
	even when ILP succeeds.
Despite ILP being NP-complete,
  the encoding does not imply that effect-safe extraction is in NP.

\paragraph{Pure Extraction in Practice}
The pure extraction problem
  is commonly handled by algorithms
  that produce best-effort results \cite{smoothe, eboost,extraction-gym}.
While ILP solvers can theoretically produce optimal outputs,
  they scale poorly
  and restrict the cost model to those that
  can be expressed as systems of linear equations.
The only practical optimal pure extraction algorithm
  is based on treewidth~\cite{goharshady, glenn-treewidth},
  but it only works on sparse \egraphs
  and struggles in the presence of cycles.

A host of related work sidesteps the problem
  of effect-safe extraction using pure extraction.
DialEgg optimizes MLIR with \egraphs~\cite{mlir-egg}. DialEgg does not directly address problems with effects,
  but suggests that users take care in
  the presence of effects.
Cranelift is a production-grade JIT compiler that uses an \egraph-like data structure for optimization~\cite{aegraph}. To deal with effects, Cranelift restricts optimizations to pure terms, relying on a complex skeleton structure to preserve the order of effectful operations.
Similarly, the \egraph-based Julia program optimizer uses a CFG skeleton structure to preserve effectful operations~\cite{eqsat-julia}.
HEC~\cite{hec} is a framework for equivalence checking using \egraphs that is able to
  handle memory and control-flow transformations, but since it performs
  equivalence checking, not optimization, the tool does not need to support effect-safe extraction.\tighten


\section{Conclusion and Future Work}
\label{sec:conclusions}

Effect-safe extraction has been a
  practical bottleneck for effectful \egraph optimizers
  for more than a decade~\cite{peggy}.
This paper addresses that challenge with \tiger,
  a tractable effect-aware extraction algorithm
  parameterized by \statewalkwidth.
Our results show that \tiger removes
  extraction as a bottleneck while producing
  high-quality \effectsafe programs.
Beyond extraction,
  \statewalkwidth offers a new lens on
  effectful interactions between terms in an \egraph,
  exposing structure that
  may be useful for analyses, verification, or
  effect-aware rewrite strategies.

More broadly,
  eliminating ILP from effectful extraction
  enables \eqsat in
  compiler and superoptimizer settings.
Many existing IRs
  mix pure computation with effectful operations.
Examples include
  bytecode formats as in WebAssembly~\cite{wasm},
  region-based pipelines in MLIR~\cite{mlir},
  SSA-based forms in LLVM IR~\cite{llvm} or Swift SIL~\cite{swift}, and
  last-stage IRs like OCaml's Cmm~\cite{ocamlcmm}.
\tiger
  makes it feasible to start exploring
  effectful e-graph optimization in these settings.

In future work,
  we are excited about 
  combining statewalk width with other
  parameters like treewidth
  to tackle the optimal effect-safe extraction
  problem for cost-critical settings like superoptimizers.
Formalizing the observations of \autoref{sec:complexity} about which rules
  increase \statewalkwidth may also provide
  stronger theoretical bounds for real programs.
A further direction is to make extraction aware of distinct categories of
  effect, so that independent effects may commute.
A notion of separation in the style of separation logic would let the state
  split into pieces that evolve independently,
  generalizing \statewalkwidth and \tiger
  to \statewalks that split and merge.


\section*{Data-Availability Statement}
The source code of \eggcc and \tiger,
  together with all benchmarks and scripts used to produce the results in
  \autoref{sec:eval},
  is archived on Zenodo~\cite{eggcc-artifact} at
  \url{https://doi.org/10.5281/zenodo.21479729}.
Development continues at \url{https://github.com/egraphs-good/eggcc}.

\begin{acks}
We thank Glenn Sun for porting the \polybench benchmarks,
  Kevin Yan for infrastructure work including pretty printing,
  and Ryan Berger for work on our nightly benchmarking infrastructure.
We thank Patrick LaFontaine for the \texttt{rs2bril} compiler
  and Adrian Sampson for creating and maintaining \bril.
We also thank the OOPSLA reviewers, whose suggestions improved both the formal
  development and the evaluation.

This material is based upon work supported by the
  \grantsponsor{GS100000001}{National Science Foundation}{https://doi.org/10.13039/100000001}
  under Grant No.~\grantnum{GS100000001}{2232339} and
  \grantnum{GS100000001}{2312195},
  and by an National Science Foundation Graduate Research Fellowship Program
  under Grant No.~\grantnum{GS100000001}{DGE-2140004}.
Any opinions, findings, and conclusions or recommendations expressed in this
  material are those of the authors and do not necessarily reflect the views of
  the National Science Foundation.
\end{acks}

\bibliographystyle{ACM-Reference-Format}
\bibliography{references}

\label{mainend}

\newpage
\appendix
\section{Formal Definitions for \autoref{sec:formalizing-rvsdgs}}
\label{app:definitions}

This appendix makes precise some auxiliary definitions used in
\autoref{sec:formalizing-rvsdgs}, referring throughout to the
syntax of \autoref{fig:syntax} and the big-step semantics of
\autoref{fig:semantics}.

\subsection*{Children and Subterms}

\begin{definition}[Children and subterms]
  \label{def:children}
  The \emph{children} of a term $t$, written $\children{t}$,
  are its immediate child terms, specifically:
  \[
    \children{t} \;=\;
    \begin{cases}
      \emptyset
        & t = \val \;\text{or}\; t = \argt, \\[4pt]
      \{e\}
        & t = e[n], \\[4pt]
      \{e_1,\ldots,e_n\}
        & t = \texttt{f}(e_1,\ldots,e_n) \;\text{or}\; t = \texttt{F}(e_1,\ldots,e_n), \\[4pt]
      \{e_0,\ldots,e_{k-1}\}
        & t = (e_0,\ldots,e_{k-1}).
    \end{cases}
  \]
  The \emph{subterms} of a term $t$, written $\subterms{t}$, are $t$
  itself together with all terms reachable by following the children
  relation. Formally,
  \(
    \subterms{t}
    \;=\;
    \{t\}
    \;\cup\;
    \bigcup_{e\,\in\,\children{t}} \subterms{e}.
  \)
\end{definition}

\subsection*{Type System}

\begin{definition}[Types]
  \label{def:types}
  We define the set of types $\tau$ as follows
  \[
    \tau \;\;::=\;\; \Phi \;\mid\; \Sigma \;\mid\; \Phi\times \cdots \times \Phi \;\mid\; \Phi\ \times \cdots \times \Phi\ \times \Sigma
  \]
\end{definition}
We say $\tau$ is an \emph{effectful type} if $\tau$ has the form $\Sigma$ or $\Phi \times \cdots \times \Phi \times \Sigma$.
We say $\tau$ is a \emph{pure type} if $\tau$ has the form $\Phi$ or $\Phi \times \cdots \times \Phi$.
With a slight abuse of notation, we also use types to denote the sets of values of those types, e.g., all base values belong to $\Phi$.

Let us assume that each $n$-ary operator has a type of the form
 $\tau_1 \times \cdots \times \tau_n \to \tau'$,
 and that its semantics follows that signature,
 mapping values of types $\tau_1,\ldots,\tau_n$ to a value of type $\tau'$.
For a pure operator $\texttt{f}$, every $\tau_i$ and $\tau'$ is a pure type.
For an effectful operator $\texttt{F}$, $\tau'$ is an \effectful type and
 exactly one argument is stateful:
 $\tau_n$ is \effectful and $\tau_1,\ldots,\tau_{n-1}$ are pure.
The typing judgement $\tau_a \vdash t : \tau$ is defined by the following rules, where $\tau_a$ is an effectful type of the program argument $\argt$.

\begin{equation*}
\begin{array}{lll}
  \inference[]{}
    {\tau_a \vdash v : \Phi}
  &\inference[]{}
    {\tau_a \vdash \sigma : \Sigma}
  \inference[]{}
    {\tau_a \vdash \argt : \tau_a}
\end{array}
\end{equation*}
\begin{equation*}
\begin{array}{l}
  \inference[$\mathcal{F}\in\{\texttt{f}, \texttt{F}\}$]{
      \tau_a \vdash e_i : \tau_i \ \text{for } i=1,\ldots,n \andalso
      \textit{type}(\mathcal{F}) = \tau_1\times \cdots \times \tau_n \to \tau'
    }
    {\tau_a \vdash \mathcal{F}(e_1,\ldots,e_n) : \tau'}
\end{array}
\end{equation*}
\begin{equation*}
\begin{array}{ll}
  \inference[]{
      \tau_a \vdash e_i : \tau_i \quad \text{for } i = 0,\ldots,k-1
    }
    {\tau_a \vdash (e_0,\ldots,e_{k-1}) : \tau_0\times \ldots\times \tau_{k-1}}
  &
  \inference[]{
      \tau_a \vdash e : \tau_0\times \ldots \times \tau_{n-1} \andalso
      0 \le k < n
    }
    {\tau_a \vdash e[k] : \tau_k}
\end{array}
\end{equation*}

\begin{lemma}[Well-typed terms are well-formed]
  \label{lem:well-typed-well-formed}
  If $\tau_a \vdash t : \tau$,
  for all value $\argval$ of type $\tau$,
  there exists a unique value $v$ of type $\tau$ such that $\bigstep{\argval}{t}{v}$.
\end{lemma}
\begin{proof}[Proof sketch]
  By structural induction on the typing rule.
\end{proof}

\subsection*{Effectful and Pure Terms}

\begin{definition}[Effectful term]
  \label{def:effectful}
  Given program argument type $\tau_a$, a term $t$ is \emph{\effectful} if it has an effectful type,
  i.e., $\tau_a \vdash t : \tau$ for effectful type $\tau$.
  Otherwise, the term is \emph{pure}.
\end{definition}

Since our discussion is always parameterized by the input program argument type
 $\tau_a$, we will omit this $\tau_a$ when referring to effectful and pure terms.

\begin{corollary}
  \label{lem:effectful-stateful}
  An effectful term always evaluates to a stateful value,
  and a pure term always evaluates to a pure value.
\end{corollary}

\begin{lemma}[At most one effectful child]
  \label{lem:one-effectful-child}
  Every well-typed term has at most one \effectful child.
\end{lemma}
\begin{proof}[Proof sketch]
  By cases on the term.
  Values and $\argt$ have no children.
  An application of a pure operator has only pure arguments,
    and an application of an effectful operator has exactly one stateful
    argument, both by the typing assumption above.
  An indexing expression $e[n]$ has a single child.
  Finally, a tuple $(e_0,\ldots,e_{k-1})$ has type
    $\tau_0\times\cdots\times\tau_{k-1}$, which is \effectful only when it has the
    form $\Phi\times\cdots\times\Phi\times\Sigma$,
    so at most its last component is \effectful.
\end{proof}

\Cref{lem:one-effectful-child} is what makes the \statewalk of an
  \effectsafe term unique (\autoref{sec:effect-safety}): following the child
  relation from an \effectful term never has to choose between two \effectful
  children.

\subsection*{Effectful E-Classes}

\begin{lemma}[]
  \label{lem:eclass-homo}
  In an \egraph $\mathcal{G}$ that respects the equivalence $\approxeq$ induced by the semantics, every \eclass either
  contains only effectful terms or contains no effectful term.
\end{lemma}

\begin{proof}[Proof sketch]
  Suppose, for contradiction, that some \eclass $C$ contains both an
  effectful term $t_e$ and a pure term $t_p$.
  Since both belong to $C$,
  $t_e \approxeq_{\mathcal{G}} t_p$, so
  $t_e \approxeq t_p$ as the \egraph respects $\approxeq$. Therefore,
  $\forall a,v.\; \bigstep{a}{t_e}{v} \iff \bigstep{a}{t_p}{v}$.

  Suppose $\tau_a \vdash t_e : \tau_e$ and $\tau_a \vdash t_p : \tau_p$
  for effectful type $\tau_e$ and pure type $\tau_p$, and let $\argval$ be a value of type $\tau_a$,
  by \Cref{lem:well-typed-well-formed}, there exists value $v_p$ of type $\tau_p$
  such that $\bigstep{\argval}{t_p}{v_p}$, so $\bigstep{\argval}{t_e}{v_p}$.
  By the uniqueness implied by \Cref{lem:well-typed-well-formed}, $v_p=v_e$.
  However, pure values and stateful values are disjoint. Contradiction.
\end{proof}

Following from the above lemma, we can classify \eclasses into \effectful \eclasses and pure \eclasses.

\begin{definition}[Effectful \eclass]
  \label{def:effectful-eclass}
  An \eclass $C$ in an \egraph $\mathcal{G}$ is \emph{\effectful}
  if it contains only effectful term.
  It is pure if it contains only pure terms.
\end{definition}

\section{\tiger Correctness}
\label{sec:tiger-appendix}

\begin{lemma}
  \label{lem:ext-map-sound}
  All of the representative terms found by $\textsc{ExtractMap}$ are \effectsafe with respect to the given \statewalk. \par
  For an \egraph $\mathcal{G}$,
  \begin{align*}
    &\forall \sw = \langle e_0, \ldots, e_n \rangle,\quad
    \forall C \in \mathcal{G},\quad
    \forall t, \\
    &\textsc{ExtractMap}(\sw)[C] = t \implies \effectsafe(t, \sw) \land t \in \termsc{C}{\mathcal{G}}.
  \end{align*}
\end{lemma}

\begin{proof}
  First, we establish that once $E[C]$ is set for any \eclass $C$, it is never changed.
  Line 6 and Line 14 are the only lines that assign to $E[C]$ for an \eclass $C$.
  Both are guarded by conditionals ensuring that $E[C]$ was previously unset.

  We proceed by induction on the smallest term in \eclass $C$ by depth\footnote{
    The depth of a leaf term is 1. The depth of a term $f(t_1, \ldots, t_n)$
    is $1 + \text{max}_{i \in 1..n}(\text{depth}(t_i))$
  }.

  In the base case, the smallest term in \eclass $C$ is a leaf term, $f$.
  For any $\sw$, there is some leaf term $f'$ such that $\textsc{ExtractMap}(\sw)[C] = f'$,
    as established by Line 6.
  If there are multiple leaf terms in $C$, they all have the same depth, and $E[C]$
    may be set to any of them, based on the enumeration order of Line 5.
  All leaf terms are \effectsafe with respect to all statewalks:
    if the term is pure, then it contains no effectful subterms, and if it is not pure,
    then it must be Arg, which is the first term in every \statewalk.
  By construction, Line 6 ensures that $f' \in C$.

  Consider an \eclass $C$ whose smallest term is $f(t_1, \ldots, t_n)$.
  If $\textsc{ExtractMap}[\sw](C) = \bot$, the claim holds vacuously.
  Otherwise, $E[C]$ was set via Line 14 of the algorithm on some iteration of the loop.
  On the iteration of the inner loop where $E[C]$ was set,
    it must be the case that $t = f(t_1, \ldots, t_n)$ because $f(t_1, \ldots, t_n)$
    is the smallest term in the \eclass and classes are added to $E$ in breadth-first order.
  The condition on Line 11 must have been true in order for $E[C]$ to be set on Line 14.
  So, for all pure children of $t$, $E[\classof{t_i}{\mathcal{G}}] \neq \bot$ and
    if there is an \effectful child of $t$, it must be present in $\sw$.
  We assume by induction that $E[\classof{t_i}{\mathcal{G}}]$ gives an \effectsafe extraction with $\sw$
    for each pure child.
  Then, the constructed term is \effectsafe with $\sw$
    since its pure children are and its effectful child (if it has one) is in $\sw$.
  Further, the constructed term is in \eclass $C$ since it has the same function symbol
    and all of its children are in the same \eclass as the corresponding child term of $t$.
\end{proof}

\begin{lemma}
  \label{lem:ext-map-complete}
  If an \effectsafe term exists for a given pure \eclass with respect to a \statewalk $\sw$,
    $\textsc{ExtractMap}(\sw)$ will have a mapping for the \eclass to some term. \par
  \begin{align*}
    &\forall\; C \in \mathcal{G}\quad
    \forall t \in \termsc{C}{\mathcal{G}} \quad
    \forall \sw = \langle e_0, \ldots, e_n \rangle, \\
    &\text{pure}(C) \implies \effectsafe(t, \sw) \implies \exists \; t' \in C\; \textsc{ExtractMap}(\sw)[C] = t'.
  \end{align*}
\end{lemma}

\begin{proof}
  By induction on the depth of term $t$.

  As established in \Cref{lem:ext-map-sound}, once $E[C]$ is set for any \eclass $C$, it is not changed.

  In the base case, $t$ is a leaf term, so $\textsc{ExtractMap}(\sw)[C] = t'$,
    as set by Line 6, where $t'$ is a leaf term in \eclass $C$.

  In the inductive case, $t = f(t_1, \ldots, t_n)$.
  We assume by induction that for each pure child of $t$, there exists $t_i'$ such that
    $t_i$ and $t_i'$ are both members of \eclass $C_i$ and $\textsc{ExtractMap}(\sw)[C_i] = t_i'$.
  That is, at the end of some iteration of the outer loop, $E[C_i] = t_i'$
    for all pure $t_i$.

  By the end of the next iteration of the inner loop
    where $t = f(t_1, \ldots, t_n)$, $E[C]$ will be set.
  If it was already set prior to this iteration, then there must be a term in $C$
    with a smaller depth than $t$ since $E$ is populated in breadth-first order,
    and the result holds by induction.
  Otherwise, the predicate on Line 6 will be true:
    $E[C] = \bot$, since it is not already set;
    $\forall i, \text{pure}(t_i) \implies E[t_i] \neq \bot$, as established previously by induction;
    if $t$ has an \effectful child, then it is in $\sw$ since $t$ is \effectsafe;
    $t$ is pure by construction.
  The constructed term $f(t_1, \ldots, t_n)$, where $E[\classof{t_i}{\mathcal{G}}]$
    has been substituted for each pure child $t_i$,
     is in \eclass $C$
    since it has the same function symbol
    and all its children are in the same \eclass as the corresponding child term of $t$.
\end{proof}

\begin{lemma}
  \label{lem:dp-inv}
  In the DP table of the \textsc{Extract} algorithm, for all \eclasses, $C$, and for all sets of \eclasses, $es$,
  if $DP[C][es] = T$ for some sequence $T = \langle t_0, \ldots, t_n \rangle $,
  then $T$ is a \statewalk rooted at \eclass $C$.
\end{lemma}

\begin{proof}
  First, we establish that once $DP[C][es]$ is set for some \eclass $C$ and set of \eclasses $es$,
    it is not changed.
  Line 4 establishes $DP$ as an empty map.
  Line 6 sets $DP[\classof{\Arg}{\mathcal{G}}][e]$, where $e$ is the
  \extractionset of $\langle \Arg \rangle$.
  The only other line that assigns to $DP$ is Line 22, and the conditional on Line 21
    ensures that the corresponding (\eclass, \extractionset) was previously not set in $DP$.

  We proceed by induction on the smallest term in $C$ by depth.

  In the base case, the smallest term in $C$ is a leaf term.
  It must be \argt since the DP table only stores entries for effectful \eclasses.
  $DP[\classof{\Arg}{\mathcal{G}}][\text{Keys}(\textsc{ExtractMap}(\langle \Arg \rangle))] = \langle \Arg \rangle$,
  as set by Line 6.
  Trivially, $\langle \Arg \rangle$ is a \statewalk rooted at $\classof{\Arg}{\mathcal{G}}$.

  Consider an \eclass $C$ whose smallest term is $f(t_1, \ldots, t_n)$.
  If there is no $es$ such that $DP[C][es]$ is set, then the claim holds vacuously.
  Otherwise, for every $es$ such that $DP[C][es] \neq \bot$, there was some iteration
    of the loop which executed Line 22 to set $DP[C][es]$.
  On the corresponding iteration of the outer loop, let the worklist item being
    processed be $(C_{curr}, es_{curr})$ such that $DP[C_{curr}][es_{curr}] = \sw_{curr}$.
  On the iteration of the inner loop which sets $DP[C][es]$, $t = f(t_1, \ldots, t_n)$
    for some term $t \in C$ such that a term in $C_{curr}$ is a direct child of $t$.
  It must be the case that the smallest term in $C_{curr}$ has a smaller depth than $t$
    since terms cannot be cyclic.
  Then, by induction, $\sw_{curr}$ is a \statewalk rooted at $C_{curr}$
  (meaning that the last term in $\sw_{curr}$ is in $C_{curr}$).
  $t'$ is constructed
    to be a well-formed term in \eclass $C$
   using the mapping from $\textsc{ExtractMap}(\sw_{curr})$
    for all of $t$'s pure children.
  By \Cref{lem:ext-map-sound} and \Cref{lem:ext-map-complete},
    $M$ will contain an \effectsafe term with respect to $\sw_{curr}$
    corresponding to each $t_i$.
  Thus, the sequence of terms constructed by adding $t'$ to the end of $\sw_{curr}$ is
    a valid \statewalk rooted at \eclass $C$.
  This sequence of terms is set to $DP[C][es]$, and the invariant holds.
\end{proof}

\begin{lemma}
  \label{lem:tiger-inv}
  For any \statewalk \xspace $\sw = \langle e_0, \ldots, e_n \rangle$, \par
  $(\classof{e_n}{\mathcal{G}}, \text{Keys}(\textsc{ExtractMap}(\langle e_0, \ldots, e_n \rangle)))$ eventually gets added to the worklist.
\end{lemma}

\begin{proof}
  By induction on $\sw$.

  In the base case, $\sw = \langle \Arg \rangle$, and $(\classof{\Arg}{\mathcal{G}},\text{Keys}(\textsc{ExtractMap}(\langle \Arg \rangle)))$
    gets added to the worklist on Line 8.

  In the inductive case, $\sw = \langle e_0, \ldots, e_n, e_{n+1} \rangle$,
    and we assume by induction that $(\classof{e_n}{\mathcal{G}}, es_n)$
    was added to the worklist, where $es_n$ is the \extractionset of $\langle e_0, \ldots, e_n \rangle$.
  On the iteration of the outer loop where
  $(\classof{e_n}{\mathcal{G}}, \text{Keys}(\textsc{ExtractMap}(\langle e_0, \ldots, e_n \rangle)))$
    is popped from the worklist, $e_{n+1}$ will be in the \C{dependents} set since $e_n$ is a child of $e_{n+1}$.
  On the iteration of the inner loop where $t = e_{n+1}$, the conditional on Line 14 will be \C{false}:
  (1) $e_{n+1}$ is not pure since it is in the statewalk;
  and
  (2) all of its pure children are \effectsafe with respect to $\langle e_0, \ldots, e_n \rangle$
  (since $\sw$ is a well-formed \statewalk),
  so by \Cref{lem:ext-map-complete}, $\text{Keys}\textsc{ExtractMap}(\langle e_0, \ldots, e_n \rangle)$
  contains a term for each pure child.

  If the conditional on Line 21 is true, then $(\classof{e_{n+1}}{\mathcal{G}}, \text{Keys}(\textsc{ExtractMap}(\langle e_0, \ldots, e_{n+1} \rangle)))$
  is added to the worklist on Line 22, as desired.
  Otherwise, $DP[\classof{e_{n+1}}{\mathcal{G}}][\text{Keys}(\textsc{ExtractMap}(\langle e_0, \ldots, e_{n+1} \rangle))]$ was set on some previous
  iteration of the loop.
  $DP[\classof{e_{n+1}}{\mathcal{G}}][\text{Keys}(\textsc{ExtractMap}(e_{n+1}))]$ could only have been set by
    executing line 22 of the algorithm on some previous iteration of the loop,
    and $(\classof{e_{n+1}}{\mathcal{G}}, \text{Keys}(\textsc{ExtractMap}(\langle e_0, \ldots, e_{n+1} \rangle)))$ would have
    been added to the worklist.
\end{proof}

\begin{lemma}
  \label{lem:tiger-sound}
  $\textsc{Extract}$ only returns \effectsafe terms. \par
  \hspace{-0.5em}
  \(\forall t,\; \forall C \in \mathcal{G}
    ,\;
    \textsc{Extract}(C) = t \implies \exists\; \sw\in \eclasssws{C},\; \effectsafe(t, \sw) \land t \in \termsc{C}{\mathcal{G}}\)
\end{lemma}

\begin{proof}
  Let $t$ be a term such that $\textsc{Extract}(C) = t$ for \eclass $C$.
  It must be the case that for some set of \eclasses $es$, $DP[C][es] = \langle e_0, \ldots, e_k \rangle$
  such that $e_k = t$.
  As established by \Cref{lem:dp-inv},
    $\langle e_0, \ldots, e_k \rangle$ must be a \statewalk rooted at \eclass $C$.
  Then, $t \in \termsc{C}{\mathcal{G}}$ and $\effectsafe(t, \langle e_0, \ldots, e_k \rangle)$, as desired.
\end{proof}

\begin{lemma}
  \label{lem:tiger-complete}
  If \xspace $\textsc{Extract}$ cannot find an \effectsafe extraction, then none exists. \par
  $  \forall C \in \mathcal{G},\;
    \textsc{Extract}(C) = \bot \implies \not\exists\;
    t \in \termsc{C}{\mathcal{G}}\quad \sw \in \eclasssws{C},\ \effectsafe(t, \sw)$
\end{lemma}

\begin{proof}
  If $\textsc{Extract}(C) = \bot$, then $C \neq \Arg$, because $DP[\Arg]$ is set on Line 6,
    and $DP[C]$ was never set via Line 22 in the inner loop.
  By \Cref{lem:tiger-inv}, all \effectsafe \statewalks are eventually added to the worklist.
  If there were an \effectsafe \statewalk for \eclass $C$, then it would have been
    added to the worklist via Line 23, which would ensure that $DP[C]$ was set on Line 22.
  Since $\textsc{Extract}(C) = \bot$, $DP[C]$ must not have been set, so there must not
    be an \effectsafe \statewalk in \eclass $C$.
\end{proof}

\section{NP-Completeness of Effectful Extraction}

\subsection{NP-Hardness}
\label{app:np-hardness}

\begin{lemma}
  \label{lem:extraction-to-sat}
  If $C_R$ has an \effectsafe extraction,
    then $\phi$ is satisfiable.
\end{lemma}

\begin{proof}
Let $t$ be an \effectsafe extraction under \statewalk $\sw$.
$t$ corresponds to a valid assignment satisfying $\phi$.
Define the assignment $\alpha$ that assigns $\top$ to $x_i$ if $t_{x_i}\in \sw$ and $\bot$ if $t_{\neg x_i}\in \sw$.
For each variable $x_i$, either $t_{x_i}\in \sw$ or $t_{\neg x_i}\in \sw$,
 but not both, since neither is a subterm of the other.
$t$ must contain a term with head symbol $g_i$ for $i \in [1, k]$,
  and by \effectsafety,
  the child of each of these terms is contained in $\sw$,
  giving witness that $\varphi_i$ is satisfiable by $\alpha$.


%
\end{proof}

\begin{lemma}
  \label{lem:no-extraction-to-unsat}
  If $\phi$ is satisfiable, $C_R$ has an \effectsafe extraction.
\end{lemma}

\begin{proof}
Suppose there exists a solution to $\phi$, written as a sequence
  $\alpha = \langle \alpha_0, \ldots, \alpha_n \rangle$,
  where each $\alpha_i$ is either $x_i$ or $\neg x_i$,
  corresponding to $x_i = \top$ or $x_i = \bot$, respectively.
This solution corresponds to a \statewalk $\sw$ ending with
 $t_{\alpha} = s(\alpha_n(s(\alpha_{n-1}(s(\ldots(\alpha_0(\text{Arg})))))))$.
Additionally, for each clause $\varphi_i= \ell_1 \lor \ldots \lor \ell_m$
 where $\ell_j$ satisfies $\varphi_i$,
 let $t_{\ell_j}$ be the unique term
 with head symbol $\ell_j$ in $\sw$ and
 let $t_{\alpha,\varphi_i}$ be $g_i(t_{\ell_j})$,
 the term that is chosen for the equivalence class of $g_i$.
We claim $R(t_{\alpha,\varphi_1},\ldots, t_{\alpha,\varphi_k}, t_a)$ is an \effectsafe extraction
 of $C_R$.
First, it is a term represented by $C_R$:
 we can inductively show that
 $t_{\alpha}$ is equivalent to $t_{x_n}$
 and
 each $t_{\alpha,\varphi_{i}}$ is equivalent to $t_{\varphi_{i,1}}$,
 so by congruence $R(t_{\alpha,\varphi_1},\ldots, t_{\alpha,\varphi_k}, t_a)$
 is equivalent to $R(t_{\varphi_1},\ldots, t_{\varphi_k}, t_{x_n})$.
Second, it is \effectsafe,
 as by our construction the only effectful subterms are
 those in $\sw$.

\end{proof}

Taken together, the lemmas show that $\phi$ is satisfiable iff there exists an \effectsafe extraction for the root \eclass.
The encoding is linear in the size of $\phi$.
Therefore, finding an \effectsafe extraction from an
  effectful dataflow \egraph is NP-hard.
This implies the optimal variant of the problem
  is also NP-Hard.



\subsection{NP-Easiness}
\label{app:np-easy}

To prove the NP-Easiness of \effectsafe extraction,
  we show a poly-sized reduction to SAT.

We start with a simple encoding
  of finding any extraction
  without considering effect safety.
We can encode the problem in $O(n^3)$ variables and clauses,
  where $n$ is the size of the \egraph.
\begin{itemize}
  \item Each \eclass can be included in the extraction or not ($O(n)$ variables).
	\item For each representative term in the \egraph,
    either its head is included in the extraction or it is not ($O(n)$ variables).
	\item An \eclass is included if and only if there exists at least
  one representative whose head is included ($O(n)$ clauses).
	\item A head can only be included if all of its children \eclasses
	are included ($O(n)$ clauses).
	\item The acyclicity requirement on the extraction
    can be encoded
	  by enforcing a total order
	  on the extracted terms with $O(n^2)$ additional variables and $O(n^3)$ additional clauses.
\end{itemize}

To guarantee effect safety,
  we need to find a \statewalk
  (i.e., a list of effectful terms)
  that validates the final extraction.
We make a copy of the simple encoding for each term in the \statewalk,
  and add additional constraints to each copy:
  (1) The extraction may only use \eclasses that are extractable
      with the \statewalk up to the previous term
      or the \eclass of the previous term itself
  and (2) The current term is an extension of the \statewalk ending with the
      previous term.

Our critical insight is that
  it is complete to only consider \statewalks up to length $n^2$.
The extractable set of a \statewalk is non-decreasing
  as the \statewalk extends,
	so there can only be up to $n$ different extractable sets
	for a \statewalk's non-empty prefixes.
Each term of the \statewalk is in one of $n$ \eclasses,
so for any \statewalk with a length greater than $n^2$,
  by the pigeonhole principle,
  there must exist two prefixes that have the same extractable set and \eclass pair,
	meaning they are equivalent under $\equiv_{\sw}$.
Thus, there exists a strictly shorter sub-\statewalk that also admits an \effectsafe extraction,
	which can be constructed by skipping over all terms between the two appearances,
	substituting the later appearance with the earlier one in all places.
Its soundness is guaranteed by substitutability.

As a result, we only need to copy the simple encoding $n^2$ times
  (i.e., search for a \statewalk up to length $n^2$).
This constitutes an encoding of the \effectsafe extraction problem
	in SAT using $O(n^5)$ variables and clauses,
	proving its NP-easiness.

Next, we prove the NP-Easiness of the optimal \effectsafe extraction problem.
The decision version of the optimal \effectsafe extraction problem requires more constraints to encode the cost model.
These constraints are analogous to those required for optimal pure extraction~\cite{stepp-thesis}.
  to encode the cost model.
Importantly, our insight of making at most $n^2$ copies still applies
  because a sub-\statewalk always leads to a cheaper extraction.
Thus, the reduction remains $O(\operatorname{poly}(n))$ in the optimal case.

\section{\eggcc's ILP Encoding Details}
\label{app:ilp-details}

As defined in \autoref{subsec:backgroundegraphs},
  an extraction is a term $t$ in an equivalence class
  $C$
  of the partial equivalence relation (PER) $\approxeq_\mathcal{G}$
  that the \egraph represents.
Unfortunately, $C$ can be infinite (see \autoref{sec:overview}).
To make the problem tractable,
  we assume that each term in the \egraph is
  used at most once in the extraction.
In practice, we found two instances in which
  this approximation fails, with a statewalk
  reusing the same term multiple times.
Since \tiger considers all statewalks up to
  \extractionset equivalence, it does not
  suffer from this limitation.

\myparagraph{ILP Variables}

The ILP constructs an extraction by selecting a set of
  terms from the regionalized \egraph and a set of edges describing how those terms are connected.
For each child of a term $t$, the ILP is free to choose
  any candidate term from the corresponding child's \eclass.
We introduce three families of decision variables:
\begin{itemize}
  \item \textbf{Selection variables.} For each \eclass $E$ and term $t \in E$,
    $p_t \in \{0,1\}$ indicates whether the extraction chooses $t$.
  \item \textbf{Child-choice variables.} For every child position $i$ of $t$ and
    every candidate child term $u$ in the corresponding child \eclass,
    $s_{t,i,u} \in \{0,1\}$ records which concrete child is attached when $t$ is selected.
  \item \textbf{Order variables.} Each term receives an integer $o_t \in [0, N)$,
    where $N$ is the total number of terms under consideration, to encode a topological order that rules out cycles.
\end{itemize}

All $p$ and $s$ variables are binary, while the $o$ variables are bounded integers.

\myparagraph{Constraints}
The variables are coupled by the following linear constraints:
\begin{enumerate}
  \item The encoding enforces that at least one term is selected from the root
  \eclass $E_{root}$:
  {
    \abovedisplayskip=3pt
    \belowdisplayskip=3pt
    \[ \sum_{t \in E_{\mathit{root}}} p_t \geq 1 \]
  }
  \item \textbf{Child completion.}  Whenever a term $t$ is selected,
    every child position must pick a concrete successor.
    For each child index $i$, the ILP encoding adds:
    {
      \abovedisplayskip=3pt
      \belowdisplayskip=3pt
      \[ \sum_u s_{t,i,u} \ge p_t \]
    }
  \item \textbf{Edge activation.}  If a choice variable is active, the
    corresponding child must also be active.
    We ensure the extraction never references an un-picked term with the constraints:
    {
      \abovedisplayskip=3pt
      \belowdisplayskip=3pt
      \[ s_{t,i,u} \le p_u \]
    }
  \item \textbf{\Effectsafe\ fan-in.}  This constraint ensures an \effectsafe extraction by
    preventing two effectful edges from targeting the same term.
    For each effectful child $u$ we add:
    {
      \abovedisplayskip=3pt
      \belowdisplayskip=3pt
      \[\sum_{(t',i) : \mathit{effectful}(t')} s_{t',i,u} \le 1\]
    }
  \item \textbf{Acyclicity.}  The order variables ensure a grounded term can be reconstructed by preventing cycles.
    The following constraints ensure each child has a lower index than its parents:
    {
      \abovedisplayskip=3pt
      \belowdisplayskip=3pt
      \[ o_u - o_t + N \cdot s_{t,i,u} \le N-1 \]
    }
\end{enumerate}

\myparagraph{Objective Function}

We define an objective function using the same cost model as
  \tiger (\autoref{sec:cost-model}).
Each term $t$ contributes to $\mathsf{cost}(t)$
  if it is selected in the extraction.
The ILP minimizes the total cost of selected terms:
\[\text{minimize} \sum_{t} p_t \cdot \mathsf{cost}(t)\]

\subsection{Recovering the Extracted Term}
For each $p_t = 1$, we construct a term $t'$ recursively by:
  \begin{enumerate}
    \item Choosing the same operator as $t$.
    \item For each child position $i$ of $t$, find a variable $s_{t,i,u} = 1$ and construct $u'$ recursively.
  \end{enumerate}

This process constructs a finite tree because of the acyclicity constraints.
The \effectsafe constraints ensure effectful operations are limited to a single \statewalk.

\section{\tiger Worked Example}
\label{app:tiger-example}

\autoref{fig:tiger-worked-example} traces \tiger
  (\autoref{alg:tiger-extract}) on a small \egraph in which
  load elimination has made two \statewalks available,
  showing the DP table entry added at each iteration
  and the \effectsafe term finally extracted.

\begin{figure}
  \includegraphics[width=0.85\textwidth]{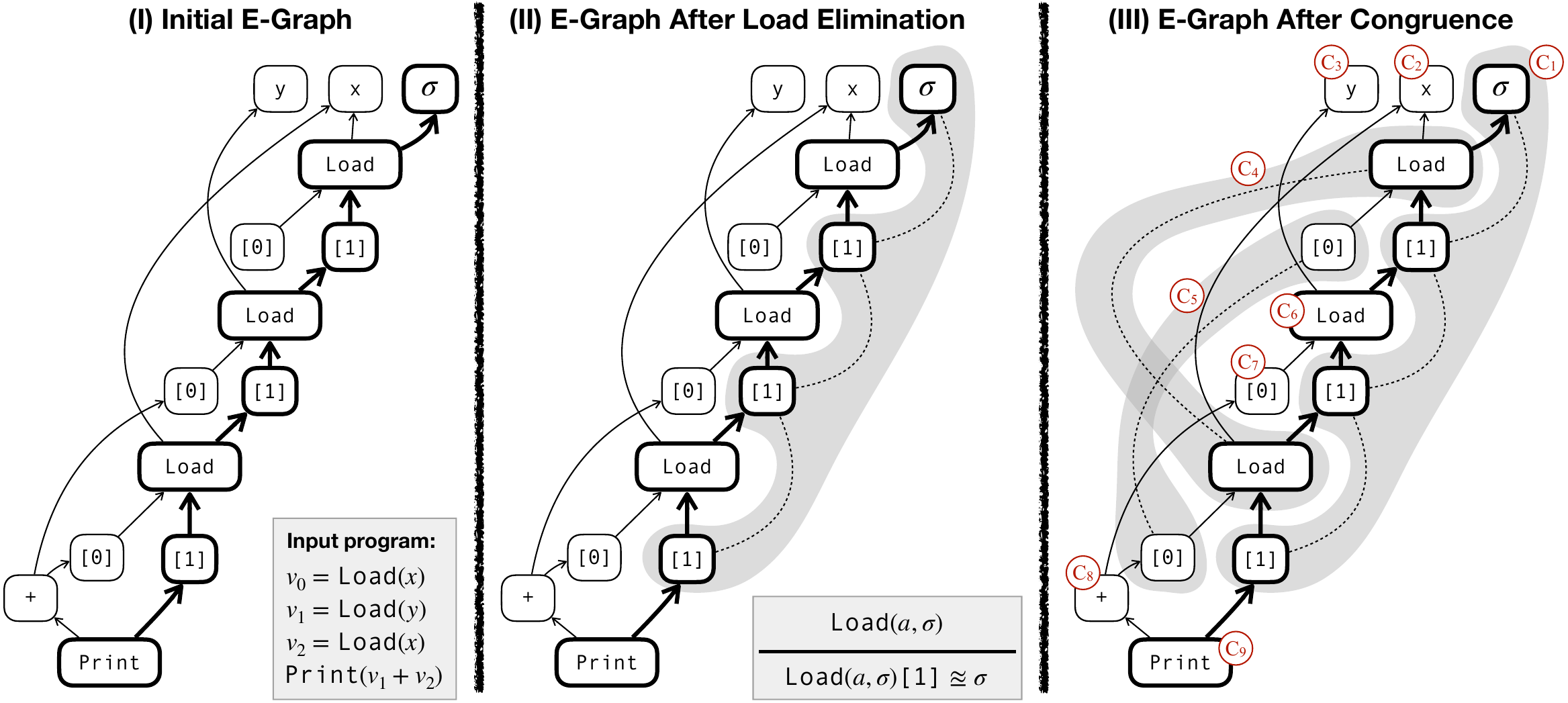}

  \vspace{0.05cm}

  \footnotesize
  \begin{tabular}{llll}
  Iteration & $(c, es)$ & $DP[c][es]$ & Added By \\ \hline
  1 & $(C_1, \{C_2, C_3\})$ & \begin{tabular}{l}
                                    $\langle \argt \rangle$ \\
                                  \end{tabular} & Line 6 \\ \hline
  2 & $(C_4, \{C_2, C_3, C_5\})$ & \begin{tabular}{l}
                                    $\langle \argt, \texttt{Load}(x, \argt) \rangle$ \\
                                  \end{tabular} & Iter 1 \\ \hline
  3 & $(C_6, \{C_2, C_3, C_7\})$ & \begin{tabular}{l}
                                    $\langle \argt, \texttt{Load}(y, \argt) \rangle$ \\
                                  \end{tabular} & Iter 1 \\ \hline
  4 & $(C_1, \{C_2, C_3, C_5\})$ & \begin{tabular}{l}
                                    $\langle \argt, \texttt{Load}(x, \argt), \texttt{Load}(x, \argt)[1] \rangle$ \\
                                  \end{tabular} & Iter 2 \\ \hline
  5 & $(C_1, \{C_2, C_3, C_7\})$ & \begin{tabular}{l}
                                    $\langle \argt, \texttt{Load}(y, \argt), \texttt{Load}(y, \argt)[1] \rangle$ \\
                                  \end{tabular} & Iter 3 \\ \hline
  6 & $(C_6, \{C_2, C_3, C_5, C_7, C_8\})$ & \begin{tabular}{l}
                                    $\langle\argt,$ \\
                                    $\phantom{\langle}\texttt{Load}(x, \argt),$ \\
                                    $\phantom{\langle}\texttt{Load}(x, \argt)[1],$ \\
                                    $\phantom{\langle}\texttt{Load}(y, \texttt{Load}(x, \argt)[1])\rangle$ \\
                                  \end{tabular} & Iter 4 \\ \hline
  7 & $(C_4, \{C_2, C_3, C_5, C_7, C_8\})$ & \begin{tabular}{l}
                                    $\langle\argt,$ \\
                                    $\phantom{\langle}\texttt{Load}(y, \argt),$ \\
                                    $\phantom{\langle}\texttt{Load}(y, \argt)[1],$ \\
                                    $\phantom{\langle}\texttt{Load}(x, \texttt{Load}(y, \argt)[1])\rangle$ \\
                                  \end{tabular} & Iter 5 \\ \hline
  8 & $(C_1, \{C_2, C_3, C_5, C_7, C_8\})$ & \begin{tabular}{l}
                                    $\langle\argt,$ \\
                                    $\phantom{\langle}\texttt{Load}(x, \argt),$ \\
                                    $\phantom{\langle}\texttt{Load}(x, \argt)[1],$ \\
                                    $\phantom{\langle}\texttt{Load}(y, \texttt{Load}(x, \argt)[1]),$ \\
                                    $\phantom{\langle}\texttt{Load}(y, \texttt{Load}(x, \argt)[1])[1]\rangle$ \\
                                  \end{tabular} & Iter 6 \\ \hline
  9 & $(C_9, \{C_2, C_3, C_5, C_7, C_8\})$ & \begin{tabular}{l}
                                    $\langle\argt,$ \\
                                    $\phantom{\langle}\texttt{Load}(x, \argt),$ \\
                                    $\phantom{\langle}\texttt{Load}(x, \argt)[1],$ \\
                                    $\phantom{\langle}\texttt{Load}(y, \texttt{Load}(x, \argt)[1]),$ \\
                                    $\phantom{\langle}\texttt{Load}(y, \texttt{Load}(x, \argt)[1])[1],$ \\
                                    $\phantom{\langle}\texttt{Print}(\texttt{Add}(\texttt{Load}(x, \argt)[0],$ \\
                                    $\phantom{\langle\texttt{Print}(\texttt{Add}(}\texttt{Load}(y, \texttt{Load}(x, \argt)[1])[0]),$ \\
                                    $\phantom{\langle\texttt{Print}(}\texttt{Load}(y, \texttt{Load}(x, \argt)[1])[1])\rangle$ \\
                                  \end{tabular} & Iter 8 \\ \hline
  \end{tabular}%

  \caption{A worked example of \tiger extraction over an \egraph.
  The panels show the \egraph as initially built (left),
  after load elimination (middle),
  and after the congruence invariant is restored (right);
  \tiger then returns an \effectsafe extraction,
  in essence choosing which load operations to perform.
  Below is the DP table computed with root \eclass $C_9$, whose columns give
  the iteration of \autoref{alg:tiger-extract}'s outer loop,
  the \eclass and extractable set processed,
  the \statewalk $\langle e_0, \ldots, e_k \rangle$ then stored in the table, and
  the iteration that added $(c, es)$ to the worklist.
  The final extracted term is the last term of the last row:
    $\texttt{Print}(\texttt{Add}(\texttt{Load}(x, \argt)[0], \texttt{Load}(y, \texttt{Load}(x, \argt)[1])[0]), \texttt{Load}(y, \texttt{Load}(x, \argt)[1])[1])$.
  }
  \Description{A three-panel e-graph diagram above a nine-row table. The panels
    show the same small e-graph for a program that loads from two addresses and
    prints their sum: first as initially built, then after load elimination adds
    an equality, then after congruence closes that equality, which creates a
    second possible statewalk. The table below has one row per iteration of the
    algorithm, giving the e-class and extractable set processed, the statewalk
    stored in the dynamic programming table, and the iteration that enqueued it.
    The statewalks grow one term per row until the final row holds the complete
    extraction.}
  \label{fig:tiger-worked-example}
\end{figure}

\FloatBarrier

\section{Absolute Running Times}
\label{app:absolute-runtimes}

\autoref{fig:eggcc-vs-llvm-bril} and \autoref{fig:eggcc-vs-llvm-polybench}
  report runtimes normalized to \texttt{LLVM-O0}.
For reference, this appendix gives the corresponding
  \emph{absolute} runtimes together with their standard deviations.

Each binary was executed 200 times, preceded by 10 warmup runs.
We count CPU cycles with the \texttt{rdtsc} instruction
  and convert cycles to time using the
  measured time-stamp-counter rate of the benchmark machine
  (1999.906 MHz).
Because the counter is invariant,
  this rate is independent of the core's operating frequency.
All ratios reported elsewhere in the paper
  are computed directly from cycle counts and
  are therefore unaffected by this conversion.
The standard deviations are small relative to the means,
  which is what makes the normalized comparisons
  in \autoref{subsec:outputquality} meaningful.

\subsection{Bril Benchmarks}
{\catcode`\_=12 \begin{longtable}{lrrr}
\toprule
 & LLVM-O0 & EGGCC-DP-O0 & LLVM-O3-O0 \\
\midrule
\endfirsthead
\toprule
 & LLVM-O0 & EGGCC-DP-O0 & LLVM-O3-O0 \\
\midrule
\endhead
\midrule \multicolumn{4}{r}{\textit{continued on next page}} \\
\endfoot
\bottomrule
\endlastfoot
ackermann & 1.81 ms $\pm$ 0.03 & 1.2 ms $\pm$ 0.06 & 1.46 ms $\pm$ 0.02 \\
adj2csr & 9.8 ms $\pm$ 0.78 & 9.51 ms $\pm$ 0.75 & 8.72 ms $\pm$ 0.61 \\
adler32 & 155.67 $\mu$s $\pm$ 10.68 & 152.07 $\mu$s $\pm$ 16.3 & 151.18 $\mu$s $\pm$ 8.97 \\
armstrong & 0.73 $\mu$s $\pm$ 0.02 & 0.89 $\mu$s $\pm$ 1.04 & 0.41 $\mu$s $\pm$ 0.02 \\
binary-fmt & 0.82 $\mu$s $\pm$ 1.28 & 0.57 $\mu$s $\pm$ 0.05 & 0.57 $\mu$s $\pm$ 0.01 \\
binary-search & 25.94 $\mu$s $\pm$ 2.92 & 25.7 $\mu$s $\pm$ 2.3 & 25.93 $\mu$s $\pm$ 6.43 \\
birthday & 31.45 $\mu$s $\pm$ 3.62 & 32.33 $\mu$s $\pm$ 3.59 & 32.16 $\mu$s $\pm$ 3.85 \\
bitshift & 0.35 $\mu$s $\pm$ 0.05 & 0.25 $\mu$s $\pm$ 0.03 & 0.24 $\mu$s $\pm$ 0.02 \\
bitwise-ops & 20.66 ms $\pm$ 0.53 & 29.72 ms $\pm$ 4.76 & 20.34 ms $\pm$ 0.85 \\
bubblesort & 25.86 $\mu$s $\pm$ 3.87 & 25.4 $\mu$s $\pm$ 3.01 & 24.7 $\mu$s $\pm$ 4.99 \\
catalan & 21.28 ms $\pm$ 2.09 & 17.91 ms $\pm$ 1.06 & 18.2 ms $\pm$ 0.98 \\
check-primes & 4.23 $\mu$s $\pm$ 0.22 & 4.15 $\mu$s $\pm$ 2.31 & 3.8 $\mu$s $\pm$ 0.26 \\
cholesky\_decomp & 33.97 $\mu$s $\pm$ 4.12 & 32.4 $\mu$s $\pm$ 3.88 & 32.93 $\mu$s $\pm$ 4.29 \\
collatz & 3.76 $\mu$s $\pm$ 0.21 & 3.74 $\mu$s $\pm$ 1.38 & 1.91 $\mu$s $\pm$ 0.04 \\
conjugate-gradient & 2.48 ms $\pm$ 0.11 & 2.61 ms $\pm$ 0.08 & 2.15 ms $\pm$ 0.02 \\
cordic & 0.63 $\mu$s $\pm$ 0.02 & 0.38 $\mu$s $\pm$ 0.03 & 0.27 $\mu$s $\pm$ 0.02 \\
csrmv & 6.37 ms $\pm$ 0.88 & 6.56 ms $\pm$ 0.62 & 344.48 $\mu$s $\pm$ 9.31 \\
dead-branch & 0.29 $\mu$s $\pm$ 0.03 & 0.16 $\mu$s $\pm$ 0.03 & 0.03 $\mu$s $\pm$ 0.01 \\
digital-root & 0.26 $\mu$s $\pm$ 0.03 & 0.3 $\mu$s $\pm$ 0.03 & 0.24 $\mu$s $\pm$ 0.01 \\
dot-product & 25.74 $\mu$s $\pm$ 3.24 & 26.6 $\mu$s $\pm$ 4.41 & 25.38 $\mu$s $\pm$ 3.11 \\
euclid & 0.19 $\mu$s $\pm$ 0.02 & 0.14 $\mu$s $\pm$ 0.01 & 0.12 $\mu$s $\pm$ 0.01 \\
euler & 81.79 ms $\pm$ 5.13 & 63.45 ms $\pm$ 4.87 & 65.21 ms $\pm$ 4.77 \\
fact & 16.25 $\mu$s $\pm$ 6.74 & 8.98 $\mu$s $\pm$ 2.51 & 1.89 $\mu$s $\pm$ 0.0 \\
factors & 0.52 $\mu$s $\pm$ 0.03 & 0.55 $\mu$s $\pm$ 0.01 & 0.43 $\mu$s $\pm$ 0.01 \\
fib & 191.51 $\mu$s $\pm$ 8.79 & 195.86 $\mu$s $\pm$ 13.8 & 157.37 $\mu$s $\pm$ 7.94 \\
fitsinside & 0.06 $\mu$s $\pm$ 0.0 & 0.04 $\mu$s $\pm$ 0.01 & 0.04 $\mu$s $\pm$ 0.0 \\
fizz-buzz & 267.23 $\mu$s $\pm$ 16.71 & 260.29 $\mu$s $\pm$ 12.42 & 261.12 $\mu$s $\pm$ 7.56 \\
gcd & 0.43 $\mu$s $\pm$ 0.01 & 0.49 $\mu$s $\pm$ 0.04 & 0.23 $\mu$s $\pm$ 0.01 \\
hanoi & 0.31 $\mu$s $\pm$ 0.06 & 0.17 $\mu$s $\pm$ 0.02 & 0.27 $\mu$s $\pm$ 0.03 \\
is-decreasing & 0.25 $\mu$s $\pm$ 0.0 & 0.19 $\mu$s $\pm$ 0.01 & 0.21 $\mu$s $\pm$ 0.0 \\
lcm & 898.31 $\mu$s $\pm$ 6.34 & 901.36 $\mu$s $\pm$ 46.46 & 874.03 $\mu$s $\pm$ 4.02 \\
leibniz & 64.34 ms $\pm$ 4.54 & 62.79 ms $\pm$ 5.36 & 63.27 ms $\pm$ 4.97 \\
loopfact & 27.78 $\mu$s $\pm$ 3.19 & 27.54 $\mu$s $\pm$ 2.76 & 27.85 $\mu$s $\pm$ 3.48 \\
major-elm & 26.0 $\mu$s $\pm$ 3.43 & 25.64 $\mu$s $\pm$ 2.8 & 25.45 $\mu$s $\pm$ 3.41 \\
mandelbrot & 108.65 ms $\pm$ 4.35 & 96.53 ms $\pm$ 4.68 & 90.07 ms $\pm$ 3.28 \\
mat-inv & 27.47 $\mu$s $\pm$ 4.26 & 26.68 $\mu$s $\pm$ 3.26 & 25.55 $\mu$s $\pm$ 3.51 \\
mat-mul & 8.82 ms $\pm$ 0.78 & 7.68 ms $\pm$ 0.6 & 5.86 ms $\pm$ 0.37 \\
max-subarray & 25.69 $\mu$s $\pm$ 2.43 & 26.38 $\mu$s $\pm$ 2.99 & 25.22 $\mu$s $\pm$ 3.27 \\
mod\_inv & 0.78 $\mu$s $\pm$ 0.03 & 0.62 $\mu$s $\pm$ 1.24 & 0.08 $\mu$s $\pm$ 0.0 \\
n\_root & 0.77 $\mu$s $\pm$ 0.08 & 0.5 $\mu$s $\pm$ 0.05 & 1.43 $\mu$s $\pm$ 0.24 \\
newton & 0.37 $\mu$s $\pm$ 0.06 & 0.32 $\mu$s $\pm$ 0.03 & 0.32 $\mu$s $\pm$ 0.01 \\
norm & 26.53 $\mu$s $\pm$ 3.2 & 26.75 $\mu$s $\pm$ 3.22 & 25.25 $\mu$s $\pm$ 3.97 \\
orders & 6.57 $\mu$s $\pm$ 1.07 & 7.3 $\mu$s $\pm$ 2.66 & 6.42 $\mu$s $\pm$ 1.92 \\
palindrome & 1.08 $\mu$s $\pm$ 0.03 & 1.12 $\mu$s $\pm$ 1.02 & 0.69 $\mu$s $\pm$ 0.02 \\
pascals-row & 24.48 $\mu$s $\pm$ 2.61 & 24.77 $\mu$s $\pm$ 3.77 & 24.2 $\mu$s $\pm$ 1.46 \\
perfect & 7.33 $\mu$s $\pm$ 1.98 & 7.41 $\mu$s $\pm$ 2.57 & 6.83 $\mu$s $\pm$ 2.03 \\
primes-between & 10.78 ms $\pm$ 0.52 & 10.73 ms $\pm$ 0.56 & 5.24 ms $\pm$ 0.11 \\
pythagorean\_triple & 17.11 $\mu$s $\pm$ 2.48 & 19.72 $\mu$s $\pm$ 3.01 & 17.38 $\mu$s $\pm$ 1.54 \\
quadratic & 0.35 $\mu$s $\pm$ 0.02 & 0.23 $\mu$s $\pm$ 0.01 & 0.17 $\mu$s $\pm$ 0.0 \\
quickselect & 25.99 $\mu$s $\pm$ 2.63 & 25.84 $\mu$s $\pm$ 2.83 & 25.56 $\mu$s $\pm$ 3.36 \\
quicksort & 25.51 $\mu$s $\pm$ 3.22 & 25.3 $\mu$s $\pm$ 3.24 & 25.27 $\mu$s $\pm$ 3.75 \\
ray-sphere-intersection & 0.14 $\mu$s $\pm$ 0.02 & 0.04 $\mu$s $\pm$ 0.0 & 0.04 $\mu$s $\pm$ 0.0 \\
recfact & 239.69 $\mu$s $\pm$ 45.57 & 142.31 $\mu$s $\pm$ 19.39 & 18.71 $\mu$s $\pm$ 0.49 \\
relative-primes & 334.35 $\mu$s $\pm$ 8.68 & 292.1 $\mu$s $\pm$ 7.76 & 310.33 $\mu$s $\pm$ 8.75 \\
reverse & 0.19 $\mu$s $\pm$ 0.0 & 0.14 $\mu$s $\pm$ 0.0 & 0.12 $\mu$s $\pm$ 0.01 \\
riemann & 165.38 ms $\pm$ 3.54 & 76.71 ms $\pm$ 4.68 & 59.47 ms $\pm$ 4.68 \\
sqrt & 0.32 $\mu$s $\pm$ 0.09 & 0.31 $\mu$s $\pm$ 0.06 & 0.29 $\mu$s $\pm$ 0.05 \\
sum-bits & 0.64 $\mu$s $\pm$ 0.04 & 0.48 $\mu$s $\pm$ 1.52 & 0.31 $\mu$s $\pm$ 0.01 \\
sum-check & 10.8 $\mu$s $\pm$ 4.13 & 10.19 $\mu$s $\pm$ 2.2 & 19.62 $\mu$s $\pm$ 1.95 \\
sum-divisors & 8.77 $\mu$s $\pm$ 1.21 & 9.01 $\mu$s $\pm$ 2.32 & 16.61 $\mu$s $\pm$ 2.75 \\
sum-sq-diff & 19.92 $\mu$s $\pm$ 2.0 & 24.2 $\mu$s $\pm$ 3.35 & 19.89 $\mu$s $\pm$ 2.37 \\
totient & 0.6 $\mu$s $\pm$ 0.03 & 0.57 $\mu$s $\pm$ 0.01 & 0.51 $\mu$s $\pm$ 0.01 \\
up-arrow & 95.77 $\mu$s $\pm$ 9.79 & 95.24 $\mu$s $\pm$ 5.95 & 95.04 $\mu$s $\pm$ 4.67 \\
vsmul & 437.53 $\mu$s $\pm$ 9.73 & 447.41 $\mu$s $\pm$ 18.76 & 436.46 $\mu$s $\pm$ 8.9 \\
\end{longtable}
}

\subsection{PolyBench Benchmarks}
{\catcode`\_=12 \begin{longtable}{lrrr}
\toprule
 & LLVM-O0 & EGGCC-DP-O0 & LLVM-O3-O0 \\
\midrule
\endfirsthead
\toprule
 & LLVM-O0 & EGGCC-DP-O0 & LLVM-O3-O0 \\
\midrule
\endhead
\midrule \multicolumn{4}{r}{\textit{continued on next page}} \\
\endfoot
\bottomrule
\endlastfoot
2mm & 138.08 ms $\pm$ 3.6 & 61.4 ms $\pm$ 3.84 & 26.02 ms $\pm$ 3.41 \\
3mm & 218.46 ms $\pm$ 4.14 & 88.1 ms $\pm$ 3.61 & 30.52 ms $\pm$ 5.49 \\
adi & 2.62 ms $\pm$ 0.03 & 2.54 ms $\pm$ 0.13 & 2.49 ms $\pm$ 0.02 \\
atax & 8.09 ms $\pm$ 0.75 & 4.58 ms $\pm$ 0.35 & 2.62 ms $\pm$ 0.48 \\
bicg & 15.48 ms $\pm$ 0.56 & 10.62 ms $\pm$ 0.59 & 8.21 ms $\pm$ 0.62 \\
cholesky & 623.12 ms $\pm$ 3.03 & 376.0 ms $\pm$ 4.09 & 66.91 ms $\pm$ 3.94 \\
correlation & 3.25 ms $\pm$ 0.04 & 3.09 ms $\pm$ 0.17 & 2.47 ms $\pm$ 0.01 \\
covariance & 2.64 ms $\pm$ 0.03 & 2.47 ms $\pm$ 0.07 & 821.08 $\mu$s $\pm$ 8.16 \\
deriche & 2.88 ms $\pm$ 0.17 & 2.68 ms $\pm$ 0.18 & 2.45 ms $\pm$ 0.29 \\
doitgen & 60.92 ms $\pm$ 4.46 & 29.28 ms $\pm$ 5.35 & 11.28 ms $\pm$ 1.46 \\
durbin & 1.93 ms $\pm$ 0.11 & 389.14 $\mu$s $\pm$ 39.0 & 279.88 $\mu$s $\pm$ 40.76 \\
fdtd-2d & 2.28 ms $\pm$ 0.23 & 2.21 ms $\pm$ 0.04 & 1.97 ms $\pm$ 0.02 \\
floyd-warshall & 20.1 ms $\pm$ 1.54 & 13.12 ms $\pm$ 0.9 & 9.92 ms $\pm$ 0.67 \\
gemm & 100.53 ms $\pm$ 4.57 & 26.13 ms $\pm$ 4.78 & 11.74 ms $\pm$ 0.71 \\
gemver & 17.63 ms $\pm$ 2.99 & 11.46 ms $\pm$ 1.32 & 8.07 ms $\pm$ 0.79 \\
gesummv & 9.29 ms $\pm$ 1.07 & 7.02 ms $\pm$ 0.69 & 5.53 ms $\pm$ 0.48 \\
gramschmidt & 506.83 $\mu$s $\pm$ 40.61 & 185.65 $\mu$s $\pm$ 15.11 & 100.93 $\mu$s $\pm$ 22.22 \\
heat-3d & 6.68 ms $\pm$ 0.56 & 4.18 ms $\pm$ 0.18 & 4.06 ms $\pm$ 0.03 \\
jacobi-1d & 60.91 $\mu$s $\pm$ 6.43 & 53.47 $\mu$s $\pm$ 5.76 & 42.9 $\mu$s $\pm$ 4.94 \\
jacobi-2d & 3.22 ms $\pm$ 0.14 & 2.73 ms $\pm$ 0.19 & 1.17 ms $\pm$ 0.07 \\
lu & 690.21 ms $\pm$ 3.2 & 370.43 ms $\pm$ 3.33 & 75.92 ms $\pm$ 4.55 \\
ludcmp & 705.09 ms $\pm$ 3.28 & 325.07 ms $\pm$ 3.49 & 74.62 ms $\pm$ 4.53 \\
mvt & 15.23 ms $\pm$ 0.54 & 11.2 ms $\pm$ 0.65 & 7.82 ms $\pm$ 0.64 \\
nussinov & 6.38 ms $\pm$ 0.4 & 5.25 ms $\pm$ 0.12 & 2.79 ms $\pm$ 0.03 \\
seidel-2d & 13.47 ms $\pm$ 0.5 & 12.93 ms $\pm$ 0.6 & 10.51 ms $\pm$ 0.59 \\
symm & 77.65 ms $\pm$ 3.93 & 17.48 ms $\pm$ 1.69 & 18.35 ms $\pm$ 0.92 \\
syr2k & 92.07 ms $\pm$ 3.83 & 23.03 ms $\pm$ 4.77 & 18.81 ms $\pm$ 1.12 \\
syrk & 60.25 ms $\pm$ 4.68 & 20.03 ms $\pm$ 1.01 & 10.75 ms $\pm$ 1.36 \\
trisolv & 3.39 ms $\pm$ 0.14 & 2.89 ms $\pm$ 0.09 & 2.13 ms $\pm$ 0.05 \\
trmm & 49.73 ms $\pm$ 4.66 & 24.0 ms $\pm$ 4.98 & 10.75 ms $\pm$ 0.63 \\
\end{longtable}
}

\section{Regionalized \Egraphs}
\label{app:regionalized}

\autoref{fig:regionalized-egraph} illustrates the decomposition described in
  \autoref{sec:controlflows}.
The example applies a rule that swaps the two branches of a conditional,
  relating terms in different regions,
  and then shows the three regionalized \egraphs
  that the resulting \egraph decomposes into for extraction.

\begin{figure}[h]
  \centering
  \includegraphics[width=\textwidth]{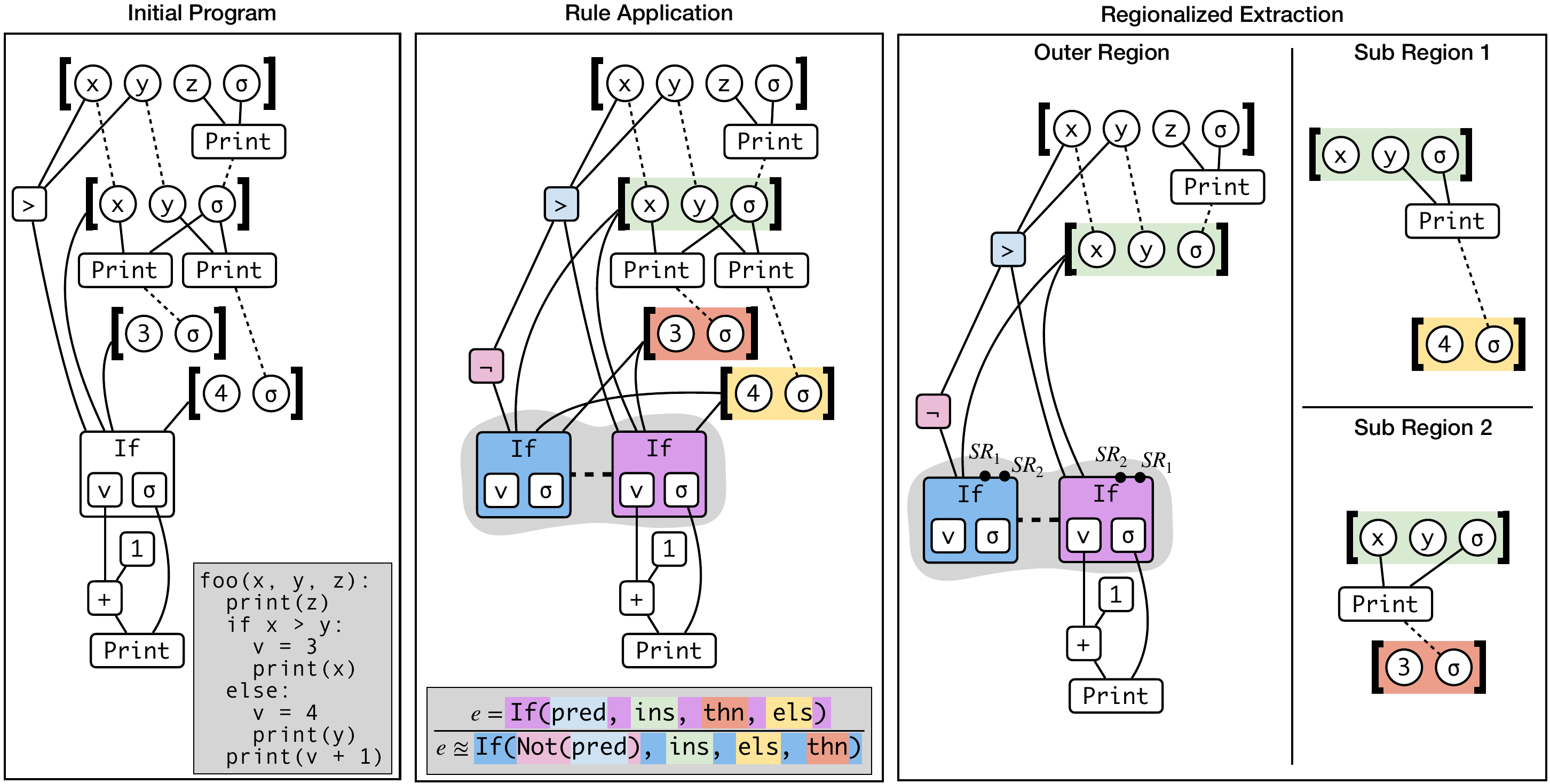}
  \caption{An example program and a rule application that switches the two
    branches of a conditional.
    The resulting \egraph contains three regionalized \egraphs, shown in the
    third pane, each of which forms a separate \effectsafe extraction problem.}
  \Description{An example program, a rule application switching the branches of
    a conditional, and the three regionalized e-graphs the result decomposes
    into.}
  \label{fig:regionalized-egraph}
\end{figure}

\section{Rules for the Lowbit Optimization}
\label{app:lowbit}

\autoref{sec:Fenwick} implements the \texttt{lowbit} optimization
  with the three \egglog rules given below as inference rules.
Each is straightforward syntactic pattern matching over \eggcc's IR,
  in which a source-level \texttt{while} loop becomes an
  \texttt{If} whose then-branch holds a \texttt{DoWhile};
  the loop's first output is its continuation predicate,
  written here as $[\mathit{pred}] \mathbin{+\!\!+} \mathit{outs}$.
Throughout, $p$ is the position of the halved value $n$
  in the \texttt{If}'s input tuple,
  $i$ is its position in the loop's tuples ---
  the index \texttt{NTZIterations} records ---
  and $j$ is the position of the accumulator.
Within the loop body, $\argt$ is the body's own argument tuple,
  so $\argt[i]$ is the halved value on the current iteration,
  as distinct from $\mathit{lp\_inputs}[i]$, the value entering the loop.

The first rule records that a predicate $e$ tests whether $x$ is even.
The relation it builds, \texttt{IsIsEven}, reads
  ``$e$ is an is-even test for $x$''.

\[
  \frac{
    e = \texttt{Eq}(x,\ \texttt{Mul}(\texttt{Div}(x,\ 2),\ 2))
  }{
    \texttt{IsIsEven}(e,\ x)
  }
\]

The second rule identifies a loop whose iteration count is the number of
  trailing zeros of an input $n$: it tests whether $n$ is even and halves $n$
  on each iteration, returning $n$ unchanged when $n$ is odd.

\[
  \frac{
    \begin{array}{c}
      e = \texttt{If}(\mathit{cond},\ \mathit{inputs},\
        \texttt{DoWhile}(\mathit{lp\_inputs},\
        [\mathit{pred}] \mathbin{+\!\!+} \mathit{outs}),\ \mathit{oddbr}) \\
      n = \mathit{inputs}[p] = \mathit{lp\_inputs}[i] = \mathit{oddbr}[i]
      \qquad
      \mathit{outs}[i] = \texttt{Div}(\argt[i],\ 2) \\
      \texttt{IsIsEven}(\mathit{cond},\ n)
      \qquad
      \texttt{IsIsEven}(\mathit{pred},\ \texttt{Div}(\argt[i],\ 2))
    \end{array}
  }{
    \texttt{NTZIterations}(e,\ n,\ i)
  }
\]

The third rule fires when such a loop also carries an accumulator that starts
  at $1$ and doubles every iteration;
  that accumulator therefore holds $\texttt{lowbit}(n)$ on exit.
It equates that output with the one-line form and, since the loop then has no
  remaining use, equates the halved value with $n / \texttt{lowbit}(n)$.

\[
  \frac{
    \begin{array}{c}
      e = \texttt{If}(\mathit{cond},\ \mathit{inputs},\
        \texttt{DoWhile}(\mathit{lp\_inputs},\
        [\mathit{pred}] \mathbin{+\!\!+} \mathit{outs}),\ \mathit{oddbr})
      \qquad
      \texttt{NTZIterations}(e,\ n,\ i) \\
      \mathit{lp\_inputs}[j] = \mathit{oddbr}[j] = 1
      \qquad
      \mathit{outs}[j] = \texttt{Mul}(\argt[j],\ 2)
    \end{array}
  }{
    e[j] \approxeq_{\mathcal{G}} \texttt{Bitand}(n,\ -n)
    \qquad
    e[i] \approxeq_{\mathcal{G}} \texttt{Div}(n,\ \texttt{Bitand}(n,\ -n))
  }
\]

Nothing forces the optimization to be split this way; a single rule would
  express it just as well.
We use three because the intermediate relations make each step readable,
  and because \texttt{IsIsEven} is reusable by other rules.



\end{document}